\documentclass[10pt]{article}

\usepackage[
  paperwidth=535.408pt,
  paperheight=697.209pt,
  left=53.5pt,
  right=53.5pt,
  top=48pt,
  bottom=50pt,
  headheight=14pt,
  headsep=12pt,
  footskip=24pt
]{geometry}

\usepackage[T1]{fontenc}
\usepackage[utf8]{inputenc}
\usepackage{microtype}

\usepackage{amsmath}
\usepackage{amsthm}
\usepackage{mathtools}
\usepackage{bm}

\usepackage{newtxtext}
\usepackage[amssymbols]{newtxmath}

\usepackage{fancyhdr}
\fancypagestyle{plain}{
  \fancyhf{}
  \fancyfoot[C]{\footnotesize\thepage}
  
}

\newtheoremstyle{thmstyleone}
  {6pt}{6pt}{\itshape}{}{\bfseries}{.}{0.5em}{}
\newtheoremstyle{thmstyletwo}
  {6pt}{6pt}{\normalfont}{}{\bfseries}{.}{0.5em}{}
\newtheoremstyle{thmstylethree}
  {6pt}{6pt}{\normalfont}{}{\bfseries}{.}{0.5em}{}

\theoremstyle{thmstyleone}
\newtheorem{theorem}{Theorem}

\newtheorem{lemma}[theorem]{Lemma}
\newtheorem{corollary}[theorem]{Corollary}

\theoremstyle{thmstylethree}

\theoremstyle{thmstyletwo}

\usepackage{titlesec}

\titleformat{\section}
  {\normalfont\large\bfseries}
  {\thesection.}{0.55em}{}

\titleformat{\subsection}
  {\normalfont\normalsize\bfseries}
  {\thesubsection.}{0.55em}{}

\titleformat{\subsubsection}
  {\normalfont\normalsize\bfseries}
  {\thesubsubsection.}{0.55em}{}

\titlespacing*{\section}{0pt}{14pt plus 3pt minus 2pt}{5pt}
\titlespacing*{\subsection}{0pt}{10pt plus 2pt minus 1pt}{3pt}
\titlespacing*{\subsubsection}{0pt}{8pt plus 2pt minus 1pt}{2pt}

\usepackage{graphicx}
\graphicspath{{Fig/}}

\usepackage{booktabs,multirow,float}
\usepackage{caption}
\usepackage{adjustbox}
\usepackage{array}
\usepackage{etoolbox}
\usepackage{placeins}

\AtBeginEnvironment{tabular}{\footnotesize}

\usepackage{algorithm}
\usepackage{algpseudocode}

\usepackage[round,authoryear]{natbib}

\usepackage{xcolor}
\definecolor{C0}{HTML}{3182ce}

\usepackage{hyperref}
\hypersetup{
  breaklinks=true,
  colorlinks=true,
  linkcolor=C0,
  citecolor=C0,
  urlcolor=C0,
  bookmarks=true,
  bookmarksopen=true,
  bookmarksnumbered=true
}
\usepackage{bookmark}

\usepackage{tikz}
\usepackage{istgame}

\usepackage{siunitx}

\usepackage[most]{tcolorbox}
\newtcolorbox{mybox}[1][]{
    colframe=blue!70!black,
    colback=blue!2,
    arc=4mm,
    boxrule=0.4mm,
    left=2mm, right=2mm,
    top=0mm, bottom=0mm,
    enhanced,
    #1
}

\def\R{\mathbb{R}}

\def\E{\mathbb{E}}

\def\1{\bm{1}}

\DeclareMathOperator*{\argmin}{argmin}

\def\btheta{{\boldsymbol{\theta}}}
\def\bgamma{\boldsymbol{\gamma}}

\def\bX{\boldsymbol{X}}

\DeclarePairedDelimiterX{\ent}[2]{(}{)}{#1\;\delimsize\|\;#2}

\makeatletter
\renewcommand{\maketitle}{%
  \thispagestyle{plain}
  \begingroup
  \raggedright
  {\fontsize{18}{22}\selectfont\bfseries \@title \par}
  \vspace{0.8em}
  {\fontsize{11}{14}\selectfont\bfseries \@author \par}
  \vspace{0.9em}
  \endgroup
}
\makeatother

\newenvironment{jrssabstract}
  {\par\noindent\textbf{Abstract}\par\vspace{3pt}\noindent\small}
  {\par\vspace{5pt}}

\newcommand{\keywords}[1]{%
  \par\noindent{\footnotesize\textbf{Key words:} #1}\par\vspace{1em}
}

\definecolor{EBred}{RGB}{215,40,50}
\definecolor{PLblue}{RGB}{35,90,220}
\definecolor{FBviolet}{RGB}{150,45,175}

\title{Empowering Bayes with sources of uncertainty:\\
from confidence modelling to sparse estimation}

\author{
Rafael Mouallem Rosa\textsuperscript{1,*},
Julyan Arbel\textsuperscript{1}
and Hien Duy Nguyen\textsuperscript{2,3}
\\[0.45em]
{\footnotesize
\textsuperscript{1}Univ. Grenoble Alpes, Inria, CNRS, Grenoble INP, LJK, 38000 Grenoble, France, 
\textsuperscript{2}School of Computing, Engineering and Mathematical Sciences, La Trobe University, Melbourne, Australia,
and
\textsuperscript{3}Institute of Mathematics for Industry, Kyushu University, Fukuoka, Japan
}
\\[0.35em]
{\footnotesize
\textsuperscript{*}Address for correspondence:
\href{mailto:rafael.mouallem-rosa@inria.fr}{rafael.mouallem-rosa@inria.fr}
}
}

\date{}

\begin{document}

\maketitle

\begin{jrssabstract}
We introduce sources of uncertainty (SoU) inference, a general framework that extends Bayesian inference by allowing researchers to explicitly encode heterogeneous confidence across model components.
This heterogeneity allows the probabilistic treatment to vary selectively across the model, controlling both the uncertainty retained within each component and its propagation to the rest of the model.
Empirical Bayes and profiling-type inference emerge as limiting configurations of this selective treatment. 
Because confidence acts through a mechanism distinct from prior informativeness,
the two can be controlled separately, providing an additional degree of freedom for Bayesian model design. 
Building on this selective treatment, 
we develop a general approach for inducing sparsity in statistical models and
illustrate its use in linear and logistic regression, as well as in Bayesian
neural networks, where it enables self-pruning.
\end{jrssabstract}


\section{Introduction}\label{sec:introduction}

Within the Bayesian paradigm, when faced with uncertainty, one forms beliefs about unknowns, and these beliefs are expressed through subjective probability distributions. Upon observing new evidence, these beliefs are updated to a posterior distribution, i.e.,
    \begin{mybox}
    \begin{center}
        \textbf{Bayesian inference:} 
    $\underbrace{p(\bX \mid \theta_1, \theta_2)}_{\text{likelihood}}, \quad 
    \underbrace{\pi_2(\theta_2 \mid \theta_1), \, \pi_1(\theta_1)}_{\text{prior}} 
    \xrightarrow{\text{update rule}} 
    \underbrace{p(\btheta \mid \bX)}_{\text{posterior}}$
    \end{center}
    \end{mybox}
Here, $\bX$ represents the observed data, and $\btheta = (\theta_1,\theta_2)$ denotes the unknown parameters, assumed here bivariate for the sake of simplicity, about which the agent holds beliefs. Bayesian statistics relies on two main inputs specified by the researcher: a prior distribution $\pi$, and a likelihood function $p$. Bayesian inference is then the process of updating the beliefs via an updating rule, which yields the posterior distribution.\footnote{Although the likelihood function can be generalised to an arbitrary function, as we do in the next sections, we retain the standard formulation here to maintain a connection with traditional Bayesian models.}

As reasonable as it is to assume that, under uncertainty, an agent can engage in probabilistic reasoning about specific unknowns, it is equally appealing to conceive that their ability to do so may vary across different sources of uncertainty, that is, across distinct components of the joint model corresponding to different subsets of parameters. This leads to different levels of confidence in the beliefs that are reported. For instance, an agent might have well-founded beliefs about one parameter distribution but only vague intuitions about another.

Our framework extends the Bayesian paradigm to account for this aspect by introducing an additional input parameter $\bgamma = (\gamma_1, \gamma_2)$, which captures the agent's perceived ability to reason probabilistically about each source of uncertainty. The belief updating process is thus generalised to incorporate $\bgamma$.
    \begin{mybox}
    \begin{center}
        \textbf{SoU inference:} 
    $\underbrace{p(\bX \mid \btheta)}_{\text{likelihood}}, \, 
    \underbrace{(\pi_2(\theta_2 \mid \theta_1), {\color{FBviolet}\gamma_2}), \, (\pi_1(\theta_1), {\color{FBviolet}\gamma_1})}_{\text{prior and {\color{FBviolet}confidence}}} 
    \xrightarrow{\color{FBviolet}\text{update rule}} 
    \underbrace{p_{\color{FBviolet}\gamma}(\btheta \mid \bX)}_{\text{posterior}}$
    \end{center}
    \end{mybox}
In this formulation, the conventional Bayesian ingredients are retained, but each prior $\pi_i$ is now accompanied by a confidence parameter $\gamma_i$.

The inspiration for introducing a model that explicitly accounts for confidence in priors stems from decision theory. In that field, an agent’s confidence in their own subjective probability assessments is linked to ambiguity attitudes. For instance, ambiguity aversion reflects a preference for known over unknown probabilities. 
In this sense, confidence can be interpreted as the degree to which an agent considers their own subjective probabilities less reliable than objective ones.
Our framework can thus be viewed as the statistical analogue of a decision-theoretic model that distinguishes between different sources of uncertainty, reflecting varying degrees of confidence in probabilistic reasoning over unknowns.

A central conceptual distinction in our framework is that between the
\emph{informativeness} of a prior and the \emph{confidence} assigned to the
source of uncertainty it represents. Both affect how prior information and data
contribute to the posterior, but in different ways. Prior informativeness
determines how strongly the posterior is pulled towards the prior, 
whereas confidence governs the probabilistic treatment of the corresponding
source and how the associated uncertainty propagates to other model components.

Increasing prior informativeness couples two effects: it pulls the posterior
more strongly \emph{towards the prior} while increasing its concentration. 
By contrast, decreasing confidence also increases posterior concentration while
attenuating the direct influence of the corresponding prior component, so that
the posterior concentrates around solutions increasingly driven by the empirical
loss, producing a movement \emph{towards the data}. Through this source-specific
confidence mechanism, the researcher gains a new handle to control 
how uncertainty is expressed and how
information from data and prior contributes to the posterior. Conceptually, the
confidence parameters $\gamma_i$ act as a complementary modelling tool, an
extra piece in the Bayesian toolkit, that can be tuned to control how the
posterior is shaped.

At low confidence, this increasing concentration makes the treatment of the
corresponding source increasingly optimisation-like. 
Because confidence is assigned source by source, this behaviour can occur
selectively across the model, with some components driven towards optimisation
while others retain a probabilistic treatment. 
This selective treatment offers a common perspective on familiar inferential
procedures, with empirical Bayes and profile-likelihood inference arising as
distinct limiting configurations of heterogeneous confidence.

Building on this selective treatment, 
we develop a general approach to induce sparsity in statistical models within the proposed framework. We proceed by endowing the model's parameters with a global-local structure, where a global component forces all parameters towards zero and local components allow individual parameters to escape it. In the standard setting, this structure promotes shrinkage of irrelevant parameters, meaning their posteriors are pulled towards zero but still assign non-negligible probability to small non-zero values. By setting the confidence associated with local parameters to a small value, we obtain posteriors that shrink irrelevant parameters by sharply concentrating their mass near zero, yielding sparse representations. We refer to this specific configuration, comprising a global-local prior with near-zero confidence on the local scales, as the \emph{sources of uncertainty sparse global-local} (SoU-SGL) estimator.

We illustrate this sparsity mechanism across three settings of increasing complexity: linear regression, logistic regression, and Bayesian neural networks (BNN). In the BNN application, the SoU-SGL formulation enables automatic self-pruning, allowing the network to eliminate redundant units and adapt its effective architecture to the data.

Code implementing the proposed methods and reproducing the numerical
experiments is available at \url{https://github.com/rafmoua/sou-paper}. 
The rest of the paper is organised as follows. Section~\ref{sec:decision-bayes} draws the decision-theoretic connections that motivate our framework. 
Section~\ref{sec:sou-bayes} formally introduces the framework and examines its
main features.
Section~\ref{sec:sou-sgl} develops the SoU-SGL estimator and applies it to the three settings described above. Section~\ref{sec:discussion} concludes with remarks and future research directions.

\section{Decision theory and Bayesian statistics}\label{sec:decision-bayes}

Since \cite{keynes1921treatise} and \cite{knight1921risk}, the distinction between risk and uncertainty has been well established. Risk refers to situations where probabilities are known, while uncertainty arises when probabilities are not. Probabilities, in the context of risk, are objectively understood, corresponding to the long-run relative frequencies of events, as described by the Law of Large Numbers. Under uncertainty, however, probabilities are treated as subjective degrees of belief, reflecting personal assessments.

\cite{savage-54} introduced the subjective expected utility (SEU) model, which became the foundational framework for modelling choices under uncertainty in decision theory. Starting with a preference
relation as primitive, Savage imposed a set of elegant axioms and derived the existence of
both a subjective probability distribution and a utility function that together govern an agent’s
choices. Savage's representation theorem also played a key role in the early development of Bayesian statistics by providing a formal justification for a subjective approach to probability.

Most models in economic theory continue to be built within the Bayesian paradigm, where agents' beliefs are represented probabilistically through an SEU model. However, these models do not capture the degree of confidence attached to these probabilistic assignments. The ambiguity aversion literature emerged from the observation that this limitation prevents models from expressing certain reasonable and compelling preferences, such as those highlighted by the Ellsberg paradox.\footnote{ Ellsberg paradox:
There is an urn containing 90 balls: 30 are red, while the remaining 60 are either black or yellow in unknown proportions. When asked to choose between the following gambles: \textit{\textbf{A:} Receive 100 dollars if a \textbf{red} ball is drawn; \textbf{B:} Receive 100 dollars if a \textbf{black} ball is drawn},
people tend to prefer \textbf{A} over \textbf{B}. However, when given a second choice:
\textit{\textbf{C:} Receive 100 dollars if a \textbf{red or yellow} ball is drawn; \textbf{D:} Receive 100 dollars if a \textbf{black or yellow} ball is drawn},
people tend to prefer \textbf{D} over \textbf{C}. If the person can be represented by an SEU model, choosing \textbf{A} over \textbf{B} reveals a belief that $p(\textbf{yellow}) > p(\textbf{black})$. If this is the case, the person should also prefer \textbf{C} over \textbf{D}, leading to an inconsistency in preferences.} 
The behavioural trait missing from these models is that, although people may think probabilistically when facing uncertainty, they do not treat subjective probabilities the same way as objective ones, tending to prefer situations with known probabilities over those with unknown probabilities \citep[for a comprehensive discussion of ambiguity aversion literature, see][]{gilboa2016ambiguity}.

To formally capture ambiguity aversion, decision-theoretic models have been developed to allow for a richer representation of preferences. In decision theory, an agent chooses among acts \( f \), which are functions mapping states of nature \( \Omega \) to outcomes in \( X \). Since there is no objective probability measure on \( \Omega \), the agent faces fundamental uncertainty. The agent’s preferences over acts are represented by a functional \( V(f) \), which ranks acts based on a criterion derived from axioms imposed on preferences \( \succeq \) over \( X^\Omega \). 
\cite{strzalecki2011axiomatic} axiomatises the following criterion, known as \textit{multiplier preferences} (MP),  originally proposed by \cite{hansen2001robust}:
\begin{equation}\label{eq:MP_preferences}
    V(f) = \min_p \left\{ \mathbb{E}_p u(f) + \gamma \text{KL}(p \| q ) \right\},
\end{equation}
where \( u \) is a utility function, \( \gamma \in (0,\infty]\) and KL stands for the Kullback--Leibler divergence.

This representation models situations where the decision maker does not have enough information to formulate a single probabilistic model with full confidence. Instead, the agent has a best guess \( q \) and a confidence level \( \gamma \) that determines how much deviation from \( q \) is considered acceptable.

The term \( q \) can be interpreted as the agent's prior, while \( p \) serves as an intermediary object: a posterior-like distribution that the agent computes in order to evaluate the act \( f \). \

Classically, Bayesian inference consists of specifying a probabilistic model \( p_\btheta \) for a random variable, parameterised by a vector of unknown quantities \( \btheta \in \Theta \subseteq \R^d \), where we assume $\Theta$ to be a Euclidean space. Here, uncertainty lies in the parameter space. The Bayesian procedure involves specifying a prior for \( \btheta \) and updating this prior after observing realisations of the random variable.

\cite{bissiri2016general} propose the following expression to compute a posterior distribution:
\begin{equation}\label{eq:General_Posterior}
    p^* = \argmin_{p} \left\{ \mathbb{E}_p L (\btheta, X) + \gamma \text{KL}(p \| q ) \right\}.
\end{equation}

This framework is known as  \textit{generalised Bayesian inference}. The term ``generalised'' reflects the flexibility of updating the prior \( q \) using any loss function \( L \), where the standard Bayesian posterior is recovered when \( L \) corresponds to the usual negative log-likelihood and \(\gamma=1\). \cite{bissiri2016general} introduce the parameter \( \gamma \) as a calibration factor to account for potential differences in the scale of loss functions. Although not originally motivated in this way, \(\gamma\) captures the researcher's confidence in her prior \(q\).

In decision theory, the agent faces multiple acts \( f \), and the probability minimising the objective of Equation~\eqref{eq:MP_preferences}, say \( p_f \), is used to evaluate and rank these acts. In contrast, in the generalised Bayes setting, \( L \) and \( X \) are fixed, and the minimising probability is itself the primary object of interest.

The parallels between the decision-theoretic and statistical formulations are evident: Both formulations involve optimising a quantity composed of an expected term and a divergence term, and both introduce a confidence-like parameter $\gamma$ governing the degree of trust in a reference distribution. The difference lies in the interpretation of the term inside the expectation: in the decision-theoretic case, it reflects a cautious behavioural assumption, where the agent ranks acts by their minimum expected utility; whereas in the generalised Bayes setting, it is a loss function expressing the desire to fit the data as closely as possible.

Both representations \eqref{eq:MP_preferences} and \eqref{eq:General_Posterior} are special cases of broader classes of models. The decision-theoretic formulation falls within the framework of \textit{Variational Preferences} (VP) of \citet{maccheroni2006ambiguity}, while the Bayesian formulation is a particular instance of \textit{Generalised Variational Inference} (GVI) \citep[see][]{knoblauch2022optimization}. These general models extend the formulations by abstracting the Kullback--Leibler divergence into a general cost function \( c(p) \) and a statistical divergence \( \text{D}(p \| q) \) respectively, allowing for greater flexibility in capturing the agent's attitude towards uncertainty and the model's prior misspecification.

\subsection{Sources of Uncertainty}

As reasonable as it is to conceive that the agent has a degree of confidence in her beliefs, it is also natural to consider that this degree of confidence varies according to the different sources of uncertainty she faces. Empirical evidence supports this idea, showing that individuals exhibit different attitudes towards distinct sources of uncertainty, preferring those in which they feel more competent or knowledgeable \citep{abdellaoui2011rich,de2010competence,fox2002ambiguity,heath1991preference,keppe1995judged,kilka2001determines,tversky1995weighing}.

We introduce a natural decomposition of uncertainty into informational stages, each represented by a conditional distribution. This gives a formal grounding to what we call \textit{sources of uncertainty}. To make this precise, we represent the state space through an information structure that can be viewed as an event tree, see Figure~\ref{fig:tree}.

\begin{figure}[h]
    \centering

    \input{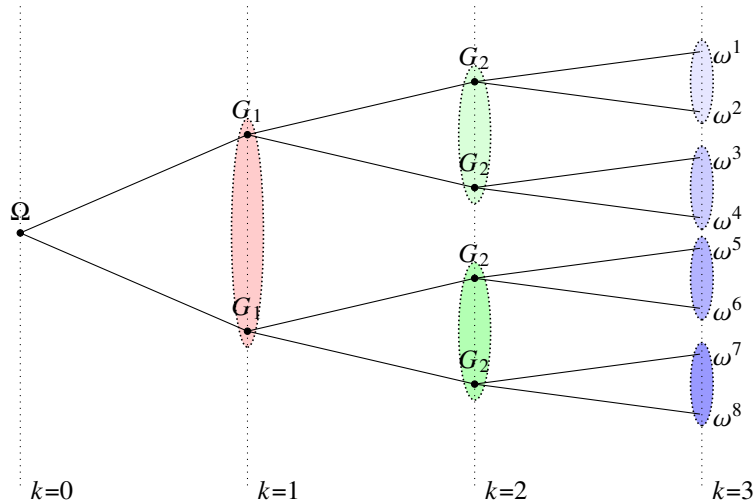}
\caption{\footnotesize{Event tree representation of sources of uncertainty. Nodes represent informational stages $k$ and the induced partitions $\mathcal{G}_k$; each stage corresponds to a distinct conditional component that can be assigned its own confidence parameter.}}
\label{fig:tree}
\end{figure}

Let $\mathcal{K} := \{0,1,\dots,K\}$ index the stages of observation, and let $\Omega_k$ denote the observation space at stage $k$. The full space of states is $\Omega := \prod_{k=1}^K \Omega_k$, where an element $\omega \in \Omega$ represents a complete observation path, $\omega = (\omega_1, \dots, \omega_K)$.

We define a sequence of partitions $\{\mathcal{G}_k\}_{k\in \mathcal{K}}$ of $\Omega$, where $G_k(\omega) \in \mathcal{G}_k$ denotes the element containing $\omega$. Intuitively, $\mathcal{G}_k$ represents the information available after observing the first $k$ components of the path. For $G \in \mathcal{G}_k$, define $p_G(\omega):=p(\omega)/p(G)$ for $\omega \in G$, and 0 otherwise.

The uncertainty represented at each node can be interpreted as a different source of uncertainty. In Bayesian statistics, an analogous procedure is to fix a factorisation of the parameter space, which similarly defines distinct sources of uncertainty and the order in which information is incorporated.

This event-tree structure has a natural dynamic interpretation and provides a foundation for dynamic versions of standard static decision-theoretic models \citep[e.g.,][]{epstein2003recursive,Maccheroni2006Dynamic,klibanoff2009recursive}. In this setup, one defines local preferences $\{\succeq_{k,\omega}\}$ at each node $(k,\omega)$ and imposes the axioms of the static model within each node. From the perspective of the initial node $\succeq_{0,\Omega}$, the entire structure can be viewed as a static problem, with the partition of the space introducing distinct sources of uncertainty.

From the axiomatisation of \citet{Maccheroni2006Dynamic}, we can recover the following dynamic version of the MP model~(\ref{eq:MP_preferences}):
\begin{equation}\label{eq:dynamic_MP}
V_{0,\Omega}(f) = \min_p \left\{ \mathbb{E}_p u(f) + \sum_{k=1}^K \gamma_k \sum_{G \in \mathcal{G}_k} p(G) \text{KL}(p_G \| q_G) \right\},
\end{equation}
where each $\gamma_k$ governs the degree of confidence in the reference distribution $q$ at stage $k$.

Building on this decision-theoretic foundation, we now introduce our statistical framework, \textit{Bayesian inference with sources of uncertainty}, as a counterpart of \eqref{eq:dynamic_MP}. We translate the underlying information structure into a blockwise factorisation of the parameter space, thereby defining \emph{sources of uncertainty} and allowing confidence parameters $\{\gamma_k\}$ to be specified block by block, enabling heterogeneous confidence across components.

\section{Bayesian inference with sources of uncertainty}\label{sec:sou-bayes}

Let $(\mathcal X,\mathcal B)$ be a standard Borel space, and let $\bX_{1:n} = (X_1, \ldots,X_n)$ be an i.i.d.\ sample from an unknown distribution $P_0$ on $(\mathcal X,\mathcal B)$. Let $\Theta = \Theta_1 \times \cdots \times \Theta_D \subseteq\mathbb R^D$ be the parameter space, and $\btheta = (\theta_1, \dots, \theta_D) \in \Theta$, and assume all considered distributions on $\Theta$ admit densities. 
For $K\leq D$, let $\{B_k\}_{k=1}^K$ be a partition of $\{1, \dots, D\}$ and write
$\btheta_k:=(\theta_j)_{j\in B_k}$ and $\btheta_{<k}:=(\theta_j)_{j\in \cup_{i<k}B_i}$. We define the blockwise factorisation of a density as $p(\btheta)=p(\btheta_1)\prod_{k=2}^K p\!\left(\btheta_k|\btheta_{<k}\right)$.
This block structure allows us to represent different sources of uncertainty and to assign distinct confidence parameters to each component of the model. The choice of the block partition and its ordering is therefore part of the model specification.

Let $L_n:\Theta\to\mathbb{R}$ denote the empirical loss. In this paper we focus on the additive form $L_n(\btheta)=\sum_{i=1}^n \ell(\btheta,X_i)$, where $\ell(\btheta,X)$ is a per-observation loss linking parameters to observations. We define the SoU posterior $q^*$ as a minimiser\footnote{\label{note1}Computational  
procedures for solving the optimisation problem are described in
Section~\ref{app:estimation} of the Supplementary Material.}
\begin{equation} \label{eq:qstar}
q^* \in \argmin_{q} J\left(q\right),
\end{equation}
where \(J\) is defined over distributions \(q\) on \(\Theta\) by
\begin{equation}
\label{eq:KB}
\begin{aligned}
J(q)
&=
\frac{1}{n}\mathbb{E}_{q}\!\left[L_n(\btheta)\right]
+
\frac{\gamma_1}{n}
\operatorname{KL}\!\left(
q(\btheta_1)\,\|\,\pi(\btheta_1)
\right)
\\
&\quad+
\sum_{k=2}^{K}\frac{\gamma_k}{n}
\mathbb{E}_{q(\btheta_{<k})}\!\left[
\operatorname{KL}\!\left(
q(\btheta_k\mid\btheta_{<k})
\,\|\,
\pi(\btheta_k\mid\btheta_{<k})
\right)
\right].
\end{aligned}
\end{equation}
Here $\pi(\cdot)$ is a prior over $\btheta$ that admits the blockwise
factorisation induced by the partition $\{B_k\}_{k=1}^K$, and $\gamma_k>0$
quantifies the researcher's confidence in the corresponding blockwise prior
component. Blocks are characterised by the confidence assigned to them. By the
chain rule for the KL divergence, any further hierarchical factorisation within
a block recombines into the original blockwise KL term, as long as the same
confidence parameter is assigned to all its components.

In this formulation, the posterior is computed under a trade-off between fitting the data, as captured by the expected loss term, and not deviating too much from the prior, as captured by the KL terms.

The standard Bayesian posterior corresponds to the case where the loss is the negative log likelihood, $L_n(\btheta)=-\sum_{i=1}^n \log f(X_i|\btheta)$, and $\gamma_k = 1$ for all $k$.
In this case the KL terms collapse to $\frac{1}{n} \text{KL} ( q \| \pi )$, that is, a single KL term scaled by $1/n$. The trade-off is then governed entirely by the sample size: with little data, the posterior remains close to the prior; as $n$ grows, the likelihood progressively dominates.

In generalised Bayesian inference as proposed by \cite{bissiri2016general}, the loss function remains arbitrary, but the coefficients $\gamma_k$ are taken equal across blocks, i.e., $\gamma_k = \gamma$ for all $k$.\footnote{If we also restrict the loss to the negative log-likelihood, this would correspond to models referred to with multiple names in the literature, such as power posterior, tempered posterior, safe-Bayes \citep{grunwald2012safe,grunwald2017inconsistency,holmes2017assigning,miller2019robust}.} 
In this case, the KL terms collapse to
$\frac{\gamma}{n}\mathrm{KL}(q\|\pi)$, which is also a single KL term but now
scaled by $\gamma/n$. The trade-off is therefore also governed by $\gamma$,
which encodes confidence in the prior relative to the empirical information.
Since inference combines information from the prior and from the data, a larger
$\gamma$ places greater relative weight on the former, so that more observations
are required to move the posterior away from the prior. Conversely, a smaller
$\gamma$ shifts the balance towards the data.

By allowing heterogeneity in the coefficients $\{ \gamma_k \}$, our model enables the researcher to express varying levels of confidence in different blocks of parameters. The trade-off between fitting the data and adhering to the prior can now be explored separately for each block, allowing the researcher to control how much each block’s posterior is influenced by the data versus the prior.

Additionally, our framework also admits a natural interpretation as the solution to a Lagrangian optimisation problem. Specifically, it can be viewed as minimising the expected loss subject to entropy-based constraints, that is, restricting each block component of the posterior to remain within a KL-neighbourhood of its prior counterpart. In this interpretation, the coefficients $\{\gamma_k\}$ play the role of Lagrange multipliers associated with these constraints.

In summary, the coefficients $\{\gamma_k\}$ provide source-specific control
over the posterior update.
A large $\gamma_k$ anchors the block posterior to its prior, whereas a small $\gamma_k$ allows the data to dominate. We now derive a closed-form characterisation of the SoU posterior.

\subsection{Closed-form characterisation}

The optimisation problem 
\eqref{eq:qstar} admits an explicit characterisation of its minimisers. The next theorem shows that the SoU posterior can be constructed via a backward recursion over the blocks. The characterisation applies to a general empirical loss $L_n(\btheta)$,
without requiring additivity. 
This recursion provides an analytical characterisation of the solution, whereas posterior computation is addressed through the optimisation formulation.

\begin{theorem}[Closed-form characterisation]\label{thm:Closed_Form} 
Assume that \(L_n\) is bounded below and that all the normalising constants defined below are finite and strictly positive.
Let $V_K(\btheta_{\le K})=L_n(\btheta)$ and, for $k=K,K-1,\dots,2$, define 
\[
Z_k(\btheta_{<k})
:=\int \exp\!\left\{-\frac{V_k(\btheta_{\le k})}{\gamma_k}\right\}
\pi_k(\mathrm d\btheta_{k}\mid \btheta_{<k}),
\quad
V_{k-1}(\btheta_{<k}):=-\gamma_k\log Z_k(\btheta_{<k}).
\]
For $k=2,\dots,K$ let $q_k^*(\mathrm d\btheta_{k}\mid \btheta_{<k})
:= \frac{\exp\!\left\{-V_k(\btheta_{\le k})/\gamma_k\right\}}{Z_k(\btheta_{<k})}
\,\pi_k(\mathrm d\btheta_{k}\mid \btheta_{<k})$ 
and define
\[
q_1^*(\mathrm d\btheta_{1})
:= \frac{\exp\!\left\{-V_1(\btheta_{1})/\gamma_1\right\}}{Z_1}\,\pi_1(\mathrm d\btheta_{1}),
\qquad
Z_1:=\int \exp\!\left\{-\frac{V_1(\btheta_{1})}{\gamma_1}\right\}\pi_1(\mathrm d\btheta_{1}).
\]
Then $q^*(\mathrm d\btheta)
:= q_1^*(\mathrm d\btheta_{1})\prod_{k=2}^K q_k^*(\mathrm d\btheta_{k}\mid \btheta_{<k})$ is the unique minimiser of $J(q)$ defined in \eqref{eq:KB}.
\end{theorem}

The proof is given in Section~\ref{app:closed_form} of the Supplementary Material.
Theorem~\ref{thm:Closed_Form} shows that each blockwise conditional SoU posterior is an exponential reweighting of the corresponding conditional prior, which generalises the usual Gibbs posterior construction.
The key difference is that the reweighting is not performed using the raw empirical loss $L_n$ but a sequence of
value functions $V_k$ defined by backward recursion.
Intuitively, $V_k(\btheta_{\le k})$ is the effective loss seen by block $k$ after accounting for the optimal contribution of the subsequent blocks, with $\gamma_k$ controlling the sharpness of
this update.

Corollary~\ref{cor:closed_form} (Supplementary Material) makes the connection with standard Gibbs posteriors particularly transparent.
If the prior factorises as $\pi\left(\mathrm{d}\btheta\right)=\bigotimes_{k=1}^{K}\pi_{k}\left(\mathrm{d}\btheta_{k}\right)$ and the empirical loss decomposes additively across blocks,
$L_n(\btheta)=\sum_{k=1}^K L_{n,k}(\btheta_{k})$, then the SoU posterior factorises across blocks as
\[
q^*(\mathrm{d}\btheta)=\prod_{k=1}^K q_{k}^*(\mathrm{d}\btheta_{k}),
\qquad
q_{k}^*(\mathrm{d}\btheta_{k}) \propto \exp\{-L_{n,k}(\btheta_{k})/\gamma_k\}\,\pi_k(\mathrm{d}\btheta_{k}).
\]

Thus, in this decoupled setting, each block is updated independently via a Gibbs posterior with its own temperature parameter \(\gamma_k\). This is precisely the setting in which heterogeneous confidence admits a blockwise likelihood-tempering interpretation: for negative log-likelihood loss, each block-specific likelihood component is raised to the power \(1/\gamma_k\). When confidence is homogeneous, \(\gamma_k=\gamma\), this reduces to the usual tempered posterior.

In general, however, heterogeneous confidence does not reduce to assigning different powers to separate likelihood components. When the loss couples the blocks, the update for block \(k\) depends on the recursively defined effective loss \(V_k\), which incorporates the treatment of subsequent blocks. Thus, the confidence parameters affect not only the sharpness of each local update, but also the information transmitted across blocks. We return to this interpretation in Section 3.2.

\subsection{Probabilistic inference, optimisation, and uncertainty propagation}
\label{sec:point_uncertainty}

The confidence parameters affect the SoU posterior in two related ways. First,
they control posterior location and dispersion at the block level. Second, they
determine how uncertainty associated with one block is transmitted to the
others. We develop these two perspectives below and show how heterogeneous
confidence connects probabilistic inference and optimisation.\\

\noindent\textbf{Location and dispersion.}
By rearranging the objective in \eqref{eq:KB}, we obtain a decomposition that 
separates the contributions governing posterior location and dispersion:\footnote{Here we work with the unscaled version of \eqref{eq:KB}; see Section~\ref{app:point_uncertainty} in the Supplementary Material for the derivation.}
\begin{equation}\label{eq:point_and_uncertainty}
    q^* = \argmin_q \big\{M(q) - H(q)\big\},
\end{equation}
where the two components are given by
\begin{align*}
    M(q) &:= -\mathbb{E}_{q} \left[\log P(\btheta)\right], \quad H(q) := \sum_{k=2}^{K} \gamma_k \, \mathbb{E}_{q(\btheta_{<k})} \left[ \mathcal{H}\left( q(\btheta_{k} \mid \btheta_{<k}) \right) \right] + \gamma_1 \mathcal{H}\left(q(\btheta_{1}) \right) , \\
    \,\,\text{with}&\,\, P(\btheta) := \prod_{k=2}^{K} \pi\left( \btheta_{k} \mid \btheta_{<k} \right)^{\gamma_k} \pi\left( \btheta_{1}\right)^{\gamma_1} \exp\left( - \sum_{i=1}^n\ell(\btheta, X_i) \right) Z^{-1},
\end{align*}
where \(\mathcal{H}(q) := -E_q[\log q]\) denotes the entropy and, assuming \(Z \in (0,\infty)\), \(Z^{-1}\) is the normalising constant ensuring that \(P(\theta)\) integrates to one.

The distribution $P(\btheta)$ aggregates the empirical loss with the prior contributions, each weighted by its block-specific confidence parameter $\gamma_k$. In the classical Bayesian setting\footnote{Where $\ell(\btheta,X) = - \log f(X \mid \btheta)$ and $\gamma_k = 1$ for all $k$.}, $P(\btheta)$ coincides with the standard posterior.

The term $M(q)$ measures the alignment between $q$ and this target distribution $P(\btheta)$ and determines the \emph{location} around which $q^*$ concentrates. The term $H(q)$ introduces a weighted entropy contribution that counters concentration, governing the \emph{dispersion} of the resulting posterior. To see this, consider first the case where the second term is ignored. The optimisation problem then reduces to $\hat{q} = \argmin_q M(q)$.
Let $\hat{\btheta} = \argmin_{\btheta \in \Theta} - \log P(\btheta)$.
In this case, the optimal solution is $\hat{q} = \delta_{\hat{\btheta}}$, where $\delta_{\cdot}$ denotes the Dirac measure; that is, without the entropy term, the posterior collapses to a single point at $\hat{\btheta}$.
The entropy term $H(q)$ works in the opposite direction. Because each entropy component appears with a positive weight $\gamma_k$, increasing spread in $q$ tends to lower the objective. Hence $H(q)$ counteracts the collapse induced by $M(q)$ by encouraging dispersion in $q$. The confidence parameters $\{\gamma_k\}$ act by attenuating or amplifying this effect for each parameter block $B_k$.\\

\noindent\textbf{Posterior concentration and probabilistic treatment.} 
Posterior concentration is a central aspect of Bayesian inference. 
In the standard Bayesian formulation, holding the likelihood and observed data fixed, the extent and direction of this concentration are governed by the \emph{informativeness} of the prior.
A more informative prior increases concentration while pulling the posterior more strongly towards its location, inducing a movement \emph{towards the prior}. 
This is the familiar regularising effect of prior information, where increased posterior concentration is coupled with greater prior influence. 
Our framework introduces a distinct mechanism through \emph{confidence}. Decreasing confidence also increases concentration, but attenuates the direct influence of the corresponding prior, so that inference becomes increasingly driven by the empirical loss---a movement, intuitively, \emph{towards the data}.

The decomposition in \eqref{eq:point_and_uncertainty} makes this explicit. Consider a single block $B_k$ and examine the effect of letting $\gamma_k \to 0$. In this limit,
\[
\begin{aligned}
&\pi(\btheta_{k}\mid \btheta_{<k})^{\gamma_k} \to 1
\quad \text{pointwise wherever } 0<\pi(\btheta_{k}\mid \btheta_{<k})<\infty,\\
&\gamma_k\,\E_{q(\btheta_{<k})}
\!\left[\mathcal H\!\left(q(\btheta_{k}\mid \btheta_{<k})\right)\right] \to 0.
\end{aligned}
\]

The direct prior contribution associated with block $B_k$ disappears from $P(\btheta)$.
Simultaneously, the corresponding entropy term vanishes, eliminating the force that would otherwise preserve uncertainty in that block. By jointly tuning prior informativeness and confidence, the practitioner can shape, for each block, both the extent and the direction of posterior concentration.

As $\gamma_k$ approaches zero, the conditional distribution associated with block $B_k$
becomes increasingly concentrated, and its treatment approaches optimisation.
For $\gamma_k$ bounded away from zero, the block retains a probabilistic
treatment, whose uncertainty is propagated to the remaining blocks. 
Thus, the confidence parameters control not only posterior
concentration within each block, but also how uncertainty is propagated across
the model.\\

\noindent\textbf{Uncertainty propagation.}
To make this mechanism explicit, consider a two-block model under negative
log-likelihood loss. From Theorem~\ref{thm:Closed_Form}, the recursion can be written
in terms of the effective likelihood transmitted from the second block to the
first,
\[
\mathcal L_{\gamma_2}(\btheta_1)
=
\left[
\int
p(\bX\mid\btheta_1,\btheta_2)^{1/\gamma_2}
\pi(\btheta_2\mid\btheta_1)
\,d\btheta_2
\right]^{\gamma_2}.
\]
The corresponding conditional distribution of the second block and the
distribution of the first block are
\begin{align}
q_{\gamma_2}(\btheta_2\mid\btheta_1,\bX)
&\propto
p(\bX\mid\btheta_1,\btheta_2)^{1/\gamma_2}
\pi_2(\btheta_2\mid\btheta_1),
\label{eq:two_block_conditional}
\\
q_{\gamma_1,\gamma_2}(\btheta_1\mid\bX)
&\propto
\pi_1(\btheta_1)
\mathcal L_{\gamma_2}(\btheta_1)^{1/\gamma_1}.
\label{eq:two_block_first}
\end{align}
The resulting marginal distribution of the second block is
\[
q_{\gamma_1,\gamma_2}(\btheta_2\mid\bX)
=
\int
q_{\gamma_2}(\btheta_2\mid\btheta_1,\bX)
q_{\gamma_1,\gamma_2}(\btheta_1\mid\bX)
\,d\btheta_1.
\]

These expressions make explicit how confidence affects uncertainty
propagation across blocks. The parameter $\gamma_2$ controls the conditional
treatment of the second block and, through
$\mathcal L_{\gamma_2}(\btheta_1)$, how its contribution is transmitted to
inference on the first block. In turn, $\gamma_1$ controls the probabilistic
treatment of the first block, whose uncertainty is propagated to the second
through marginalisation.\\

\noindent\textbf{Empirical Bayes and profiling as selective limits.}
The two-block representation makes the transition between probabilistic
treatment and optimisation particularly transparent.

First, let $\btheta_1$ denote a hyperparameter block and $\btheta_2$ the
target parameter block. Setting $\gamma_2=1$ gives
\[
\mathcal L_1(\btheta_1)
=
\int
p(\bX\mid\btheta_1,\btheta_2)
\pi(\btheta_2\mid\btheta_1)
\,d\btheta_2
=
m(\bX\mid\btheta_1),
\]
so the second block receives the standard Bayesian treatment, while
\[
q_{\gamma_1,1}(\btheta_1\mid\bX)
\propto
\pi(\btheta_1)
m(\bX\mid\btheta_1)^{1/\gamma_1}.
\]
At $\gamma_1=1$, both blocks receive the standard Bayesian treatment,
recovering full hierarchical Bayes. As $\gamma_1$ decreases, the distribution
over the hyperparameter block becomes increasingly concentrated around values
with high marginal likelihood. Under regularity conditions, as
$\gamma_1\to0$,
\[
q_{\gamma_1,1}(\btheta_1\mid\bX)
\Longrightarrow
\delta_{\widehat{\btheta}_1},
\qquad
\widehat{\btheta}_1
=
\arg\max_{\btheta_1}
m(\bX\mid\btheta_1).
\]
Consequently,
\[
q_{\gamma_1,1}(\btheta_2\mid\bX)
\longrightarrow
q(\btheta_2\mid\widehat{\btheta}_1,\bX),
\]
recovering \textit{empirical Bayes} \citep{carlin2000bayes}. 
Under full Bayesian treatment, uncertainty about the hyperparameters is
propagated to the target parameters through marginalisation,
whereas
empirical Bayes suppresses this propagation through optimisation. SoU places
both cases within a broader continuum, with the confidence parameter $\gamma_1$
controlling how much of this uncertainty is retained and propagated.

Conversely, let $\btheta_1$ denote the target parameter block and
$\btheta_2$ a nuisance parameter block, and keep $\gamma_1=1$. 
As $\gamma_2\to0$,
the conditional distribution 
in \eqref{eq:two_block_conditional} 
becomes increasingly concentrated around values with high conditional
likelihood. 
Under regularity conditions and assuming a unique conditional maximiser for
simplicity, this yields
\[
q_{\gamma_2}(\btheta_2\mid\btheta_1,\bX)
\Longrightarrow
\delta_{\widehat{\btheta}_2(\btheta_1)},
\qquad
\widehat{\btheta}_2(\btheta_1)
=
\arg\max_{\btheta_2}
p(\bX\mid\btheta_1,\btheta_2).
\]
Moreover, by the Laplace principle,
\[
\mathcal L_{\gamma_2}(\btheta_1)
\longrightarrow
\sup_{\btheta_2}
p(\bX\mid\btheta_1,\btheta_2)
=:L_{\mathrm{prof}}(\btheta_1).
\]
Therefore,
\[
q_{1,\gamma_2}(\btheta_1\mid\bX)
\longrightarrow
q_{\mathrm{prof}}(\btheta_1\mid\bX)
\propto
\pi(\btheta_1)L_{\mathrm{prof}}(\btheta_1),
\]
recovering a \textit{profiling-type treatment} \citep{murphy2000profile}.\\

\noindent\textbf{Heterogeneous confidence and selective collapse.}
These examples highlight the distinctive role of heterogeneous confidence.
Standard Bayes assigns $\gamma_k=1$ to every block, while
generalised Bayes assigns a common $\gamma_k=\gamma$. In both cases all
sources receive the same confidence treatment. SoU relaxes this homogeneity,
making \emph{selective collapse} possible: one source may approach an
optimisation regime while others retain a probabilistic treatment.

In the two-block setting,
\begin{center}
\begin{minipage}[c]{0.54\textwidth}
\[
\begin{aligned}
(\gamma_1,\gamma_2) &= (1,1)
&& \text{full Bayes (\textcolor{FBviolet}{FB})},\\
(\gamma_1,\gamma_2) &= (\gamma,\gamma)
&& \text{generalised Bayes (\textcolor{FBviolet}{GB})},\\
(\gamma_1,\gamma_2) &\to (0,1)
&& \text{empirical Bayes (\textcolor{EBred}{EB})},\\
(\gamma_1,\gamma_2) &\to (1,0)
&& \text{profile likelihood (\textcolor{PLblue}{PL})},\\
(\gamma_1,\gamma_2) &\to (0,0)
&& \text{empirical risk minimisation (ERM)}.
\end{aligned}
\]
\end{minipage}
\hfill
\begin{minipage}[c]{0.40\textwidth}
\centering
\begin{tikzpicture}[scale=2.1,>=Latex]
\begin{scope}
\clip (0,0) rectangle (1.28,1.28);
\foreach \ix in {0,...,63} {
  \foreach \iy in {0,...,63} {
    \pgfmathsetmacro{\xx}{\ix/64*1.28}
    \pgfmathsetmacro{\yy}{\iy/64*1.28}
    \pgfmathsetmacro{\u}{min(\xx,1)}
    \pgfmathsetmacro{\v}{min(\yy,1)}
    \pgfmathsetmacro{\rr}{(1-\u)*\v*215 + \u*(1-\v)*35 + \u*\v*150}
    \pgfmathsetmacro{\gg}{(1-\u)*\v*40 + \u*(1-\v)*90 + \u*\v*45}
    \pgfmathsetmacro{\bb}{(1-\u)*\v*50 + \u*(1-\v)*220 + \u*\v*175}
    \pgfmathsetmacro{\R}{(1-.13)*255+.13*\rr}
    \pgfmathsetmacro{\G}{(1-.13)*255+.13*\gg}
    \pgfmathsetmacro{\B}{(1-.13)*255+.13*\bb}
    \definecolor{tilecolor}{RGB}{\R,\G,\B}
    \fill[tilecolor] (\xx,\yy) rectangle ++(.0201,.0201);
  }
}
\end{scope}
\draw[thick,->,black!70] (0,0)--(1.28,0) node[right] {$\gamma_1$};
\draw[thick,->,black!70] (0,0)--(0,1.28) node[above] {$\gamma_2$};
\draw[thick,->,black!70] (0,0) -- (1.28,1.28);
\draw (1,.025)--(1,-.025) node[below=3pt] {$1$};
\draw (.025,1)--(-.025,1) node[left=3pt] {$1$};
\node[below left=2pt] at (0,0) {$0$};
\foreach \i in {0,...,49} {
 \pgfmathsetmacro{\a}{.965*\i/50}
 \pgfmathsetmacro{\b}{.965*(\i+1)/50}
 \pgfmathtruncatemacro{\p}{100*(\i+.5)/50}
 \draw[line width=1.05pt,draw=EBred!\p!black] (0,\a)--(0,\b);
 \draw[line width=1.05pt,draw=PLblue!\p!black] (\a,0)--(\b,0);
 \draw[line width=1.15pt,draw=FBviolet!\p!black] (\a,\a)--(\b,\b);
}
\draw[-{Latex[length=3.2pt,width=2.6pt]},EBred,line width=1.05pt] (0,.955)--(0,1);
\draw[-{Latex[length=3.2pt,width=2.6pt]},PLblue,line width=1.05pt] (.955,0)--(1,0);
\draw[-{Latex[length=3.5pt,width=3pt]},FBviolet,line width=1.15pt] (.955,.955)--(1,1);
\draw[dashed,dash pattern=on 2.2pt off 1.7pt,line width=.8pt,EBred!55!FBviolet] (0,1)--(1,1);
\draw[dashed,dash pattern=on 2.2pt off 1.7pt,line width=.8pt,PLblue!55!FBviolet] (1,0)--(1,1);
\fill[black] (0,0) circle (1.5pt);
\fill[EBred] (0,1) circle (1.8pt);
\fill[PLblue] (1,0) circle (1.8pt);
\fill[FBviolet] (1,1) circle (1.8pt);
\node[above right=2pt] at (.055,-.025) {\scriptsize ERM};
\node[EBred,above right=2pt] at (0,1) {\scriptsize EB};
\node[PLblue,above right=2pt] at (1,0) {\scriptsize PL};
\node[FBviolet,above=4pt] at (1,1) {\scriptsize FB};
\node[rotate=45,above=3pt,text=FBviolet!75!black] at (.53,.53) {\scriptsize GB};
\node[black!80] at (.69,.32) {\small SoU};
\end{tikzpicture}
\end{minipage}
\end{center}

Empirical Bayes and profiling therefore arise as heterogeneous-confidence
configurations that cannot be obtained through a single global learning rate.
Geometrically, 
generalised Bayes is confined to the 
non-negative diagonal $(\gamma,\gamma)$ of the confidence plane, whereas SoU opens this
diagonal into the full non-negative quadrant, encompassing empirical Bayes and
profile-likelihood configurations.

Moreover, SoU naturally extends beyond two blocks, with the same geometric
picture carrying over to the corresponding higher-dimensional non-negative
orthant. For example, it can accommodate configurations such as
\(
(\gamma_1,\gamma_2,\gamma_3)\to(0,1,0),
\)
where the first and third blocks are driven towards optimisation while the
second retains a probabilistic treatment.

\noindent\textbf{\textbf{Objective collapse.}}
Selective collapse becomes global when all confidence parameters vanish. In
this regime, every block is driven towards optimisation and the SoU objective
approaches empirical risk minimisation.

The following result formalises the ``objective collapse'' phenomenon: for a fixed $n$, as the confidence parameters $\left\Vert \bgamma\right\Vert _{\infty}=\max_{1\le k\le K}\gamma_{k}$ vanish, the SoU posterior concentrates on empirical loss minimisers and the KL regularisation becomes irrelevant. We write $L_n^*:=\inf_{\btheta\in\Theta} L_n(\btheta)$ for the minimum empirical loss, and, for $\delta>0$, let
$S_\delta:=\{\btheta\in\Theta: L_n(\btheta)\le L_n^*+\delta\}$ denote the $\delta$-near-optimal level set of $L_n$.

\begin{theorem}[Objective collapse as $\|\bgamma\|_\infty\to0$]\label{thm:collapse}
Assume (i) minimisers exist for every $\bgamma\in(0,\infty)^K$, and (ii) the prior assigns positive
mass to every near-optimal level set of $L_n$, that is, $\pi(S_\delta)>0$ for every $\delta>0$.
Let $\bgamma^{(m)}$ satisfy $\|\bgamma^{(m)}\|_\infty\to0$ and let $q_m$ be any corresponding sequence
of minimisers.
Then $\E_{q_m}L_n\to L_n^*$ and $q_m\!\left(\btheta:L_n(\btheta)\ge L_n^*+\varepsilon\right)\to0$ for every $\varepsilon>0$.
If $L_n$ has a unique well-separated minimiser $\hat\btheta_n$, then $q_m \rightsquigarrow \delta_{\hat\btheta_n}$.
Moreover, the weighted $\mathrm{KL}$ penalty vanishes: 
$$\sum_{k=1}^K\gamma_k^{(m)}\mathrm{KL}_k(q_m\|\pi)\to0,$$ 
where $\mathrm{KL}_{k}$ denotes the $k$-th blockwise $\mathrm{KL}$ term appearing in the objective. 
\end{theorem}

The proof is given in Section~\ref{app:collapse} of the Supplementary Material.
This limiting interpretation is related to the Bayesian Learning Rule (BLR)
\citep{khan2023bayesian}, in which an entropy term lifts a deterministic optimisation
objective into a probabilistic update. 
From this perspective, SoU can be viewed as performing this lifting selectively
across blocks, with the confidence parameters controlling which components are
lifted into a probabilistic treatment and which remain in an optimisation
regime.

Theorem \ref{thm:collapse} describes a finite-sample regime (fixed $n$) as $\|\bgamma\|_\infty\to0$. Next, we study the asymptotic regime $n \to \infty$, allowing $\bgamma$ to depend on $n$.

\subsection{Consistency of the SoU posterior}\label{subsec:consistency}

The following theorem provides a basic frequentist consistency guarantee for the SoU posterior.
Under standard regularity and identifiability conditions, and provided the prior assigns sufficient local mass
around the population-risk minimiser $\btheta_0$ while the confidence weights $\bgamma_n$ do not dominate the data as
$n\to\infty$, the SoU posterior concentrates on $\btheta_0$.

\begin{theorem}[Consistency]\label{thm:consistency}
Let $X_1,\dots,X_n$ be i.i.d.\ from $\mathrm P_0$.
Define 
$$R_n(\btheta):=n^{-1}\sum_{i=1}^n \ell(\btheta,X_i)\quad \text{and} \quad R(\btheta):=\E_{\mathrm P_0}\ell(\btheta,X).$$
Let $q_n^*$ denote any minimiser of \eqref{eq:KB} with prior $\pi$ and confidence vector $\bgamma_n\in(0,\infty)^K$.
Assume: (i) $R_n$ converges to $R$ uniformly on compacts a.s.; (ii) $\Theta$ is compact; (iii) $\btheta_0$ is the unique well-separated minimiser of $R$; and
(iv) the prior has sufficient local mass around $\btheta_0$ in the sense that for some $r_n\downarrow0$, $\pi(B(\btheta_0,r_n))>0$ for sufficiently large $n$ and
\[
\frac{\|\bgamma_n\|_\infty}{n}\log\!\left(\frac{1}{\pi(B(\btheta_0,r_n))}\right)\to0.
\]
Then for every $\varepsilon>0$, $q_n^*(B(\btheta_0,\varepsilon))\xrightarrow{\mathrm{a.s.}}1$, and hence
$q_n^*\rightsquigarrow \delta_{\btheta_0}$ almost surely.
\end{theorem}

The proof is given in Section~\ref{app:consistency}  of the Supplementary Material. For simplicity, we stated the result for compact 
$\Theta$ and a fixed prior $\pi$. Section~\ref{app:consistency} establishes a more general version allowing non-compact $\Theta$ and $n$-dependent priors $\pi=\pi_n$.

\subsection{Choosing confidence}\label{sec:choosing_confidence}

The confidence parameters are part of the inferential specification of the
SoU model. They may be chosen directly to encode how different sources of
uncertainty should be treated, or calibrated according to a chosen inferential
objective. 

Empirical Bayes and profiling, as discussed in
Section~\ref{sec:point_uncertainty}, provide canonical examples of the first
approach, in which different blocks are assigned probabilistic or
optimisation-like treatments. The sparse construction developed in
Section~\ref{sec:sou-sgl} follows the same principle, with low confidence
assigned to the local components governing sparsity while probabilistic
inference is retained for the remaining parameters. This yields an
``optimise structure, infer parameters'' configuration.

For the second approach, we consider three calibration criteria based on
predictive performance, uncertainty quantification, and information matching.\\

\noindent\textbf{Predictive calibration.}
Each confidence vector $\bgamma$ induces a different SoU posterior and,
consequently, a different posterior predictive distribution,
\begin{equation}
p_{\bgamma}(\widetilde X\mid\bX)
=
\int
p(\widetilde X\mid\btheta)
q_{\bgamma}(\btheta\mid\bX)
\,d\btheta ,
\end{equation}
where $\widetilde X$ denotes future data.
The confidence parameters can therefore be calibrated according to
out-of-sample predictive performance. Let
$\widehat{\mathcal S}_{\mathrm{pred}}(p_{\bgamma})$
denote an estimated predictive score, obtained, for example, from held-out
data or cross-validation. We then choose
\begin{equation}
\widehat{\bgamma}
=
\arg\max_{\bgamma\in\Gamma}
\widehat{\mathcal S}_{\mathrm{pred}}(p_{\bgamma}),
\end{equation}
where $\Gamma$ denotes the set of candidate confidence vectors. The
predictive score may assess the full predictive distribution or a
problem-specific predictive objective.\\

\noindent\textbf{Uncertainty calibration.}
Confidence can also be calibrated according to posterior uncertainty for a
parameter of inferential interest. Consider, for simplicity, two blocks
$\btheta_t$ and $\btheta_s$, where $t$ denotes the inferential target and $s$
the source whose confidence is to be calibrated. These labels are functional
rather than positional. We fix
\(
\gamma_t=1
\)
and
\(
\gamma_s=\gamma,
\)
so that different values of $\gamma$ correspond to different treatments of
the uncertainty associated with $\btheta_s$, and hence induce different
posterior distributions $q_\gamma(\btheta_t\mid\bX)$ for the same target.

Let $C_{\gamma,\alpha}(\bX)$ denote a $(1-\alpha)$ posterior credible region,
such that
\[
q_\gamma
\left(
\btheta_t\in C_{\gamma,\alpha}(\bX)
\mid\bX
\right)
=
1-\alpha.
\]
A natural calibration principle is to choose $\gamma$ so that the resulting
credible region has the corresponding repeated-sampling coverage,
\begin{equation}
P_0
\left\{
\btheta_{t,0}
\in
C_{\gamma,\alpha}(\bX)
\right\}
=
1-\alpha,
\end{equation}
where $\btheta_{t,0}$ denotes the population target.

In practice, this coverage can be approximated by bootstrap. Let
$\widehat{\btheta}_t$ be a $\gamma$-independent consistent estimate of the
target,\footnote{For example, an empirical risk minimiser may be used.} computed from the original data, and let
$\bX^{*(1)},\ldots,\bX^{*(B)}$ denote bootstrap samples. Define
\[
\widehat c_\alpha(\gamma)
=
\frac{1}{B}
\sum_{r=1}^B
\mathbf 1
\left\{
\widehat{\btheta}_t
\in
C_{\gamma,\alpha}(\bX^{*(r)})
\right\}.
\]
The confidence parameter can then be selected so that
\(
\widehat c_\alpha(\widehat\gamma)
\approx
1-\alpha.
\)
Thus, the confidence assigned to source $\btheta_s$ is calibrated according
to how its treatment affects uncertainty quantification for the target
$\btheta_t$. In the empirical Bayes setting discussed above, this provides a
way to choose how much hyperparameter uncertainty is propagated to the target.\\

\noindent\textbf{Information matching.}
A third approach is to calibrate confidence by matching the asymptotic
information of the resulting posterior to a reference distribution. We first
recall the homogeneous-confidence case of \citet{lyddon2019general}. Let
\[
\btheta_0
=
\argmin_{\btheta}
\mathbb E_{P_0}
\left[
\ell(\btheta,X)
\right],
\]
and define
\[
I(\btheta_0)
=
\mathbb E_{P_0}
\left[
\nabla\ell(\btheta_0,X)
\nabla\ell(\btheta_0,X)^\top
\right],
\qquad
J(\btheta_0)
=
\mathbb E_{P_0}
\left[
\nabla^2\ell(\btheta_0,X)
\right].
\]
Under regularity conditions, the generalised posterior with confidence
parameter $\gamma$ and the corresponding loss-likelihood bootstrap are
asymptotically Gaussian, with covariance matrices
\[
\Sigma_{\mathrm{GB}}(\gamma)
=
\gamma J(\btheta_0)^{-1},
\qquad
\Sigma_{\mathrm{LL}}
=
J(\btheta_0)^{-1}
I(\btheta_0)
J(\btheta_0)^{-1},
\]
for the corresponding $\sqrt n$-scaled distributions.

\citet{lyddon2019general} calibrate $\gamma$ by matching the Fisher
information numbers of the two covariance matrices,
\(
K(\Sigma)
=
\operatorname{tr}(\Sigma^{-1}).
\)
Imposing
$K\!\left(\Sigma_{\mathrm{GB}}(\gamma)\right)
=
K\!\left(\Sigma_{\mathrm{LL}}\right)$
gives
\begin{equation}\label{eq:gamma_information_matching}
\gamma^*
=
\frac{
\operatorname{tr}\!\left\{J(\btheta_0)\right\}
}{
\operatorname{tr}\!\left\{
J(\btheta_0)
I(\btheta_0)^{-1}
J(\btheta_0)^\top
\right\}
}.
\end{equation}
For a correctly specified likelihood,
$I(\btheta_0)=J(\btheta_0)$, and therefore $\gamma^*=1$, recovering the
standard Bayesian calibration.

The same principle can be extended to heterogeneous confidence by matching
information at the level of the blocks. 
Let $I:=I(\btheta_0)$ and $J:=J(\btheta_0)$, and write $J$ according to the
block structure of the SoU model.
Starting from $J^{(K)}=J$, for $k=K,\ldots,2$ define
\begin{equation}
D_k
=
J^{(k)}_{kk},
\qquad
J^{(k-1)}
=
J^{(k)}_{<k,<k}
-
J^{(k)}_{<k,k}
D_k^{-1}
J^{(k)}_{k,<k},
\end{equation}
with $D_1=J^{(1)}$. 
Thus, $D_k$ determines the local quadratic approximation associated with block
$k$ after the subsequent blocks have been integrated out through the SoU
recursion. It follows that, for the corresponding $\sqrt n$-scaled SoU factors,
\begin{equation}\label{eq:sou_factor_covariance}
\operatorname{Cov}_{\mathrm{SoU}}
\left(
\btheta_k\mid\btheta_{<k},\bX
\right)
\approx
\gamma_k D_k^{-1},
\qquad
k=1,\ldots,K,
\end{equation}
where the conditioning is absent for $k=1$; see
Section~\ref{app:information_matching} of the Supplementary Material for the
derivation.

We use the loss-likelihood bootstrap as the reference distribution, with
asymptotic covariance
\[
\Sigma_{\mathrm{LL}}
=
J^{-1}IJ^{-1}.
\]
Factorising the corresponding Gaussian distribution according to the same
block order, let $V_k$ denote the covariance of its $k$th conditional factor.
Specifically,
\begin{equation}
V_1
=
(\Sigma_{\mathrm{LL}})_{11},
\qquad
V_k
=
(\Sigma_{\mathrm{LL}})_{kk}
-
(\Sigma_{\mathrm{LL}})_{k,<k}
(\Sigma_{\mathrm{LL}})_{<k,<k}^{-1}
(\Sigma_{\mathrm{LL}})_{<k,k},
\quad k\geq2.
\end{equation}

Information matching can then be applied separately to each block,
\[
K\!\left(\gamma_kD_k^{-1}\right)
=
K(V_k),
\qquad
k=1,\ldots,K.
\]
Since $K(\Sigma)=\operatorname{tr}(\Sigma^{-1})$, this gives the
source-specific calibration
\begin{equation}\label{eq:sou_information_matching}
\gamma_k^*
=
\frac{\operatorname{tr}(D_k)}
{\operatorname{tr}(V_k^{-1})},
\qquad
k=1,\ldots,K.
\end{equation}
Thus, the single global calibration parameter of generalised Bayes is replaced
by source-specific confidence parameters, each matching the information of the corresponding factor in the SoU decomposition.
For a correctly specified likelihood, $I(\btheta_0)=J(\btheta_0)$, so that
$V_k=D_k^{-1}$ and $\gamma_k^*=1$ for every block, recovering standard Bayes.

\subsection{Example of the confidence mechanism}\label{sec:examples_confidence}

We illustrate the role of confidence using an exponential-family model with a conjugate prior. For clarity, we work in the homogeneous setting $\gamma_k = \gamma$, which corresponds to generalised Bayesian inference \citep{bissiri2016general}. This yields simple closed-form expressions and isolates the effect of confidence; the same mechanism applies blockwise under heterogeneous $\{\gamma_k\}$.

In exponential-family models, a convenient parametrisation elegantly reveals how prior information combines with observed data
\citep{diaconis1979conjugate}. 
Let the model belong to the exponential family, with likelihood for a sample 
$\bX_{1:n} = (X_1, \dots, X_n)$ given by
\begin{equation*}
p(\bX_{1:n} \mid \btheta)
= \left[\prod_{i=1}^n h(X_i)\right] c(\btheta)^n 
\exp\left\{ \btheta^\top \sum_{i=1}^n T(X_i) \right\},
\end{equation*}
where $T(\cdot)$ is the sufficient statistic for the parameter $\btheta$. 
A conjugate prior for this model can be written as
\begin{equation*}
\pi(\btheta \mid N_0, T_0)
= \kappa(N_0, T_0)\, c(\btheta)^{N_0} 
\exp\left\{ N_0 \btheta^\top T_0 \right\},
\end{equation*}
where $T_0$ represents a prior guess for the sufficient statistic 
and $N_0$ quantifies the strength of this prior belief, governing the prior’s concentration around $T_0$.
The resulting posterior distribution remains in the same family as the prior, 
with updated parameters $\pi(\btheta \mid N_n, T_n)$, where:
\[
N_n := N_0+n,\quad 
T_n := \alpha_n T_0 + (1-\alpha_n) \bar T_n, 
\quad
\alpha_n := \frac{N_0}{N_0 + n}, \quad 
\bar T_n := \frac{1}{n} \sum_{i=1}^n T(X_i).
\]
This parametrisation yields a clear interpretation: computing the posterior is equivalent to having observed 
$N_0$ pseudo-observations with average sufficient statistic $T_0$, which are pooled with the 
$n$ actual observations summarised by 
$\bar T_n$.

Under this representation, \(N_0\) performs two tasks simultaneously. It determines the shrinkage of the posterior parameter \(T_n\) towards the prior value \(T_0\) through the factor \(\alpha_n\), and it controls the overall precision of the posterior through the updated effective sample size \(N_n\). This coupling implies a practical constraint: making the posterior more concentrated by increasing \(N_0\) necessarily increases the influence of the prior sufficient-statistic parameter \(T_0\) on \(T_n\).

In our extended framework, the posterior's effective sample size and the shrinkage factor respectively become (see Section~\ref{app:Exponential} of the Supplementary Material for the derivation)
\[
N_n = N_0 + \frac{n}{\bgamma},\quad 
\alpha_n = \frac{\gamma N_0}{\gamma N_0 + n}.
\]

Here, smaller values of \(\gamma\) increase \(N_n\), yielding a more concentrated posterior, while at the same time decreasing \(\alpha_n\), shifting \(T_n\) towards the data-driven component \(\bar T_n\).

By combining the original prior quantity \(N_0\), capturing the degree of trust in the prior location \(T_0\), with the confidence parameter \(\gamma\), capturing the degree of trust in the specified prior distribution \(\pi\), the researcher can now tune separately the informativeness of the prior and the degree of posterior concentration.
The pair $(N_0,\gamma)$ thus allows one to specify independently both where the posterior is pulled (via 
$\alpha_n$) and how tightly it contracts (via 
$N_n$), providing enhanced flexibility for Bayesian learning and regularisation.

Figure~\ref{fig:oneParameter} illustrates this mechanism as a special case in the Normal-Normal model with unknown mean $\mu$ and known variance $\sigma^2$, varying the prior scale $\sigma_0$ and the confidence parameter $\gamma_\mu$; see Section~\ref{app:normal_normal} of the Supplementary Material for the derivation.

\begin{figure}[h]
\centering
\includegraphics[width=\textwidth]{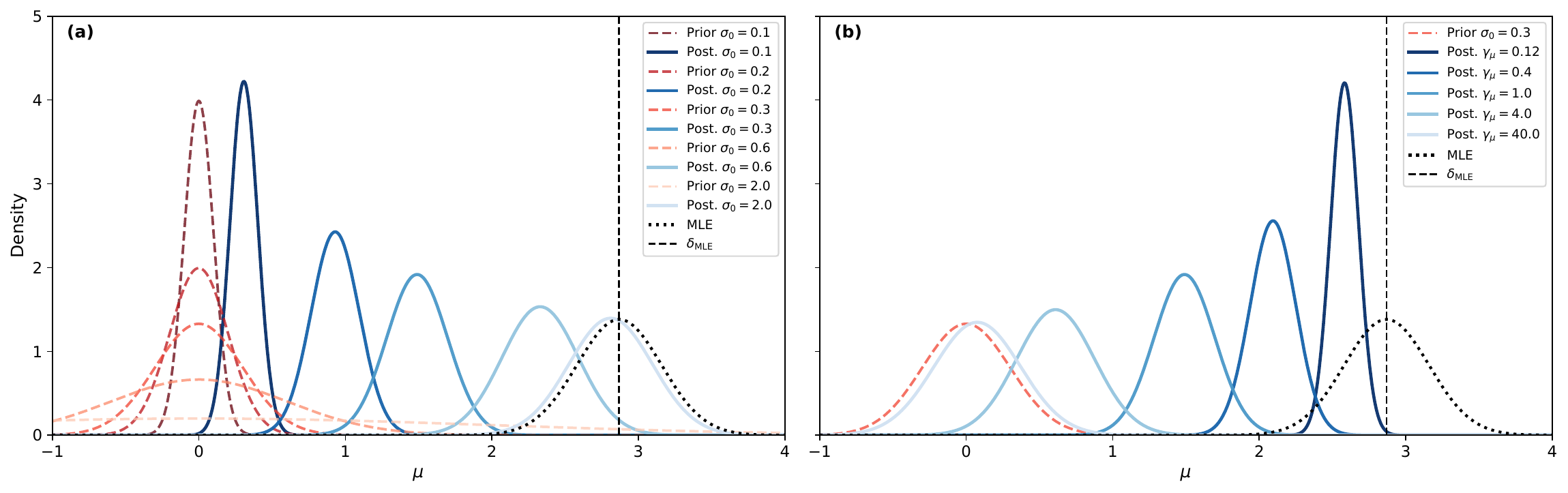}
\caption{\footnotesize{Normal-Normal illustration of prior informativeness vs confidence. (a) Prior (dashed) and posterior (solid) for $\mu$ across prior scales $\sigma_0$. (b) Prior and SoU posterior for $\mu$ as the confidence parameter $\gamma_\mu$ varies (prior fixed). Contrasting regularisation towards the prior (a) and towards the data (b).}}
\label{fig:oneParameter}
\end{figure}

\section{SoU sparse global-local (SoU-SGL)}\label{sec:sou-sgl}

In this section, we build on SoU’s selective treatment of uncertainty to develop a general approach for inducing sparsity in statistical models.\footnote{Sparsity is understood here in an effective numerical sense, based on estimated coefficients falling below a small magnitude threshold rather than being exactly zero.} 
We begin by illustrating the mechanics of our approach in the canonical Normal-means setting with a global-local prior structure. We then extend the analysis to linear regression, logistic regression and  Bayesian neural networks. Implementation details and additional experiments are provided in Section~\ref{app:expe} of the Supplementary Material.

\subsection{Simple global-local shrinkage}\label{sec:simple_gl}

Assume there are $I$ distinct units (or parameters) of interest. For each unit $i$, we observe $n$ independent and identically distributed realisations. The model is specified as follows:
\begin{align*}
y_{ij} \mid \beta_i, \sigma^2 &\sim \mathcal{N}(\beta_i, \sigma^2),  \quad \text{for } j = 1, \dots, n, \\
\beta_i \mid \tau_i, \nu &\sim \mathcal{N}(0, \tau_i^2 \nu^2), \\
\tau_i^2 \sim \pi_\tau, \quad \nu^2 &\sim \pi_\nu, \quad \sigma^2 \sim \pi_\sigma,
\end{align*}
where \( \nu > 0 \) is a global shrinkage parameter and \( \tau_i > 0 \) is a local shrinkage parameter specific to the \( i \)-th coordinate.

In this setup, the posterior expectation of \( \beta_i \) given \( \mathbf{y}_i:=\{ y_{ij}\}_{j=1}^n \), \( \tau_i \), \( \nu \), and \(\sigma^2\) is
\begin{equation}
\mathbb{E}[\beta_i \mid \mathbf{y}_i, \tau_i, \nu, \sigma^2] = (1 - \kappa_i) \bar y_i + \kappa_i \cdot 0, \nonumber
\end{equation}
where $\bar y_i = \frac{1}{n}\sum_{j=1}^n y_{ij}$ is the sample mean, and the shrinkage factor $\kappa_i$ is:
\begin{equation}
 \kappa_i = \frac{1}{1 + n \tau_i^2 \nu^2 / \sigma^2}. \nonumber
\end{equation}
This expression shows that the posterior mean of \( \beta_i \) is a convex combination of the sample mean \( \bar y_i \) and zero, with \(\kappa_i \in [0,1]\) governing the degree of shrinkage towards zero.  


This global-local structure separates model components with qualitatively different roles. The global parameter \(\nu\) shrinks all coefficients towards zero, while the local parameters \(\tau_i\) allow individual   coefficients to escape. 
We exploit this distinction through SoU's selective treatment of uncertainty, assigning near-zero confidence to the local scales, \(\gamma_\tau \approx 0\), while retaining unit confidence for the remaining model components. 
In this setting, the treatment of the local scales governing the escape decisions becomes increasingly optimisation-like, whereas uncertainty over the coefficients and the remaining shrinkage hierarchy retains the standard Bayesian treatment. 
We refer to this configuration as the \emph{sources of uncertainty sparse global-local} (SoU-SGL) estimator. 
This construction forms the basis for the applications to linear regression, logistic regression, and Bayesian neural networks considered later in this section.

As observed by \cite{carvalho2009handling}, in standard global-local models the distribution assigned to the local parameter \(\tau_i\) plays a central role in determining the pattern of shrinkage. Different choices of \(\pi_\tau\) induce different distributions for \(\kappa_i\), and consequently different shrinkage profiles. For instance, taking \(\tau_i \sim C^+(0,1)\) induces the characteristic horseshoe-shaped prior on \(\kappa_i\), expressing the prior belief that coefficients should tend either to escape shrinkage (\(\kappa_i \approx 0\)) or to be strongly shrunk towards zero (\(\kappa_i \approx 1\)). 
SoU-SGL leads to two effects that fundamentally differentiate our approach from standard 
global-local inference. 
First, it attenuates the direct influence of the local prior \(\pi_\tau\), since \(\pi_\tau^{\gamma_\tau} \to 1\) when $\gamma_\tau \to 0$ (pointwise wherever $0<\pi_\tau<\infty$), progressively decoupling the location of the local-scale posterior from the specification of \(\pi_\tau\). 
Second, the posterior for \(\tau_i\) concentrates sharply around its data-driven solution, thereby reducing uncertainty over the local escape mechanism and sharpening the separation between coefficients that remain under the global shrinkage and those that escape it. 
This sharper separation promotes not only shrinkage but sparsity.

To illustrate, we run a simulation experiment with $I=100$ parameters and $n=5$ observations per parameter, with 20 strong signals ($\beta_i=10$) and 80 nulls ($\beta_i=0$). We illustrate the posterior distributions of the local shrinkage factors $\kappa_i$ under the Horseshoe and under SoU-SGL in Figure~\ref{fig:Simple_HS}.

\begin{figure}[h]
    \centering
    \includegraphics[width=1\textwidth]{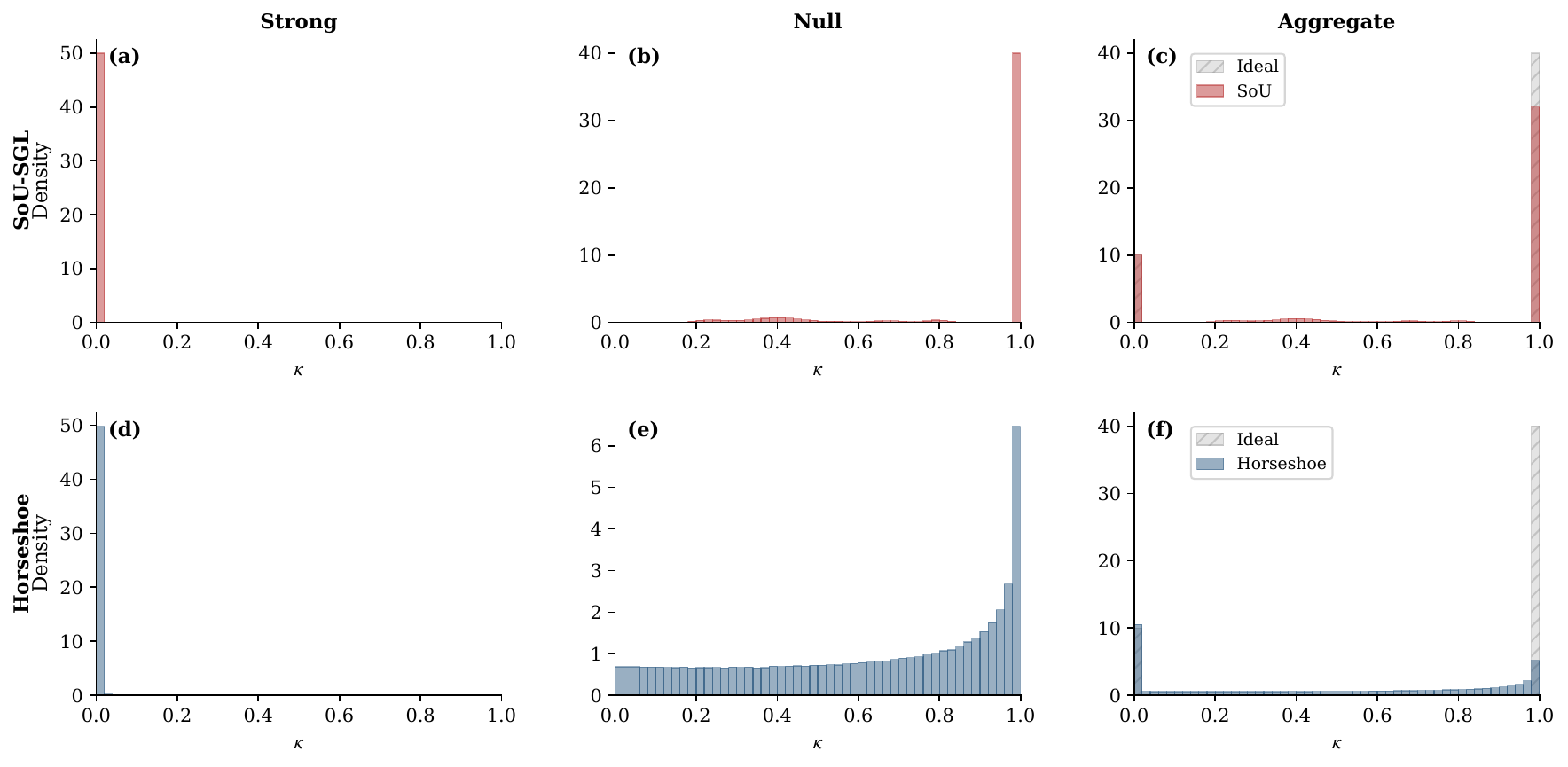} 
    \caption{\footnotesize{Shrinkage under SoU-SGL vs Horseshoe in a Normal-means simulation. Posterior densities of $\kappa_i$ for strong signal coordinates (left column), null coordinates (middle column), and aggregated across coordinates (right column). 
    SoU-SGL concentrates sharply near \(\kappa_i=0\) for strong signals and near \(\kappa_i=1\) for null coordinates.}}
    \label{fig:Simple_HS}
\end{figure}

For strong signals, both methods concentrate their posterior mass near \(\kappa_i=0\), thereby preserving the large effects. The difference is more pronounced for null coordinates. Under the Horseshoe, substantial mass remains spread over intermediate values of \(\kappa_i\), whereas SoU-SGL concentrates sharply near \(\kappa_i=1\). Thus, SoU-SGL identifies null coordinates more decisively, without compromising the ability of strong signals to escape shrinkage.
In aggregate, this produces a distribution substantially closer to the ideal 20/80 signal-null separation.

\subsection{SoU-SGL regression}

We now apply our SoU-SGL model to linear regression. We adapt and extend the regression experiments of \citet{carvalho2009handling}, benchmarking our SoU-SGL model against Lasso and Horseshoe models. 
For concision, we present only one experiment in the main text; the remaining experiments yield the same qualitative pattern and are reported in Section~\ref{app:linreg_setup} of the Supplementary Material.

The true parameter vector $\boldsymbol{\beta}$ is constructed by repeating a fixed block of non-zero coefficients given by $\beta_{1:10} = (1,2,3,4,5,6,7,8,9,10)$ until the desired number $s$ of non-zero coefficients is reached, and padding the remaining components with zeros. The design matrix $\mathbf{X}$ is generated with standard normal entries, and the response vector $\mathbf{y}$ is simulated as $\mathbf{y} = \mathbf{X}\boldsymbol{\beta} + \boldsymbol{\varepsilon}$, where $\boldsymbol{\varepsilon} \sim \mathcal{N}(0, \mathbf{I}_n)$. Each configuration is repeated over 100 simulations, and the reported performance values are the means over these simulation runs. The Lasso regularisation parameter was chosen by cross-validation.

\begin{table}[htbp]
\centering
\caption{Linear regression simulation results. Test-prediction RMSE and, for null coordinates ($\beta_j=0$), the proportion of estimated coefficients $ |\hat\beta_j|$ below fixed thresholds, averaged over repetitions, comparing Lasso, Horseshoe, and SoU-SGL across $(n,p,s)$ configurations.}
\label{tab:experiment2}
\footnotesize
\setlength{\tabcolsep}{4.2pt}
\begin{tabular}{@{} rr r *{3}{S[table-format=1.2]} *{3}{S[table-format=1.2]} *{3}{S[table-format=1.2]} *{3}{S[table-format=1.2]} @{}}
\toprule
\multirow{2}{*}{$n$} & \multirow{2}{*}{$p$} & \multirow{2}{*}{$s$} & \multicolumn{3}{c}{RMSE} & \multicolumn{3}{c}{prop $<10^{-1}$} & \multicolumn{3}{c}{prop $<10^{-2}$} & \multicolumn{3}{c}{prop $<10^{-3}$} \\
\cmidrule(lr){4-6}\cmidrule(lr){7-9}\cmidrule(lr){10-12}\cmidrule(lr){13-15}
& & & Lasso & HS & SoU & Lasso & HS & SoU & Lasso & HS & SoU & Lasso & HS & SoU \\
\midrule
25 & 20 & 10 & 1.86 & 1.70 & \textbf{1.37} & 0.45 & 0.34 & \textbf{0.96} & 0.28 & 0.04 & \textbf{0.96} & 0.26 & 0.00 & \textbf{0.96} \\

\addlinespace
\multirow{2}{*}{55} & \multirow{2}{*}{50} & 10
& 1.42 & 1.29 & \textbf{1.21}
& 0.82 & 0.81 & \textbf{0.94}
& 0.62 & 0.13 & \textbf{0.94}
& 0.59 & 0.01 & \textbf{0.94} \\

& & 20
& 1.87 & 1.72 & \textbf{1.37}
& 0.59 & 0.55 & \textbf{0.95}
& 0.35 & 0.07 & \textbf{0.95}
& 0.32 & 0.01 & \textbf{0.95} \\

\addlinespace
\multirow{3}{*}{105} & \multirow{3}{*}{100} & 10
& 1.21 & \textbf{1.10} & 1.15
& 0.95 & \textbf{0.97} & 0.93
& 0.78 & 0.34 & \textbf{0.92}
& 0.75 & 0.04 & \textbf{0.92} \\

& & 20
& 1.38 & 1.27 & \textbf{1.21}
& 0.88 & 0.89 & \textbf{0.93}
& 0.64 & 0.18 & \textbf{0.92}
& 0.60 & 0.02 & \textbf{0.92} \\

& & 50
& 2.94 & 2.13 & \textbf{1.51}
& 0.41 & 0.52 & \textbf{0.94}
& 0.15 & 0.06 & \textbf{0.94}
& 0.12 & 0.01 & \textbf{0.94} \\

\addlinespace
\multirow{3}{*}{205} & \multirow{3}{*}{200} & 10
& 1.13 & \textbf{1.05} & 1.18
& 0.99 & \textbf{1.00} & 0.94
& \textbf{0.88} & 0.66 & 0.88
& 0.85 & 0.09 & \textbf{0.88} \\

& & 50
& 1.55 & 1.41 & \textbf{1.30}
& 0.89 & 0.92 & \textbf{0.93}
& 0.56 & 0.18 & \textbf{0.89}
& 0.51 & 0.02 & \textbf{0.89} \\

& & 100
& 4.06 & 2.34 & \textbf{1.57}
& 0.38 & 0.61 & \textbf{0.92}
& 0.12 & 0.07 & \textbf{0.91}
& 0.09 & 0.01 & \textbf{0.91} \\

\addlinespace
\multirow{4}{*}{405} & \multirow{4}{*}{400} & 10
& 1.07 & \textbf{1.02} & 1.18
& \textbf{1.00} & 1.00 & 0.97
& \textbf{0.94} & 0.91 & 0.85
& \textbf{0.92} & 0.22 & 0.84 \\

& & 50
& 1.28 & \textbf{1.15} & 1.24
& 0.99 & \textbf{1.00} & 0.97
& 0.78 & 0.51 & \textbf{0.85}
& 0.72 & 0.06 & \textbf{0.85} \\

& & 100
& 1.56 & 1.40 & \textbf{1.33}
& 0.95 & \textbf{0.97} & 0.96
& 0.59 & 0.26 & \textbf{0.86}
& 0.52 & 0.03 & \textbf{0.86} \\

& & 200
& 5.05 & 2.37 & \textbf{1.60}
& 0.39 & 0.73 & \textbf{0.94}
& 0.09 & 0.10 & \textbf{0.88}
& 0.05 & 0.01 & \textbf{0.88} \\

\bottomrule
\end{tabular}
\end{table}

In terms of RMSE, the Horseshoe improves on Lasso across all reported
specifications, and SoU-SGL further reduces RMSE in most specifications.
To distinguish shrinkage from sparsity, we report (for null coordinates)
the proportion of estimated coefficients with $|\hat{\beta}_j|$ below
progressively tighter thresholds. This proportion remains high for SoU-SGL
even at the tightest threshold across all specifications, indicating strong
numerical sparsity. In contrast, the Horseshoe remains comparatively dense,
while the sparsity achieved by Lasso deteriorates as the number of signals
increases.

Figure~\ref{fig:reg_hist} complements Table~\ref{tab:experiment2} by showing the distribution of false positives. Conditional on $|\hat \beta |> 10^{-3}$, Lasso false positives are mostly near zero and become more dispersed as $s$ increases. SoU-SGL is more dichotomous: most nulls fall below the threshold, and the remaining false positives show little mass near zero 
and a comparatively more stable profile across \(s\).

\begin{figure}[h]
    \centering
    \includegraphics[width=1\textwidth]{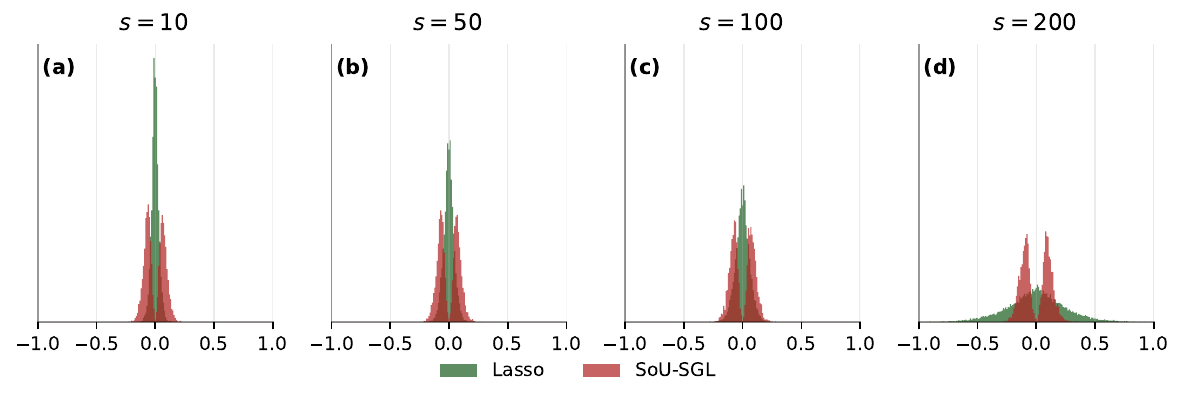} 
    \caption{\footnotesize{False-positive profiles in linear regression. Normalised histograms (densities) of $\hat \beta_j$ for null coordinates ($\beta_j =0$), conditional on $|\hat \beta_j |> 10^{-3}$, in the $(n,p)=(405,400)$ design. Panels correspond to increasing numbers of signals $s\in\{10,50,100,200\}$.}}
    \label{fig:reg_hist}
\end{figure}

\FloatBarrier
\subsection{SoU-SGL classification}\label{sec:logistic}

We next apply our framework to Bayesian logistic regression, benchmarking our SoU-SGL model against standard Logistic Regression (LR), Logistic Regression with an $\ell_1$ penalty (LR-L1) and Bayesian logistic regression with a global-local prior structure (GL).

We evaluate the proposed model on the MNIST dataset, a standard benchmark for image classification consisting of 70,000 greyscale images of handwritten digits (0-9), each of size $28 \times 28$ pixels. We use the standard 60,000-train/10,000-test split.

For each method, we train a one-vs-rest classifier for each digit, that is, a separate model is estimated for distinguishing digit $k$ from all others, for $k=0,\dots,9$. For LR-L1, the penalty parameter is chosen by cross-validation, separately for each one-vs-rest task, allowing task-specific regularisation. At prediction time, we assign each observation to the class corresponding to the model that yields the highest predicted probability.

Table~\ref{tab:resultados_accuracy} shows that all methods achieve similar accuracy, but sparsity differs sharply as the threshold $\epsilon$ tightens: for small $\epsilon$, LR and GL retain few coefficients below $\epsilon$, LR-L1 retains a moderate proportion, whereas SoU-SGL retains a large proportion.

\begin{table}[htbp]
\centering
\caption{MNIST logistic regression: accuracy and sparsity for the four considered models. Test accuracy (second column) and proportion of coefficients with $|\hat \beta_j|<\epsilon$ for multiple $\epsilon$ levels (last three columns) across Logistic Regression (LR), Logistic Regression with an $\ell_1$ penalty (LR-L1), Bayesian logistic regression with a global-local prior structure (GL), and our sources of uncertainty sparse global-local model (SoU-SGL).}
\label{tab:resultados_accuracy}
\footnotesize
\begin{tabular}{@{} l S[table-format=1.3] S[table-format=1.3] S[table-format=1.3] S[table-format=1.3] @{}}
\toprule
\multirow{2}{*}{Model} & \multicolumn{1}{c}{\multirow{2}{*}{Accuracy}} & \multicolumn{3}{c}{Sparsity: prop $|\hat \beta_j|<\epsilon$} \\
\cmidrule(l){3-5}
 &  & \multicolumn{1}{c}{$\epsilon = 10^{-1}$} & \multicolumn{1}{c}{$\epsilon = 10^{-2}$} & \multicolumn{1}{c}{$\epsilon = 10^{-3}$} \\
\midrule
LR & 0.917 & 0.57 & 0.17 & 0.10 \\
LR-L1 & 0.918 & 0.82 & 0.46 & 0.38 \\
GL & \textbf{0.921} & \textbf{0.88} & 0.28 & 0.10 \\
SoU-SGL (ours) & 0.919 & 0.84 & \textbf{0.82} & \textbf{0.82} \\
\bottomrule
\end{tabular}
\end{table}

Figure~\ref{fig:MnistCoefficients} visualises these differences. LR yields dense coefficient maps, while LR-L1 and GL reduce many coefficients but retain a diffuse background of small non-zero values. In contrast, SoU-SGL produces high-contrast maps: most pixel coefficients are driven towards virtually zero, and the remaining ones are visually higher in magnitude under the common colour scale.

\begin{figure}[h]
    \centering
    \includegraphics[width=1\textwidth]{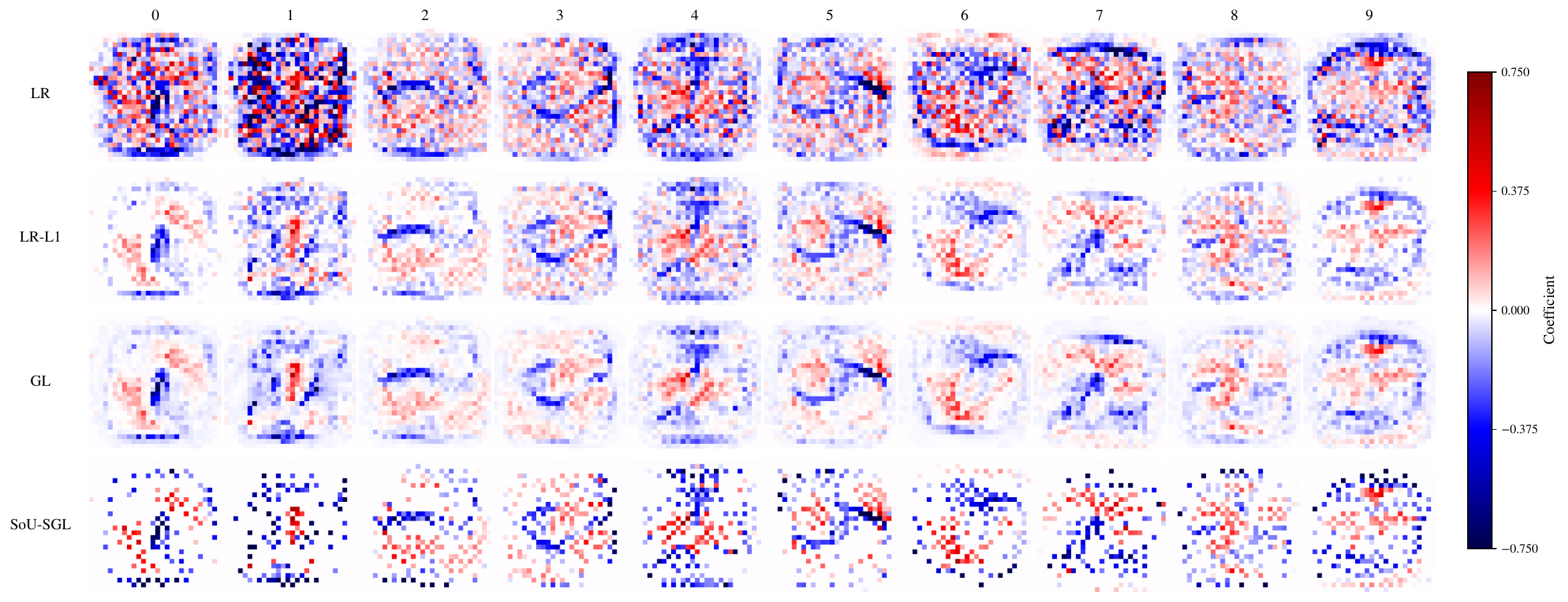} 
    \caption{\footnotesize MNIST logistic regression coefficient maps. Heatmaps of fitted pixel coefficients for digits 
$0$ to $9$ under LR, LR-L1, GL, and SoU-SGL, shown on a common colour scale.}
    \label{fig:MnistCoefficients}
\end{figure}

\subsection{SoU-SGL Bayesian neural networks}\label{sec:bnn}

We next apply our SoU-SGL framework to Bayesian neural networks (BNNs), where it induces self-pruning during training. We consider the horseshoe BNN of \citet{ghosh2018structured,ghosh2019model} and follow their implementation under a fully factorised variational approximation. In this construction, the global-local prior is placed on hidden units rather than on individual scalar weights. Specifically, unit $k$ in layer $\ell$ has a local scale $\tau_{k\ell}$ and shares a layer-specific global scale $\nu_\ell$, so that the product $\tau_{k\ell}\nu_\ell$ controls the magnitude of its incoming weight vector. Near-zero scales therefore correspond to pruned units. Our SoU-SGL model differs from the horseshoe BNN only by introducing a small confidence parameter on the local component ($\gamma_\tau$). We compare against a standard BNN and the corresponding horseshoe BNN (HS).\\

\noindent\textbf{Regression problem.} 
We first consider a toy regression problem modelled by a two-layer network under a Gaussian observation model with unknown noise precision $\phi$.
In this setting, aleatoric uncertainty is captured by the posterior over $\phi$, while epistemic uncertainty is captured by the posterior over the network weights. The data-generating process is $y = \sin(X) + \epsilon$, with $\epsilon \sim \mathcal{N}(0, 0.2^2)$.

Figure~\ref{fig:SoU_BNN_toy1D} shows that both BNN and HS struggle to recover the underlying signal in smaller-sample regimes or as network width increases. In contrast, SoU-SGL consistently recovers the true functional structure and produces uncertainty estimates better aligned with the variability of the data-generating process. Quantitative summaries of the uncertainty decomposition are reported in Section~\ref{app:BNN} of the Supplementary Material.

\begin{figure}[h]
    \centering
    \includegraphics[width=1.\textwidth]{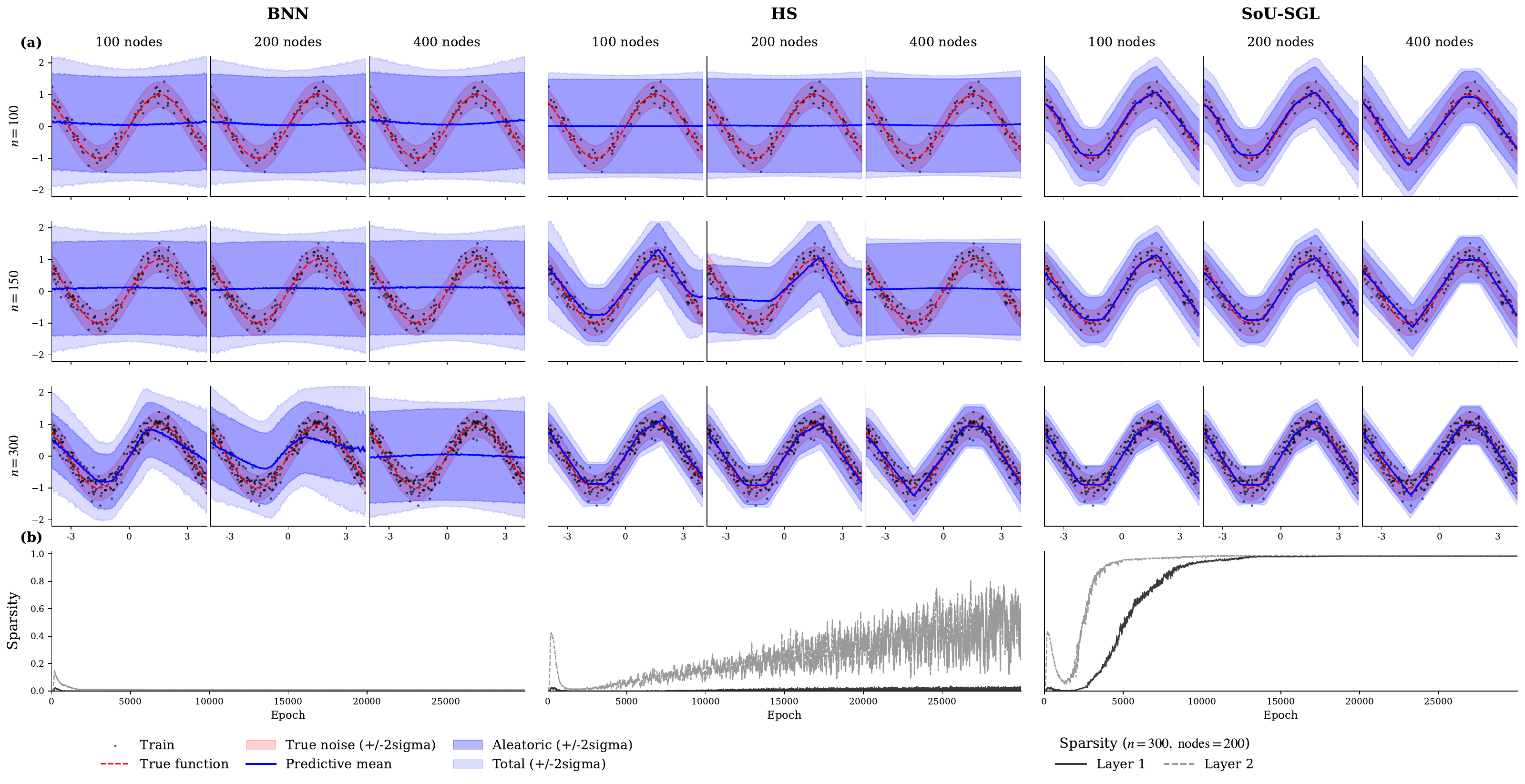}
    \caption{\footnotesize Toy regression with Bayesian neural networks: prediction and self-pruning. Top: predictive mean and uncertainty bands across sample sizes and network widths (BNN, Horseshoe, SoU-SGL). Bottom: layer-wise sparsity over training epochs for the configuration $n=300$ and 200 hidden units per layer, computed from effective weights using the threshold $|w|<10^{-5}$.}
    \label{fig:SoU_BNN_toy1D}
\end{figure}

Additionally, unlike BNN and HS, SoU-SGL consistently produces sparse neural networks across all configurations. Here sparsity is reported in terms of effective weights for comparability across models; in SoU-SGL, it is accompanied by pruning at the level of hidden units. Consistent with the simplicity of the data-generating process, the fitted SoU-SGL networks reduce to highly compact effective architectures across all specifications considered, with three active nodes in the first layer and typically two in the second.\\

\FloatBarrier
\noindent\textbf{Classification problem.} 
We next turn to MNIST for multi-class classification, using two-layer networks with 200 hidden units per layer. Results for other widths are qualitatively similar and not presented.

The top panel of Figure~\ref{fig:accuracy_sparsity_training_dynamics} shows that HS and SoU-SGL follow nearly identical accuracy trajectories, both consistently outperforming BNN. The bottom panels show that, for HS and BNN, sparsity stabilises once accuracy plateaus. In contrast, SoU-SGL exhibits a similar initial sparsification pattern to HS but continues to increase sparsity in both layers thereafter, while maintaining the same predictive performance.

\begin{figure}[h]
    \centering
    \includegraphics[width=1.\textwidth]{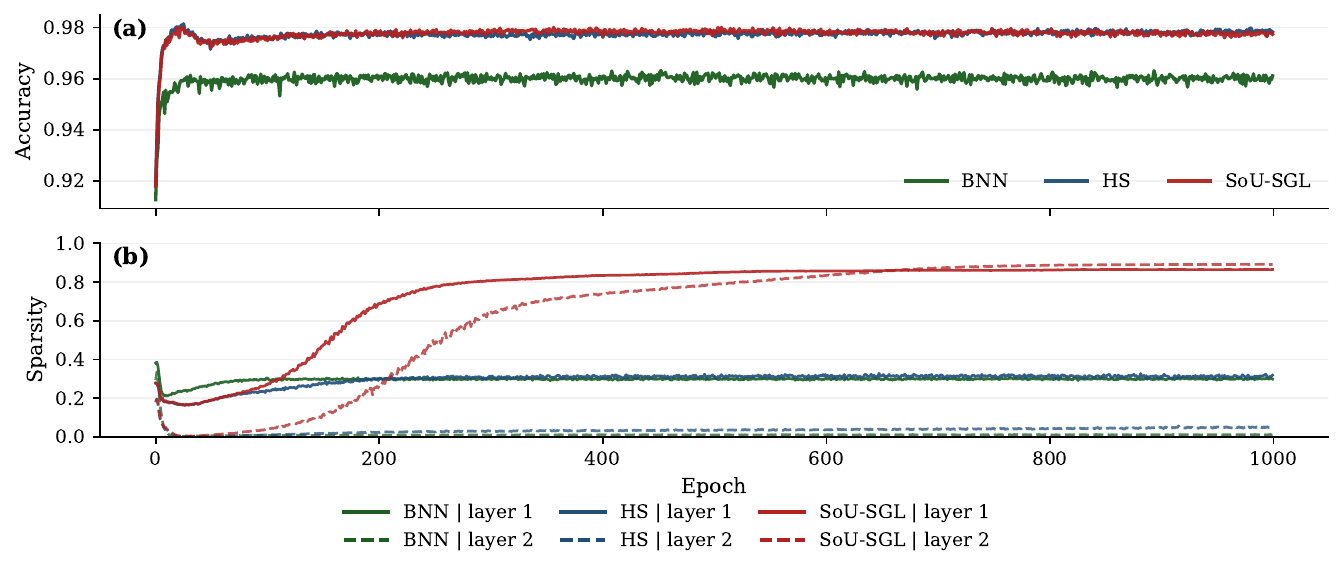}
    \caption{\footnotesize MNIST BNN training dynamics. Test accuracy (top) and layer-wise sparsity (bottom) over epochs for BNN, Horseshoe (HS), and SoU-SGL, where sparsity is computed from effective weights using the threshold $|w|<10^{-5}$.}
\label{fig:accuracy_sparsity_training_dynamics}
\end{figure}

\FloatBarrier
\noindent\textbf{Varying architecture.} To further assess the architectural effect of this pruning mechanism, we varied the initial network width and measured the number of hidden units that remained active at the end of training under SoU-SGL. Table~\ref{tab:mnist_table_b} shows that, although the initial architectures range from 50 to 1600 hidden units per layer, the fitted models retain only a small number of active nodes, with essentially unchanged test accuracy. The resulting numbers of active units are also similar across initial widths, suggesting convergence toward a common compact architecture.

\begin{table}[htbp]
\centering
\caption{MNIST BNNs under SoU-SGL: active hidden nodes and final test accuracy across initial network widths. Reported values are the numbers of active hidden units in each layer at the end of training and the corresponding final test accuracy. A hidden unit is counted as active when $\mathbb{E}_q[\tau_{k\ell}\nu_\ell] \geq 10^{-3}$.}
\label{tab:mnist_table_b}
\footnotesize
\begin{tabular}{@{} l r r r r r r r @{}}
\toprule
Initial width & 50 & 100 & 200 & 400 & 800 & 1200 & 1600 \\
\midrule
Layer 1 active & 26 & 26 & 32 & 30 & 36 & 33 & 31 \\
Layer 2 active & 29 & 32 & 24 & 20 & 19 & 15 & 15 \\
Accuracy & 0.978 & 0.977 & 0.977 & 0.977 & 0.979 & 0.976 & 0.977 \\
\bottomrule
\end{tabular}
\end{table}

In summary, SoU-SGL induces self-pruning in BNNs, progressively removing redundant weights and hidden units, and thereby transitions from over-parameterised initial architectures to leaner data-adapted networks without sacrificing predictive performance.

\FloatBarrier

\section{Concluding remarks}\label{sec:discussion}

This paper introduced a general framework that extends Bayesian inference by explicitly incorporating the researcher's confidence in each source of uncertainty.
This adds a new dimension to Bayesian modelling, allowing the researcher to
selectively control the probabilistic treatment of different model components
and how the uncertainty associated with them is propagated through the model.

From this perspective, SoU provides a common framework for familiar departures
from full Bayesian inference. 
Generalised Bayes modifies the probabilistic treatment globally through a common
confidence parameter, whereas, in the two-source setting, empirical Bayes and
profile likelihood arise as selective heterogeneous-confidence configurations
in which one source is driven towards optimisation while the other retains full
Bayesian treatment. 
SoU combines these two forms of flexibility by allowing the probabilistic
treatment to be specified separately for each source, and extends naturally
to multiple sources of uncertainty.

The SoU-SGL estimator illustrates how this source-specific flexibility can be
used for inferential design, providing a general approach for inducing sparsity.
By adopting a global-local prior structure and assigning near-zero confidence to local parameters, the method induces strong sparsity, driving many coefficients to negligible values rather than merely shrinking them. The procedure is flexible and 
can be applied across a broad range of settings 
where sparsity is desirable. In Bayesian neural networks, the approach yields automatic self-pruning, allowing the model to simplify itself while preserving predictive performance.

In summary, SoU introduces confidence as a distinct dimension of Bayesian
modelling, complementary to prior informativeness. Within this framework, the
confidence parameters may be specified directly to encode how each source should
be treated, or calibrated to achieve a desired inferential property.

\bibliographystyle{abbrvnat}
\bibliography{references}

@article{murphy2000profile,
  author  = {Murphy, Susan A. and Van Der Vaart, Aad W.},
  title   = {On Profile Likelihood},
  journal = {Journal of the American Statistical Association},
  year    = {2000},
  volume  = {95},
  number  = {450},
  pages   = {449--465},
}

@book{carlin2000bayes,
  author    = {Carlin, Bradley P. and Louis, Thomas A.},
  title     = {{Bayes} and Empirical {Bayes} Methods for Data Analysis},
  edition   = {2},
  year      = {2000},
  publisher = {Chapman \& Hall/CRC},
}

@article{lyddon2019general,
  title={General {Bayesian} updating and the loss-likelihood bootstrap},
  author={Lyddon, Simon P and Holmes, Chris C and Walker, Stephen G},
  journal={Biometrika},
  volume={106},
  number={2},
  pages={465--478},
  year={2019},
  publisher={Oxford University Press}
}

@article{epstein2003recursive,
  title={Recursive multiple-priors},
  author={Epstein, Larry G and Schneider, Martin},
  journal={Journal of Economic Theory},
  volume={113},
  number={1},
  pages={1--31},
  year={2003},
  publisher={Elsevier}
}

@article{klibanoff2009recursive,
  title={Recursive smooth ambiguity preferences},
  author={Klibanoff, Peter and Marinacci, Massimo and Mukerji, Sujoy},
  journal={Journal of Economic Theory},
  volume={144},
  number={3},
  pages={930--976},
  year={2009},
  publisher={Elsevier}
}

@inproceedings{ranganath2014black,
  title={Black box variational inference},
  author={Ranganath, Rajesh and Gerrish, Sean and Blei, David},
  booktitle={Artificial Intelligence and Statistics},
  year={2014},
}

@inproceedings{jaakkola1997variational,
  title={A variational approach to {Bayesian} logistic regression models and their extensions},
  author={Jaakkola, Tommi S and Jordan, Michael I},
  booktitle={Artificial Intelligence and Statistics},
  pages={283--294},
  year={1997},
  organization={PMLR}
}

@article{maccheroni2006ambiguity,
  title={Ambiguity aversion, robustness, and the variational representation of preferences},
  author={Maccheroni, Fabio and Marinacci, Massimo and Rustichini, Aldo},
  journal={Econometrica},
  volume={74},
  number={6},
  pages={1447--1498},
  year={2006},
  publisher={Wiley Online Library}
}

@article{diaconis1979conjugate,
  title={Conjugate priors for exponential families},
  author={Diaconis, Persi and Ylvisaker, Donald},
  journal={The Annals of Statistics},
  pages={269--281},
  year={1979},
  publisher={JSTOR}
}

@article{miller2019robust,
  title={Robust {Bayesian} inference via coarsening},
  author={Miller, Jeffrey W and Dunson, David B},
  journal={Journal of the American Statistical Association},
  year={2019},
volume = {114},
number = {527},
pages = {1113--1125},  
  publisher={Taylor \& Francis}
}

@article{holmes2017assigning,
  title={{Assigning a value to a power likelihood in a general {{Bayesian}} model}},
  author={Holmes, Chris C and Walker, Stephen G},
  journal={Biometrika},
  volume={104},
  number={2},
  pages={497--503},
  year={2017},
  publisher={Oxford University Press}
}

@article{grunwald2017inconsistency,
  title={Inconsistency of {{Bayesian}} inference for misspecified linear models, and a proposal for repairing it},
  author={Gr{\"u}nwald, Peter and Van Ommen, Thijs},
  journal = {{Bayesian} Analysis},
  year={2017},
  volume = {12}, 
  number = {4},
  pages = {1069–1103}
}

@inproceedings{grunwald2012safe,
  title={The safe {B}ayesian: learning the learning rate via the mixability gap},
  author={Gr{\"u}nwald, Peter},
  booktitle={International Conference on Algorithmic Learning Theory},
  year={2012},
}

@incollection{gilboa2016ambiguity,
  title={Ambiguity and the {{Bayesian}} paradigm},
  author={Gilboa, Itzhak and Marinacci, Massimo},
  booktitle={Readings in Formal Epistemology: Sourcebook},
  pages={385--439},
  year={2016},
  publisher={Springer}
}

@article{ghosh2019model,
  title={Model selection in {{Bayesian}} neural networks via horseshoe priors},
  author={Ghosh, Soumya and Yao, Jiayu and Doshi-Velez, Finale},
  journal={Journal of Machine Learning Research},
  volume={20},
  number={182},
  pages={1--46},
  year={2019}
}

@inproceedings{ghosh2018structured,
  title={Structured variational learning of {Bayesian} neural networks with horseshoe priors},
  author={Ghosh, Soumya and Yao, Jiayu and Doshi-Velez, Finale},
  booktitle={International Conference on Machine Learning},
  year={2018},
}

@article{khan2023bayesian,
  title={The {Bayesian} learning rule},
  author={Khan, Mohammad Emtiyaz and Rue, H{\aa}vard},
  journal={Journal of Machine Learning Research},
  volume={24},
  number={281},
  pages={1--46},
  year={2023}
}

@inproceedings{carvalho2009handling,
  title={Handling sparsity via the horseshoe},
  author={Carvalho, Carlos M and Polson, Nicholas G and Scott, James G},
  booktitle={Artificial Intelligence and Statistics},
  pages={73--80},
  year={2009},
  organization={PMLR}
}

@article{park2008bayesian,
  title={The {Bayesian} {Lasso}},
  author={Park, Trevor and Casella, George},
  journal={Journal of the American Statistical Association},
  volume={103},
  number={482},
  pages={681--686},
  year={2008},
  publisher={Taylor and Francis Ltd.}
}

@book{savage-54,
    Author = {Savage, L. J.},
    Title = {The Foundations of Statistics},
    Publisher = {John Wiley and Sons},
    Year = {1954}}

@article{keynes1921treatise,
  title={A Treatise on Probability},
  author={Keynes, John Maynard},
  year={1921},
  publisher={Courier Corporation}  
}

@article{maccheroni2006dynamic,
  title={Dynamic variational preferences},
  author={Maccheroni, Fabio and Marinacci, Massimo and Rustichini, Aldo},
  journal={Journal of Economic Theory},
  volume={128},
  number={1},
  pages={4--44},
  year={2006},
  publisher={Elsevier}
}

@article{fox2002ambiguity,
  title={Ambiguity aversion, comparative ignorance, and decision context},
  author={Fox, Craig R and Weber, Martin},
  journal={Organizational Behavior and Human Decision Processes},
  volume={88},
  number={1},
  pages={476--498},
  year={2002},
  publisher={Elsevier}
}

@article{kilka2001determines,
  title={What determines the shape of the probability weighting function under uncertainty?},
  author={Kilka, Michael and Weber, Martin},
  journal={Management Science},
  volume={47},
  number={12},
  pages={1712--1726},
  year={2001},
  publisher={INFORMS}
}

@article{de2010competence,
  title={Competence effects for choices involving gains and losses},
  author={de Lara Resende, Jos{\'e} Guilherme and Wu, George},
  journal={Journal of Risk and Uncertainty},
  volume={40},
  number={2},
  pages={109--132},
  year={2010},
  publisher={Springer}
}

@article{tversky1995weighing,
  title={Weighing risk and uncertainty.},
  author={Tversky, Amos and Fox, Craig R},
  journal={Psychological Review},
  volume={102},
  number={2},
  pages={269},
  year={1995},
  publisher={American Psychological Association}
}

@article{keppe1995judged,
  title={Judged knowledge and ambiguity aversion},
  author={Keppe, Hans-J{\"u}rgen and Weber, Martin},
  journal={Theory and Decision},
  volume={39},
  number={1},
  pages={51--77},
  year={1995},
  publisher={Springer}
}

@article{heath1991preference,
  title={Preference and belief: Ambiguity and competence in choice under uncertainty},
  author={Heath, Chip and Tversky, Amos},
  journal={Journal of Risk and Uncertainty},
  volume={4},
  number={1},
  pages={5--28},
  year={1991},
  publisher={Springer}
}

@article{abdellaoui2011rich,
  title={The rich domain of uncertainty: Source functions and their experimental implementation},
  author={Abdellaoui, Mohammed and Baillon, Aur{\'e}lien and Placido, Laetitia and Wakker, Peter P},
  journal={American Economic Review},
  volume={101},
  number={2},
  pages={695--723},
  year={2011}
}

@article{hansen2001robust,
  title={Robust control and model uncertainty},
  author={Hansen, Lars Peter and Sargent, Thomas J},
  journal={American Economic Review},
  volume={91},
  number={2},
  pages={60--66},
  year={2001},
  publisher={American Economic Association}
}

@article{strzalecki2011axiomatic,
  title={Axiomatic foundations of multiplier preferences},
  author={Strzalecki, Tomasz},
  journal={Econometrica},
  volume={79},
  number={1},
  pages={47--73},
  year={2011},
  publisher={Wiley Online Library}
}

@book{knight1921risk,
  title={Risk, uncertainty and profit},
  author={Knight, Frank Hyneman},
  volume={31},
  year={1921},
  publisher={Houghton Mifflin}
}

@article{knoblauch2022optimization,
  title={An optimization-centric view on {Bayes}' rule: Reviewing and generalizing variational inference},
  author={Knoblauch, Jeremias and Jewson, Jack and Damoulas, Theodoros},
  journal={Journal of Machine Learning Research},
  volume={23},
  number={132},
  pages={1--109},
  year={2022}
}

@article{bissiri2016general,
  title={A general framework for updating belief distributions},
  author={Bissiri, Pier Giovanni and Holmes, Chris C and Walker, Stephen G},
  journal={Journal of the Royal Statistical Society Series B: Statistical Methodology},
  volume={78},
  number={5},
  pages={1103--1130},
  year={2016},
  publisher={Oxford University Press}
}

\clearpage
\appendix

\setcounter{page}{1}
\renewcommand{\thepage}{S\arabic{page}}

\providecommand{\theHequation}{\theequation}
\renewcommand{\theHequation}{S.\arabic{equation}}

\bookmarksetup{startatroot}
\phantomsection
\pdfbookmark[0]{Supplementary Material}{supplementary-material}

\begin{center}
    {\Large Supplementary Material}
\end{center}

This supplementary material contains the experimental details, estimation
procedures, derivations, and proofs supporting the results in the main paper.
Section~\ref{app:expe} describes the model specifications, priors, variational
families, hyperparameter settings, and additional simulation results for the
normal-means, linear regression, logistic regression, and Bayesian neural
network experiments. Section~\ref{app:estimation} summarises the optimisation
procedures used in the numerical studies. Section~\ref{app:updates} provides
the derivations of the SoU objective,
the source-specific information-matching calibration, and
model-specific optimality conditions and
coordinate updates, while the proofs of the theoretical results stated in the
main text are reported in Section~\ref{app:proofs}.

\section{Experimental setup: models, priors, and additional results}\label{app:expe}

\subsection{Normal means}\label{app:nmeans_setup}
\subsubsection{Model, priors, and variational family}

We recall the Normal-means model with a global-local prior from Section~\ref{sec:simple_gl} (main paper). For each coordinate $i=1,\dots,I$:
\begin{align*}
y_{ij} \mid \beta_i, \sigma^2 &\sim \mathcal{N}(\beta_i, \sigma^2),  \quad \text{for } j = 1, \dots, n, \\
\beta_i \mid \tau_i, \nu &\sim \mathcal{N}(0, \tau_i^2 \nu^2), \\
\tau_i^2 \stackrel{\text{i.i.d.}}{\sim} \mathrm{Exp}(\lambda_\tau),\quad
\nu^2 &\sim \mathrm{Exp}(\lambda_\nu),\quad
\sigma^2 \sim \mathrm{Inv\text{-}Gamma}(a_{\pi,\sigma}, b_{\pi,\sigma}).
\end{align*}
We use a mean-field variational family of the form
\[
q(\boldsymbol{\beta},\tau^2,\nu^2,\sigma^2)
=
q(\nu^2)\,q(\sigma^2)\prod_{i=1}^I q(\beta_i)\,q(\tau_i^2),
\]
with
\begin{align*}
q(\beta_i)&=\mathcal{N}(m_i,s_i^2),\qquad
q(\sigma^2)=\mathrm{Inv\text{-}Gamma}(a_\sigma,b_\sigma),\\
q(\tau_i^2)&=\mathrm{Inv\text{-}Gamma}(a_{\tau_i},b_{\tau_i}),\qquad
q(\nu^2)=\mathrm{Inv\text{-}Gamma}(a_\nu,b_\nu).
\end{align*}
The model is estimated using Algorithm~\ref{alg:fixed_point}; the corresponding optimality
conditions are given in Section~\ref{app:sousgl_normalmeans}.

We use
$\lambda_\tau = 10$, $\lambda_\nu = 1000$, 
$a_{\pi,\sigma}=3$,
$b_{\pi,\sigma}=2$,
$\gamma_\beta=\gamma_\sigma=\gamma_\nu=1$,
and $\gamma_\tau=10^{-9}$.

For the horseshoe baseline, we keep the same likelihood and prior on $\sigma^2$, and replace the exponential
priors by half-Cauchy distributions on the scales:
\begin{align*}
\tau_i \stackrel{\text{i.i.d.}}{\sim} \mathrm{C}^+(0,1),\qquad
\nu \sim \mathrm{C}^+(0,1).
\end{align*}

\subsection{Linear regression}\label{app:linreg_setup}
\subsubsection{Model, priors, and variational family}

We consider the Gaussian linear model with a global-local prior,
\begin{align*}
y \mid \alpha, \mathbf{X}, \boldsymbol{\beta}, \sigma^2 &\sim \mathcal{N}(\alpha \mathbf{1}_n + \mathbf{X}\boldsymbol{\beta}, \sigma^2 \mathbf{I}_n), \\
\beta_j \mid \tau_j^2, \nu^2, \sigma^2 &\sim \mathcal{N}\!\left(0, \tau_j^2 \nu^2 \sigma^2\right), \qquad j=1,\dots,p,\\
\alpha \sim \mathcal{N}(0, \sigma_\alpha^2), \quad
\tau_j^2 \stackrel{\text{i.i.d.}}{\sim} \mathrm{Exp}(\lambda_\tau), &\quad
\nu^2 \sim \mathrm{Exp}(\lambda_\nu), \quad
\sigma^2 \sim \mathrm{Inv\text{-}Gamma}(a_{\pi,\sigma}, b_{\pi,\sigma}).
\end{align*}
Here, $\nu$ is a global shrinkage parameter and $\{\tau_j\}_{j=1}^p$ are local shrinkage parameters.
Setting $\nu \equiv 1$ recovers the Bayesian Lasso of \citet{park2008bayesian}.

We use a mean-field variational family of the form
\[
q(\alpha,\boldsymbol{\beta},\boldsymbol{\tau}^2,\nu^2,\sigma^2)
=
q(\alpha)q(\boldsymbol{\beta})q(\nu^2)q(\sigma^2)
\prod_{j=1}^p q(\tau_j^2),
\]
with
\begin{align*}
q(\alpha)&=N(m_\alpha,s_\alpha^2), \quad q(\boldsymbol{\beta})=\mathcal{N}(m,S),\quad
q(\sigma^2)=\mathrm{Inv\text{-}Gamma}(a_\sigma,b_\sigma),\quad \\
q(\tau_j^2)&=\mathrm{Inv\text{-}Gamma}(a_{\tau_j},b_{\tau_j}), \quad q(\nu^2)=\mathrm{Inv\text{-}Gamma}(a_\nu,b_\nu).\qquad
\end{align*}

The model is estimated using Algorithm~\ref{alg:fixed_point}; the corresponding optimality
conditions are given in Section~\ref{app:sousgl_linreg}.
In the simulations, $\mathbf{X}$ has i.i.d.\ standard normal entries and the data are generated without an intercept; accordingly, we fit all models with $\alpha \equiv 0$.

We use
$\lambda_\tau = 10$, $\lambda_\nu = 1000$, 
$a_{\pi,\sigma}=3$, 
$b_{\pi,\sigma}=2$, 
$\gamma_\alpha=\gamma_\beta=\gamma_\sigma=\gamma_\nu=1$, 
and $\gamma_\tau=10^{-9}$.

For the horseshoe prior, we keep the same likelihood and prior on $\sigma^2$, and replace the exponential
priors by half-Cauchy distributions on the scales:
\begin{align*}
\tau_j \stackrel{\text{i.i.d.}}{\sim} \mathrm{C}^+(0,1),\qquad
\nu \sim \mathrm{C}^+(0,1).
\end{align*}

\subsubsection{Additional results}

This section reports the remaining linear regression experiments of \citet{carvalho2009handling} omitted from the main text. We use the same design for $\mathbf{X}$, repeat each configuration over 100 simulations, and report averages over runs. In Experiment~2, we follow the main-text setup, but use the repeated non-zero block
$\beta_{1:10}=(2,2,2,2,2,2,2,2,5,20)$.

\begin{table}[htbp]
\centering
\caption{Linear regression simulation results (Experiment 2). Test-prediction RMSE and, for null coordinates ($\beta_j=0$), the proportion of estimated coefficients $|\hat\beta_j|$ below fixed thresholds, averaged over repetitions, comparing Lasso, Horseshoe, and SoU-SGL across $(n,p,s)$ configurations.}
\label{tab:experiment1}
\footnotesize
\setlength{\tabcolsep}{4.2pt}
\begin{tabular}{@{} rr r *{3}{S[table-format=1.2]} *{3}{S[table-format=1.2]} *{3}{S[table-format=1.2]} *{3}{S[table-format=1.2]} @{}}
\toprule
\multirow{2}{*}{$n$} & \multirow{2}{*}{$p$} & \multirow{2}{*}{$s$}
& \multicolumn{3}{c}{RMSE}
& \multicolumn{3}{c}{prop $<10^{-1}$}
& \multicolumn{3}{c}{prop $<10^{-2}$}
& \multicolumn{3}{c}{prop $<10^{-3}$} \\
\cmidrule(lr){4-6}\cmidrule(lr){7-9}\cmidrule(lr){10-12}\cmidrule(lr){13-15}
& & & Lasso & HS & SoU
& Lasso & HS & SoU
& Lasso & HS & SoU
& Lasso & HS & SoU \\
\midrule

24 & 20 & 10
& 1.94 & 1.80 & \textbf{1.51}
& 0.47 & 0.38 & \textbf{0.96}
& 0.31 & 0.04 & \textbf{0.96}
& 0.28 & 0.01 & \textbf{0.96} \\

\addlinespace
\multirow{2}{*}{60} & \multirow{2}{*}{50} & 10
& 1.34 & 1.25 & \textbf{1.18}
& 0.83 & 0.82 & \textbf{0.95}
& 0.60 & 0.13 & \textbf{0.95}
& 0.58 & 0.01 & \textbf{0.95} \\

& & 20
& 1.67 & 1.57 & \textbf{1.30}
& 0.63 & 0.58 & \textbf{0.95}
& 0.38 & 0.07 & \textbf{0.95}
& 0.34 & 0.01 & \textbf{0.95} \\

\addlinespace
\multirow{3}{*}{120} & \multirow{3}{*}{100} & 10
& 1.19 & \textbf{1.11} & 1.14
& 0.95 & \textbf{0.97} & 0.93
& 0.76 & 0.35 & \textbf{0.93}
& 0.73 & 0.03 & \textbf{0.93} \\

& & 20
& 1.33 & 1.25 & \textbf{1.19}
& 0.90 & 0.91 & \textbf{0.93}
& 0.64 & 0.19 & \textbf{0.93}
& 0.60 & 0.02 & \textbf{0.93} \\

& & 50
& 1.91 & 1.76 & \textbf{1.40}
& 0.59 & 0.63 & \textbf{0.94}
& 0.24 & 0.08 & \textbf{0.94}
& 0.20 & 0.01 & \textbf{0.94} \\

\addlinespace
\multirow{3}{*}{240} & \multirow{3}{*}{200} & 10
& 1.10 & \textbf{1.05} & 1.14
& 1.00 & \textbf{1.00} & 0.95
& 0.88 & 0.69 & \textbf{0.90}
& 0.85 & 0.10 & \textbf{0.90} \\

& & 50
& 1.40 & 1.32 & \textbf{1.24}
& 0.93 & \textbf{0.94} & 0.94
& 0.60 & 0.22 & \textbf{0.90}
& 0.54 & 0.02 & \textbf{0.90} \\

& & 100
& 1.91 & 1.79 & \textbf{1.41}
& 0.72 & 0.76 & \textbf{0.93}
& 0.27 & 0.10 & \textbf{0.92}
& 0.21 & 0.01 & \textbf{0.92} \\

\addlinespace
\multirow{4}{*}{480} & \multirow{4}{*}{400} & 10
& 1.07 & \textbf{1.03} & 1.15
& \textbf{1.00} & 1.00 & 0.98
& \textbf{0.95} & 0.92 & 0.86
& \textbf{0.92} & 0.24 & 0.86 \\

& & 50
& 1.21 & \textbf{1.13} & 1.20
& 1.00 & \textbf{1.00} & 0.98
& 0.79 & 0.55 & \textbf{0.87}
& 0.72 & 0.07 & \textbf{0.87} \\

& & 100
& 1.39 & 1.31 & \textbf{1.27}
& 0.98 & \textbf{0.98} & 0.97
& 0.62 & 0.30 & \textbf{0.87}
& 0.54 & 0.03 & \textbf{0.87} \\

& & 200
& 1.89 & 1.80 & \textbf{1.45}
& 0.85 & 0.87 & \textbf{0.96}
& 0.29 & 0.14 & \textbf{0.89}
& 0.20 & 0.01 & \textbf{0.89} \\

\bottomrule
\end{tabular}
\end{table}

The results are consistent with the main text: SoU-SGL achieves competitive RMSE while yielding 
markedly stronger sparsity recovery among null coordinates across all specifications.

In Experiment~3, we fix $n=55$ and $p=50$. We generate coefficients as $\beta_j=Z_jT_j$, where $T_j\sim t_2$ and $Z_j\sim\mathrm{Bernoulli}(\rho)$ independently, so that $\rho$ controls the expected signal proportion. We consider $\rho\in\{0.1,0.2,0.3,0.4,0.5\}$.

\begin{table}[htbp]
\centering
\caption{Linear regression simulation: Experiment 3 (heavy-tailed signals). Test-prediction RMSE and, for null coordinates ($\beta_j=0$), the proportion of estimated coefficients $|\hat\beta_j|$ below fixed thresholds, averaged over repetitions, comparing Lasso, Horseshoe, and SoU-SGL across $\rho\in\{0.1,0.2,0.3,0.4,0.5\}$.}
\label{tab:experiment3}
\footnotesize
\setlength{\tabcolsep}{4.2pt}
\begin{tabular}{@{} l *{3}{S[table-format=1.2]} *{3}{S[table-format=1.2]} *{3}{S[table-format=1.2]} *{3}{S[table-format=1.2]} @{}}
\toprule
\multirow{2}{*}{\text{prop.\ signals}}
& \multicolumn{3}{c}{RMSE}
& \multicolumn{3}{c}{prop $<10^{-1}$}
& \multicolumn{3}{c}{prop $<10^{-2}$}
& \multicolumn{3}{c}{prop $<10^{-3}$} \\
\cmidrule(lr){2-4}
\cmidrule(lr){5-7}
\cmidrule(lr){8-10}
\cmidrule(lr){11-13}
& Lasso & HS & SoU
& Lasso & HS & SoU
& Lasso & HS & SoU
& Lasso & HS & SoU \\
\midrule

0.1
& 1.18 & \textbf{1.11} & 1.15
& 0.92 & \textbf{0.94} & 0.94
& 0.79 & 0.32 & \textbf{0.93}
& 0.77 & 0.04 & \textbf{0.93} \\

0.2
& 1.36 & 1.25 & \textbf{1.25}
& 0.85 & 0.86 & \textbf{0.93}
& 0.68 & 0.18 & \textbf{0.93}
& 0.65 & 0.02 & \textbf{0.93} \\

0.3
& 1.52 & 1.41 & \textbf{1.37}
& 0.79 & 0.77 & \textbf{0.93}
& 0.58 & 0.13 & \textbf{0.93}
& 0.55 & 0.01 & \textbf{0.93} \\

0.4
& 1.69 & 1.58 & \textbf{1.57}
& 0.72 & 0.70 & \textbf{0.91}
& 0.51 & 0.10 & \textbf{0.91}
& 0.48 & 0.01 & \textbf{0.91} \\

0.5
& 1.83 & \textbf{1.73} & 1.78
& 0.65 & 0.63 & \textbf{0.90}
& 0.43 & 0.08 & \textbf{0.90}
& 0.41 & 0.01 & \textbf{0.90} \\

\bottomrule
\end{tabular}
\end{table}

In this heavy-tailed signal setting, SoU-SGL remains competitive in RMSE and consistently yields substantially higher sparsity among null coordinates across all $\rho$.

\subsection{Logistic regression}\label{app:logreg_setup}

\subsubsection{Model, priors, and variational family}

We consider Bayesian logistic regression with a global-local shrinkage prior:
\begin{align*}
y_i \mid \alpha,\mathbf{X}_i,\boldsymbol{\beta}
&\sim \mathrm{Bernoulli}\!\left(\sigma(\alpha+\mathbf{X}_i^\top\boldsymbol{\beta})\right),
\qquad i=1,\dots,n,\\
\beta_j \mid \tau_j^2,\nu^2
&\sim \mathcal{N}\!\left(0,\tau_j^2\nu^2\right), \qquad j=1,\dots,p,\\
\alpha \sim \mathcal{N}(0,\sigma_\alpha^2), \quad \tau_j^2 &\stackrel{\text{i.i.d.}}{\sim} \mathrm{Exp}(\lambda_\tau),\qquad
\nu^2 \sim \mathrm{Exp}(\lambda_\nu),
\end{align*}
where $\sigma(z)=1/(1+e^{-z})$. We use the mean-field variational family
\[
q(\alpha,\boldsymbol{\beta},\tau^2,\nu^2)
=
q(\alpha)\,q(\boldsymbol{\beta}) q(\nu^2)\prod_{j=1}^p q(\tau_j^2),
\]
with
\begin{align*}
q(\alpha)&=\mathcal{N}(m_\alpha,s_\alpha^2),\qquad
q(\boldsymbol{\beta})=\mathcal{N}(m,S),\\
q(\tau_j^2)&=\mathrm{Inv\text{-}Gamma}(a_{\tau_j},b_{\tau_j}),\qquad
q(\nu^2)=\mathrm{Inv\text{-}Gamma}(a_\nu,b_\nu).
\end{align*}

To obtain a tractable variational objective, we use the Jaakkola--Jordan quadratic bound
\citep{jaakkola1997variational},
introducing parameters $(\xi_i)_{i=1}^n$ associated with the linear predictors $\eta_i:=\alpha+\mathbf{X}_i^\top\boldsymbol{\beta}$,
\[
\log \sigma(\eta_i)
\;\ge\;
\frac{\eta_i-\xi_i}{2}
-\lambda(\xi_i)(\eta_i^2-\xi_i^2)
+\log\sigma(\xi_i),
\qquad
\lambda(\xi)=\frac{1}{2\xi}\Big(\sigma(\xi)-\frac12\Big),
\]
and updating $\xi_i$ in closed form as
\[
\xi_i \leftarrow \Big( (m_\alpha+\mathbf{X}_i^\top m)^2 + s_\alpha^2 + \mathbf{X}_i^\top S\mathbf{X}_i \Big)^{1/2},
\qquad i=1,\dots,n.
\]

The model is estimated using Algorithm~\ref{alg:fixed_point}; the corresponding optimality
conditions are given in Section~\ref{app:sousgl_logreg}.

Unless stated otherwise, we use
$\lambda_\tau = 10$, $\lambda_\nu = 1000$, 
$\sigma_\alpha=10$, 
$\gamma_\alpha=\gamma_\beta=\gamma_\nu=1$; we set $\gamma_\tau=10^{-6}$ for SoU-SGL and $\gamma_\tau=1$ for GL.

\FloatBarrier
\subsection{Bayesian neural networks}\label{app:BNN}

\subsubsection{Model, priors, and variational family}

We consider a Bayesian neural network (BNN) with $L$ layers. A network is parameterised by a collection of weight matrices $\mathbf{W} = \{ W_\ell \}_{\ell=1}^L$, where $W_\ell \in \mathbb{R}^{(K_{\ell-1}+1)\times K_\ell}$ denotes the weights connecting layer $\ell-1$ to layer $\ell$ (including bias terms).  
Given an input $\mathbf{X} \in \mathbb{R}^{K_0}$, the network produces the output
\[
f(\mathbf{W}; \mathbf{X}) = h_L(W_L^\top [z_{L-1}; 1]),
\]
where $z_0 = \mathbf{X}$, $z_\ell = h_\ell(W_\ell^\top [z_{\ell-1}; 1])$ for $\ell = 1,\dots,L-1$, and $h_\ell(\cdot)$ denotes an elementwise nonlinearity such as the ReLU function.  

A Bayesian neural network introduces uncertainty over the weights by placing prior distributions on them.  
Let $\mathcal{D} = \{ (\mathbf{X}_n, y_n) \}_{n=1}^N$ denote the observed data.  
As usual, the posterior over the weights is given by
\[
p(\mathbf{W} \mid \mathcal{D}) \propto p(\mathbf{W}) \prod_{n=1}^N p(y_n \mid f(\mathbf{W}; \mathbf{X}_n)),
\]
where $p(\mathbf{W})$ denotes the prior and $\prod_{n=1}^N p(y_n \mid f(\mathbf{W}; \mathbf{X}_n))$ the likelihood.

Let $\mathbf{w}_{k\ell} \in \mathbb{R}^{K_{\ell-1}+1}$ denote the vector of incoming weights to unit $k$ in layer $\ell$. 
The Normal prior that defines a standard BNN is:
\begin{align*}
\mathbf{w}_{k\ell} \sim \mathcal{N}(0,\sigma_{w,\ell}^2 I),
\end{align*}
The Horseshoe prior defines the following structure:
\begin{align*}
\mathbf{w}_{k\ell} \mid \tau_{k\ell}, \nu_\ell &\sim \mathcal{N}\left(0, \tau_{k\ell}^2 \nu_\ell^2 I\right), \\
\tau_{k\ell} \sim \mathcal{C}^+(0, b_0)&, \quad \nu_\ell \sim \mathcal{C}^+(0, b_g),
\end{align*}
where $\mathcal{C}^+(0, b)$ denotes the half-Cauchy distribution with scale $b$, $\tau_{k\ell}$ controls the degree of local shrinkage specific to unit $k$ in layer $\ell$, and $\nu_\ell$ acts as a global shrinkage parameter shared across the entire layer.

This is the structure considered in \citet{ghosh2018structured} and \citet{ghosh2019model}.
We closely follow their implementation and refer the reader to the original works for a comprehensive description.
For all models, we approximate the posterior distribution using a fully factorised variational family over all weights and scale parameters.

We use $\sigma_{w,\ell} = 1/\sqrt{\mathrm{dim}(\mathbf{w}_{k\ell})}$ for BNN and $b_0=1$, and $b_g=0.1$ for HS and SoU-SGL. We set $\gamma_\tau=10^{-9}$ for SoU-SGL and $\gamma_\tau=1$ for HS.\\

\noindent\textbf{Node-level and weight-level sparsity.}
The horseshoe prior admits the equivalent non-centred representation
\[
\mathbf{w}_{k\ell}=(\tau_{k\ell}\nu_\ell)\boldsymbol{\beta}_{k\ell},
\qquad
\boldsymbol{\beta}_{k\ell}\sim\mathcal N(0,I).
\]
Thus, the same multiplicative scale $\tau_{k\ell}\nu_\ell$ acts on all incoming weights of unit $k$ in layer $\ell$. This makes the shrinkage mechanism node-specific: when $\tau_{k\ell}\nu_\ell$ is close to zero, the full vector $\mathbf{w}_{k\ell}$ is jointly shrunk toward zero, and the corresponding hidden unit is pruned. We count a unit as inactive when $\mathbb{E}_q[\tau_{k\ell}\nu_\ell]$ falls below a fixed threshold.

We also consider sparsity directly at the level of individual weights. Specifically, we define the effective scalar weights as the variational posterior means $w_{j,k\ell}^{\mathrm{eff}} := \mathbb{E}_q[w_{j,k\ell}]$, where $w_{j,k\ell}$ denotes the $j$-th component of the incoming weight vector $\mathbf{w}_{k\ell}$. A scalar weight is counted as inactive when its absolute effective value falls below a fixed threshold. We use this notion in the main-text figures because it is directly comparable across BNN, HS, and SoU-SGL.

\subsubsection{Additional results}\label{app:bnn_additional}

This subsection complements the BNN toy regression experiment in Section~\ref{sec:bnn} (main paper) by reporting additional quantitative summaries. Table~\ref{tab:resultados_uncertainty} reports the empirical decomposition of predictive uncertainty into aleatoric and epistemic components, together with layerwise sparsity at the end of training (observation samples $n=300$). 
Across network widths, SoU-SGL yields tighter uncertainty estimates and consistently produces highly sparse networks.

\begin{table}[htbp]
\centering
\caption{Toy regression BNNs ($n=300$): uncertainty decomposition and sparsity. Estimated aleatoric and epistemic uncertainty and layer-wise sparsity (threshold 
$\epsilon=10^{-5}$) for BNN, Horseshoe (HS), and SoU-SGL across network widths (number of nodes).}
\label{tab:resultados_uncertainty}
\footnotesize
\begin{tabular}{@{} c S[table-format=1.3] S[table-format=1.3] S[table-format=1.3] S[table-format=1.3] S[table-format=1.3] S[table-format=1.3] S[table-format=1.3] S[table-format=1.3] S[table-format=1.3] @{}}
\toprule
\multirow{2}{*}{Metric} & \multicolumn{3}{c}{100 nodes} & \multicolumn{3}{c}{200 nodes} & \multicolumn{3}{c}{400 nodes} \\
\cmidrule(lr){2-4} \cmidrule(lr){5-7} \cmidrule(lr){8-10}
 & BNN & HS & SoU & BNN & HS & SoU & BNN & HS & SoU \\
\midrule
Aleatoric uncertainty & 0.431 & 0.305 & \textbf{0.291} & 0.561 & 0.306 & \textbf{0.289} & 0.719 & 0.309 & \textbf{0.292} \\
Epistemic uncertainty & 0.238 & 0.104 & \textbf{0.064} & 0.270 & 0.102 & \textbf{0.066} & 0.147 & 0.095 & \textbf{0.057} \\
Layer 1 weight sparsity & 0.000 & 0.005 & \textbf{0.970} & 0.000 & 0.010 & \textbf{0.985} & 0.001 & 0.004 & \textbf{0.993} \\
Layer 2 weight sparsity & 0.003 & 0.235 & \textbf{0.980} & 0.007 & 0.386 & \textbf{0.990} & 0.015 & 0.962 & \textbf{0.997} \\
\bottomrule
\end{tabular}
\end{table}

\FloatBarrier

\section{Estimation}\label{app:estimation}

We describe the optimisation procedures used to solve the SoU objective from the main paper (Section~\ref{sec:sou-bayes}).

For the Normal-means model and the linear and logistic regression models in Sections~\ref{sec:simple_gl}--\ref{sec:logistic} (main paper),
the objective yields explicit coordinate-wise optimality conditions. In these cases we use a
simple fixed-point scheme: we derive the first-order conditions and iterate the resulting
coordinate updates until convergence (Algorithm~\ref{alg:fixed_point}). When a coordinate update
does not admit a closed form, we solve the corresponding optimality equation numerically (e.g.,
via Newton--Raphson).

More generally, the SoU objective can be
optimised using black-box variational methods (BBVI) \citep{ranganath2014black}. In particular, the Black-Box Generalised
Variational Inference (BBGVI) algorithm of \citet{knoblauch2022optimization} applies directly 
to this class of objectives. For the Bayesian neural networks in Section~\ref{sec:bnn}, we follow
a standard BBVI procedure as in \citet{ghosh2018structured,ghosh2019model}. 
The SoU formulation leaves the variational family and optimisation scheme unchanged and has essentially the same per-iteration computational cost as the corresponding horseshoe BNN.

\begin{algorithm}[htbp]
\caption{Coordinate fixed-point algorithm for minimising $J(\btheta)$}
\label{alg:fixed_point}
\begin{algorithmic}[1]

\Require Objective $J(\btheta)$, tolerance $\epsilon$, maximum iterations $T$
\Ensure Variational parameters $\btheta^{*}$

\State \textbf{Coordinate updates:}
\State Derive the first-order conditions $\nabla_\btheta J(\btheta)=0$ and, for each component $\theta_i$, obtain either:
\State \hspace{1em} (i) an explicit update map $\theta_i \leftarrow f_i(\theta_{-i})$, or
\State \hspace{1em} (ii) a one-dimensional equation $g_i(\theta_i;\theta_{-i})=0$ to be solved numerically

\State \textbf{Initialise:} choose $\btheta^{(0)}$

\For{$t = 1, \dots, T$}
    \State $\btheta^{(t)} \gets \btheta^{(t-1)}$
    \For{each component $\theta_i$}
        \If{explicit update available}
            \State $\theta_i^{(t)} \gets f_i(\theta_{-i}^{(t)})$
        \Else
            \State $\theta_i^{(t)} \gets \mathrm{Solve}(g_i(\cdot;\theta_{-i}^{(t)})=0)$ \Comment{e.g., Newton}
        \EndIf
    \EndFor
    \If{$\max_i |\theta_i^{(t)}-\theta_i^{(t-1)}| < \epsilon$}
        \State \Return $\btheta^{*} = \btheta^{(t)}$
    \EndIf
\EndFor
\end{algorithmic}
\end{algorithm}

\FloatBarrier

\section{Derivations and model-specific updates}\label{app:updates}

\subsection{Rewriting the optimisation problem}\label{app:point_uncertainty}

We recall the objective functional from the main text
\begin{multline} \label{eq:J_main}
J(q)=
\frac{1}{n}\mathbb{E}_{q} [ L_n(\btheta) ] + \frac{\gamma_1}{n} \, 
\text{KL}\left( 
q(\btheta_1) \,\|\, \pi(\btheta_1) 
\right)+ 
\sum_{k=2}^{K} \frac{\gamma_k}{n} \, 
\mathbb{E}_{q(\btheta_{<k})} \left[ 
\text{KL}\left( 
q(\btheta_k | \btheta_{<k}) \,\|\, 
\pi(\btheta_k | \btheta_{<k}) 
\right)
\right],
\end{multline}
so that $q^*=\arg\min_q J(q)$. Since multiplying the objective by \(n\) does not change its minimisers, we work with the unscaled objective $\widetilde J(q) := nJ(q)$.

We first rewrite the expected loss using the log--exp identity

\[
\mathbb{E}_q\!\left[\sum_{i=1}^{n}\ell(\btheta,X_i)\right]
=
\mathbb{E}_q\!\left[-\log\exp\!\left(-\sum_{i=1}^{n}\ell(\btheta,X_i)\right)\right].
\]

Next, for any $q\ll\pi$,
\begin{equation}\label{eq:kl_decomp_general}
\mathrm{KL}(q\|\pi)
=
\mathbb{E}_q[\log q]-\mathbb{E}_q[\log \pi]
=
-\mathcal{H}(q)-\mathbb{E}_q[\log \pi],
\qquad
\mathcal{H}(q):=-\mathbb{E}_q[\log q].
\end{equation}
Applying \eqref{eq:kl_decomp_general} to the marginal and conditional KL terms in \eqref{eq:J_main}, the cross-entropy terms combine with the expected loss to form $M(q)$, while the entropy terms collect into $H(q)$, yielding the decomposition
\begin{equation}\label{eq:Jq_rewritten}
\widetilde J(q) = M(q) - H(q) - \log Z,
\end{equation}
where $M(q)$ and $H(q)$ are as defined in Section~\ref{sec:point_uncertainty} (main paper).

\subsection{SoU information matching}
\label{app:information_matching}

Section~\ref{sec:choosing_confidence} extends the information-matching criterion
of \citet{lyddon2019general} to heterogeneous confidence by matching information
blockwise. Here we derive this result, first for the two-block case and then
recursively for $K$ blocks.

\paragraph{Two-block SoU.}

Consider a two-block model,
$\boldsymbol{\theta}
=
(\boldsymbol{\theta}_1,\boldsymbol{\theta}_2)$, with confidence parameters
$(\gamma_1,\gamma_2)$. 
Let $\widehat{\btheta}$ denote the empirical risk minimiser, and set
$I:=I(\btheta_0)$ and $J:=J(\btheta_0)$. We write $J$ in block form as
\[
J=
\begin{pmatrix}
J_{11} & J_{12}\\
J_{21} & J_{22}
\end{pmatrix}.
\]

Locally on the $n^{-1/2}$ scale, the empirical loss admits the quadratic approximation

\begin{equation}
L_n(\boldsymbol{\theta})
\approx
L_n(\widehat{\boldsymbol{\theta}})
+
\frac{n}{2}
(\boldsymbol{\theta}-\widehat{\boldsymbol{\theta}})^\top
J(\boldsymbol{\theta}_0)
(\boldsymbol{\theta}-\widehat{\boldsymbol{\theta}}).
\end{equation}

For two blocks, the SoU recursion is
\begin{equation}
q_{\gamma_2}
(d\boldsymbol{\theta}_2
\mid\boldsymbol{\theta}_1,\bX)
\propto
\exp
\left\{
-\frac{
L_n(\boldsymbol{\theta}_1,\boldsymbol{\theta}_2)
}{\gamma_2}
\right\}
\pi_2
(d\boldsymbol{\theta}_2\mid\boldsymbol{\theta}_1),
\end{equation}
with effective loss
\begin{equation}
V_1(\boldsymbol{\theta}_1)
=
-\gamma_2
\log
\int
\exp
\left\{
-\frac{
L_n(\boldsymbol{\theta}_1,\boldsymbol{\theta}_2)
}{\gamma_2}
\right\}
\pi_2
(d\boldsymbol{\theta}_2\mid\boldsymbol{\theta}_1),
\end{equation}
and
\begin{equation}
q_{\gamma_1,\gamma_2}
(d\boldsymbol{\theta}_1\mid \bX)
\propto
\exp
\left\{
-\frac{V_1(\boldsymbol{\theta}_1)}{\gamma_1}
\right\}
\pi_1(d\boldsymbol{\theta}_1).
\end{equation}

Let
\[
\boldsymbol{\delta}_k
=
\boldsymbol{\theta}_k
-
\widehat{\boldsymbol{\theta}}_k.
\]
The quadratic term can be written as
\begin{align}
&
\boldsymbol{\delta}_1^\top J_{11}\boldsymbol{\delta}_1
+
2\boldsymbol{\delta}_1^\top J_{12}\boldsymbol{\delta}_2
+
\boldsymbol{\delta}_2^\top J_{22}\boldsymbol{\delta}_2
=
\left(
\boldsymbol{\delta}_2
+
J_{22}^{-1}J_{21}\boldsymbol{\delta}_1
\right)^\top
J_{22}
\left(
\boldsymbol{\delta}_2
+
J_{22}^{-1}J_{21}\boldsymbol{\delta}_1
\right)
+
\boldsymbol{\delta}_1^\top D_1\boldsymbol{\delta}_1,
\end{align}
where
\begin{equation}
D_1
=
J_{11}
-
J_{12}J_{22}^{-1}J_{21}.
\end{equation}

Hence, for locally positive and smooth priors,
\begin{equation}
q_{\gamma_2}
(\boldsymbol{\theta}_2
\mid\boldsymbol{\theta}_1,\bX)
\approx
\mathcal{N}
\left(
\widehat{\boldsymbol{\theta}}_2
-
J_{22}^{-1}J_{21}
(\boldsymbol{\theta}_1-\widehat{\boldsymbol{\theta}}_1),
\;
\frac{\gamma_2}{n}J_{22}^{-1}
\right).
\end{equation}

Integrating the second block in the SoU recursion gives the local effective
loss
\begin{equation}
V_1(\boldsymbol{\theta}_1)
\approx
C_n
+
\frac{n}{2}
(\boldsymbol{\theta}_1-\widehat{\boldsymbol{\theta}}_1)^\top
D_1
(\boldsymbol{\theta}_1-\widehat{\boldsymbol{\theta}}_1).
\end{equation}

where $C_n$ does not affect the local distribution of the first block.
Therefore,
\begin{equation}
q_{\gamma_1,\gamma_2}
(\boldsymbol{\theta}_1\mid \bX)
\approx
\mathcal{N}
\left(
\widehat{\boldsymbol{\theta}}_1,
\;
\frac{\gamma_1}{n}D_1^{-1}
\right).
\end{equation}

Equivalently, for the corresponding $\sqrt{n}$-scaled factors, the covariance
matrices are approximated by
\begin{equation}
\Sigma^{\mathrm{SoU}}_1
\approx
\gamma_1 D_1^{-1},
\qquad
\Sigma^{\mathrm{SoU}}_{2\mid1}
\approx
\gamma_2 J_{22}^{-1}.
\end{equation}

We can now apply the information-matching principle factor by factor. Recall
that the loss-likelihood bootstrap has asymptotic covariance
\begin{equation}
\Sigma_{\mathrm{LL}}
=
J^{-1}IJ^{-1}.
\end{equation}
Partitioning $\Sigma_{\mathrm{LL}}$ conformably with
$(\boldsymbol{\theta}_1,\boldsymbol{\theta}_2)$,
\begin{equation}
\Sigma_{\mathrm{LL}}
=
\begin{pmatrix}
\Sigma_{11} & \Sigma_{12}\\
\Sigma_{21} & \Sigma_{22}
\end{pmatrix},
\end{equation}
its Gaussian factorisation in the same block order has covariance
$\Sigma_{11}$ for the first block and conditional covariance
\begin{equation}
\Sigma_{22\mid1}
=
\Sigma_{22}
-
\Sigma_{21}\Sigma_{11}^{-1}\Sigma_{12}
\end{equation}
for the second.

This suggests extending the information-matching criterion of
\citet{lyddon2019general} to each SoU factor,
\begin{equation}
K\!\left(\gamma_1D_1^{-1}\right)
=
K(\Sigma_{11}),
\qquad
K\!\left(\gamma_2J_{22}^{-1}\right)
=
K(\Sigma_{22\mid1}).
\end{equation}
Since $K(\Sigma)=\operatorname{tr}(\Sigma^{-1})$, the corresponding
source-specific confidence parameters are
\begin{equation}
\gamma_1^*
=
\frac{\operatorname{tr}(D_1)}
{\operatorname{tr}(\Sigma_{11}^{-1})},
\qquad
\gamma_2^*
=
\frac{\operatorname{tr}(J_{22})}
{\operatorname{tr}(\Sigma_{22\mid1}^{-1})}.
\end{equation}

\paragraph{Multiple-block SoU.}

The two-block argument extends recursively. Consider
\begin{equation}
\boldsymbol{\theta}
=
(\boldsymbol{\theta}_1,\ldots,\boldsymbol{\theta}_K),
\qquad
\boldsymbol{\gamma}
=
(\gamma_1,\ldots,\gamma_K),
\end{equation}
and let $J^{(K)}=J$. For $k=K,\ldots,2$, define
\begin{equation}
D_k
=
J^{(k)}_{kk},
\end{equation}
and
\begin{equation}
J^{(k-1)}
=
J^{(k)}_{<k,<k}
-
J^{(k)}_{<k,k}
D_k^{-1}
J^{(k)}_{k,<k},
\end{equation}
with
\begin{equation}
D_1=J^{(1)}.
\end{equation}
Thus, $D_k$ is the local curvature of block $k$ after the
subsequent blocks have been integrated through the SoU recursion.

Repeated application of the two-block argument gives, for the corresponding
$\sqrt{n}$-scaled factors,
\begin{equation}
\operatorname{Cov}_{\mathrm{SoU}}
\left(
\boldsymbol{\theta}_k
\mid
\boldsymbol{\theta}_{<k},\bX
\right)
\approx
\gamma_k D_k^{-1},
\qquad
k=1,\ldots,K,
\end{equation}
where the conditioning is absent for $k=1$.

The loss-likelihood bootstrap reference remains
\begin{equation}
\Sigma_{\mathrm{LL}}
=
J^{-1}IJ^{-1}.
\end{equation}
Factorising the corresponding Gaussian distribution in the same block order,
define
\begin{equation}
V_1
=
(\Sigma_{\mathrm{LL}})_{11},
\end{equation}
and, for $k\geq2$,
\begin{equation}
V_k
=
(\Sigma_{\mathrm{LL}})_{kk}
-
(\Sigma_{\mathrm{LL}})_{k,<k}
(\Sigma_{\mathrm{LL}})_{<k,<k}^{-1}
(\Sigma_{\mathrm{LL}})_{<k,k}.
\end{equation}
Here $V_k$ is the conditional covariance of block $k$ given the preceding
blocks in the Gaussian reference distribution.

Applying the information-matching criterion factor by factor,
\begin{equation}
K\!\left(\gamma_k D_k^{-1}\right)
=
K(V_k),
\qquad
k=1,\ldots,K,
\end{equation}
gives
\begin{equation}
\gamma_k^*
=
\frac{\operatorname{tr}(D_k)}
{\operatorname{tr}(V_k^{-1})},
\qquad
k=1,\ldots,K.
\end{equation}

Hence, the single global calibration parameter of generalised Bayes is
replaced by source-specific confidence parameters, each matching the
asymptotic information of the corresponding factor in the SoU decomposition.

\subsection{Exponential family: tempered conjugate update}\label{app:Exponential}

Consider a single block $\btheta$ and the SoU objective with negative log-likelihood loss and KL weight $\gamma>0$:
\[
J_n(q)
=\E_q\!\left[-\frac{1}{n}\sum_{i=1}^n \log p(X_i\mid\btheta)\right]
+\frac{\bgamma}{n}\text{KL}\!\big(q(\btheta)\,\|\,\pi(\btheta)\big).
\]
Up to an additive constant (independent of $q$), we have
\[
\frac{n}{\bgamma}J_n(q)
=
\text{KL}\!\big(q(\btheta)\,\|\,p_\gamma(\btheta\mid \bX_{1:n})\big),
\]
where the \emph{tempered posterior} is defined by
\[
p_\gamma(\btheta\mid \bX_{1:n})
\propto
\pi(\btheta)\prod_{i=1}^n p(X_i\mid\btheta)^{1/\gamma}.
\]
Hence the minimiser satisfies $q^*(\btheta)=p_\gamma(\btheta\mid \bX_{1:n})$.

Now assume the exponential-family likelihood and conjugate prior described in Section~\ref{sec:examples_confidence} (main paper):
\[
p(X\mid\btheta)=h(X)\,c(\btheta)\exp\{\btheta^\top T(X)\},
\qquad
\pi(\btheta)\propto c(\btheta)^{N_0}\exp\{\btheta^\top(N_0T_0)\}.
\]
Then
\[
q^*(\btheta)
\propto
c(\btheta)^{\,N_0+n/\gamma}\,
\exp\!\Big\{
\btheta^\top\!\Big(N_0T_0+\frac{1}{\bgamma}\sum_{i=1}^n T(X_i)\Big)
\Big\},
\]
so $q^*$ remains in the conjugate family with updated natural parameters
\[
N_n:=N_0+\frac{n}{\bgamma},
\qquad
N_n T_n := N_0T_0+\frac{1}{\bgamma}\sum_{i=1}^n T(X_i).
\]
Equivalently,
\[
T_n=\alpha_n T_0+(1-\alpha_n)\bar T_n,
\qquad
\alpha_n=\frac{\gamma N_0}{\gamma N_0+n},
\qquad
\bar T_n=\frac{1}{n}\sum_{i=1}^n T(X_i).
\]

\subsection{Normal-normal model: closed-form SoU posterior}\label{app:normal_normal}

We revisit the normal-normal example from Section~\ref{sec:examples_confidence} (main paper).
Let $y_1,\dots,y_n$ be i.i.d.\ with $y_i\mid \mu\sim \mathcal N(\mu,\sigma^2)$ for known $\sigma^2$, and prior
$\mu\sim\mathcal N(\mu_0,\sigma_0^2)$. We take a Gaussian variational family $q(\mu)=\mathcal N(\mu_q,\sigma_q^2)$. The SoU objective is then:
\begin{multline} \label{eq:mu_J}
J(q)=
\frac{1}{2}\log(2\pi\sigma^2)
+\frac{1}{2\sigma^2}
\left[
    \sigma_q^2
    +
    \frac{1}{n}\sum_{i=1}^n (y_i-\mu_q)^2
\right] 
 + \frac{\gamma_\mu}{2n}
\left[
\frac{(\mu_0-\mu_q)^2}{\sigma_0^2}
+
\frac{\sigma_q^2}{\sigma_0^2}
-
\log\!\left(\frac{\sigma_q^2}{\sigma_0^2}\right)
-
1
\right]. \nonumber
\end{multline}

The first-order conditions give the closed-form solution

\begin{align}
\sigma_q^2
&=\left(\frac{n}{\gamma_\mu\sigma^2}+\frac{1}{\sigma_0^2}\right)^{-1}, 
\nonumber\\
\mu_q
&=\sigma_q^2\left(\frac{n}{\gamma_\mu\sigma^2}\,\bar y+\frac{1}{\sigma_0^2}\,\mu_0\right)
=\alpha_{\mathrm{SoU}}\mu_0+(1-\alpha_{\mathrm{SoU}})\bar y, 
\nonumber
\end{align}
where 
\[\bar y = \frac{1}{n}\sum_{i=1}^n y_i, \quad \text{and} \quad
\alpha_{\mathrm{SoU}}
=\frac{\frac{1}{\sigma_0^2}}{\frac{1}{\sigma_0^2}+\frac{n}{\gamma_\mu\sigma^2}}.
\]

\subsection{SoU-SGL models optimality conditions}\label{app:sousgl_optconds}

We provide, for each SoU-SGL model, the objective function and the corresponding first-order optimality conditions. These equations define the coordinate-wise updates iterated in Algorithm~\ref{alg:fixed_point} until convergence.

\subsubsection{SoU-SGL normal means}\label{app:sousgl_normalmeans}

For the Normal-means model in Section~\ref{app:nmeans_setup}, the SoU variational objective can be written as

\begin{align*}
J  =
& \frac{In}{2} \left[ \log(2\pi) + \log b_\sigma - \psi(a_\sigma) \right] + \frac{a_\sigma}{2 b_\sigma}\sum_{i=1}^I \sum_{j=1}^n\left[(y_{ij}-m_i)^2 + s_i^2\right] \\
& + \frac{\gamma_{\beta}}{2} \sum_{i=1}^I\left[\frac{a_\nu a_{\tau_i}(m_i^2 + s_i^2)}{b_\nu b_{\tau_i}} - \log s_i^2 + (\log b_\nu - \psi(a_\nu)) + (\log b_{\tau_i}-\psi(a_{\tau_i}))  - 1\right] \\
& + \gamma_{\tau} \sum_{i=1}^I \left[ a_{\tau_i}\log b_{\tau_i} - \log\Gamma(a_{\tau_i}) - (a_{\tau_i}+1)(\log b_{\tau_i} - \psi(a_{\tau_i})) - a_{\tau_i} - \log\lambda_\tau + \frac{\lambda_\tau b_{\tau_i}}{a_{\tau_i}-1} \right] \\
& + \gamma_\nu \left[ a_\nu \log b_\nu - \log\Gamma(a_\nu) - (a_\nu+1)(\log b_\nu - \psi(a_\nu)) - a_\nu - \log\lambda_\nu + \frac{\lambda_\nu b_\nu}{a_\nu-1} \right] \\
& + \gamma_\sigma \left[ a_\sigma \log b_\sigma - a_{\pi,\sigma} \log b_{\pi,\sigma} + \log \left( \frac{\Gamma(a_{\pi,\sigma})}{\Gamma(a_\sigma)} \right) \right. \\
& \quad \left. + (\log b_\sigma - \psi(a_\sigma))(a_{\pi,\sigma} - a_\sigma) + \left(\frac{a_\sigma}{b_\sigma}\right)(b_{\pi,\sigma} - b_\sigma) \right].
\end{align*}

Let $\bar y_i := n^{-1}\sum_{j=1}^n y_{ij}$. The first-order conditions yield the
following explicit updates for $(\{m_i,s_i^2,b_{\tau_i}\}_{i=1}^I,b_\sigma,b_\nu)$:

\begin{align*}
m_i &= \left(\frac{a_\sigma n}{b_\sigma \gamma_\beta} +  \frac{a_\nu a_{\tau_i}}{b_\nu b_{\tau_i}}\right)^{-1} \frac{a_\sigma n}{b_\sigma \gamma_\beta} \bar y_i,\\
s_i^2 &= \left(\frac{a_\sigma n}{b_\sigma \gamma_\beta} + \frac{a_\nu a_{\tau_i}}{b_\nu b_{\tau_i}}\right)^{-1},\\
b_{\tau_i} &= \frac{\left( \gamma_\tau - \frac{\gamma_\beta}{2} \right) (a_{\tau_i} - 1) + (a_{\tau_i} - 1) \sqrt{\left( \gamma_\tau - \frac{\gamma_\beta}{2} \right)^2 + \frac{2 \lambda_\tau \gamma_\tau \gamma_\beta a_{\tau_i} a_\nu}{(a_{\tau_i}-1)  b_\nu} (m_i^2 + s_{i}^2)}}{2 \lambda_\tau \gamma_\tau},\\
b_\sigma &= \frac{ 
    a_\sigma  \left[ \sum_{i=1}^I \sum_{j=1}^n\left[(y_{ij}-m_i)^2 + s_i^2\right] \right]
    + 2 \gamma_\sigma a_\sigma b_{\pi,\sigma} 
}{ 
    I n + 2 \gamma_\sigma a_{\pi,\sigma} 
},\\
b_\nu &= \frac{\left( \gamma_\nu - \frac{I \gamma_\beta}{2} \right) (a_\nu - 1) + (a_\nu - 1) \sqrt{\left( \gamma_\nu - \frac{I \gamma_\beta}{2} \right)^2 + \frac{2 \lambda_\nu \gamma_\nu \gamma_\beta a_\nu \sum_{i=1}^I \frac{a_{\tau_i}}{b_{\tau_i}} (m_i^2 + s_{i}^2)}{(a_\nu-1)}}}{2 \lambda_\nu \gamma_\nu}.
\end{align*}
The updates for the remaining parameters $(\{a_{\tau_i}\}_{i=1}^I,a_\sigma,a_\nu)$ are obtained by solving the equations, respectively:
\begin{align*}
0 &= - \frac{\gamma_\beta}{2}\left[\frac{ a_\nu}{b_{\tau_i} b_\nu} (m_i^2 + s_{i}^2) - \psi'(a_{\tau_i})\right] - \gamma_\tau\left[\psi'(a_{\tau_i})(a_{\tau_i}+1) - 1 - \frac{\lambda_\tau b_{\tau_i}}{(a_{\tau_i}-1)^2}\right],\\
0 &= \frac{I n}{2} \psi'(a_\sigma) - \frac{1}{2 b_\sigma} \left[ \sum_{i=1}^I \sum_{j=1}^n\left[(y_{ij}-m_i)^2 + s_i^2\right] \right] - \gamma_\sigma \left[ - \psi'(a_\sigma)(a_{\pi,\sigma} - a_\sigma) + \frac{b_{\pi,\sigma}}{b_\sigma} -1 \right],\\
0 &= -\frac{\gamma_\beta}{2} \left[ \sum_{i=1}^I \frac{a_{\tau_i}}{b_{\tau_i} b_\nu}(m_i^2 + s_{i}^2) - I \psi'(a_\nu) \right] - \gamma_\nu \left[ \psi'(a_\nu)(a_\nu+1) -1 - \frac{\lambda_\nu b_\nu}{(a_\nu-1)^2} \right].
\end{align*}

\subsubsection{SoU-SGL Linear Regression}\label{app:sousgl_linreg}

For the Linear Regression model in Section~\ref{app:linreg_setup}, the SoU variational objective can be written as

\begin{align*}
J =
& \frac{n}{2} [\log(2\pi) + \log b_\sigma - \psi(a_\sigma)] + \frac{a_\sigma}{2 b_\sigma}\left[\|y\|^2 - 2m_\alpha\sum_{i=1}^n y_i - 2\sum_{i=1}^n y_i\mathbf{X}_i^\top m + n m_\alpha^2 \right. \\
&\quad \left. + 2m_\alpha\sum_{i=1}^n\mathbf{X}_i^\top m + \sum_{i=1}^n (\mathbf{X}_i^\top m)^2 + \text{tr}(\mathbf{X}^\top\mathbf{X}S) + n s_\alpha^2\right] \\
&+\frac{\gamma_\alpha}{2}\left[\frac{m_\alpha^2 + s_\alpha^2}{\sigma_\alpha^2} - \log\left(\frac{s_\alpha^2}{\sigma_\alpha^2}\right)  - 1\right]\\
& + \frac{\gamma_{\beta}}{2} \left[ \sum_{j=1}^p S_{jj} \frac{a_\sigma}{b_\sigma} \frac{a_{\tau_j}}{b_{\tau_j}} \frac{a_\nu}{b_\nu} + \sum_{j=1}^p m_j^2 \frac{a_\sigma}{b_\sigma} \frac{a_{\tau_j}}{b_{\tau_j}} \frac{a_\nu}{b_\nu} \right. \\
& \quad \left. - \log|S| + p (\log b_\sigma - \psi(a_\sigma)) + \sum_{j=1}^p \left(\log b_{\tau_j} - \psi(a_{\tau_j})\right) + p (\log b_\nu - \psi(a_\nu)) - p \right] \\
& + \gamma_{\tau} \sum_{j=1}^p \left[ a_{\tau_j}\log b_{\tau_j} - \log\Gamma(a_{\tau_j}) - (a_{\tau_j}+1)(\log b_{\tau_j} - \psi(a_{\tau_j})) - a_{\tau_j} - \log\lambda_\tau + \frac{\lambda_\tau b_{\tau_j}}{a_{\tau_j}-1} \right] \\
& + \gamma_\nu \left[ a_\nu \log b_\nu - \log\Gamma(a_\nu) - (a_\nu+1)(\log b_\nu - \psi(a_\nu)) - a_\nu - \log\lambda_\nu + \frac{\lambda_\nu b_\nu}{a_\nu-1} \right] \\
& + \gamma_\sigma \left[ a_\sigma \log b_\sigma - a_{\pi,\sigma} \log b_{\pi,\sigma} + \log \left( \frac{\Gamma(a_{\pi,\sigma})}{\Gamma(a_\sigma)} \right) \right. \\
& \quad \left. + (\log b_\sigma - \psi(a_\sigma))(a_{\pi,\sigma} - a_\sigma) + \left(\frac{a_\sigma}{b_\sigma}\right)(b_{\pi,\sigma} - b_\sigma) \right].
\end{align*}

Let $D := \text{diag}\left(\frac{a_{\tau_j}}{b_{\tau_j}}\right)$ and \(\text{SSR} := \|y\|^2 - 2m_\alpha\mathbf{1}_n^\top y - 2y^\top\mathbf{X}m + n m_\alpha^2 + 2m_\alpha\mathbf{1}_n^\top\mathbf{X}m + m^\top\mathbf{X}^\top\mathbf{X}m + \text{tr}(\mathbf{X}^\top\mathbf{X}S) + n s_\alpha^2\). The first-order conditions yield the
following explicit updates for $(m_\alpha,s^2_\alpha,m,S,\{b_{\tau_j}\}_{j=1}^p,b_\sigma,b_\nu$):

\begin{align*}
m_\alpha &= \frac{\frac{a_\sigma}{b_\sigma}\sum_{i=1}^n (y_i - \mathbf{X}_i^\top m)}{\frac{n a_\sigma}{b_\sigma} + \frac{\gamma_\alpha}{\sigma_\alpha^2}},\\
s_\alpha^2 &= \frac{1}{\frac{n a_\sigma}{\gamma_\alpha b_\sigma} + \frac{1}{\sigma_\alpha^2}},\\
m &= \left(\mathbf{X}^\top\mathbf{X} + \gamma_\beta \frac{a_\nu}{b_\nu}D\right)^{-1}\mathbf{X}^\top(y - \mathbf{1}_n m_\alpha),\\
S &= \frac{\gamma_\beta b_\sigma}{a_\sigma}\left(\mathbf{X}^\top\mathbf{X} + \gamma_\beta \frac{a_\nu}{b_\nu}D\right)^{-1},\\
b_{\tau_j} &= \frac{\left( \gamma_\tau - \frac{\gamma_\beta}{2} \right) (a_{\tau_j} - 1) + (a_{\tau_j} - 1) \sqrt{\left( \gamma_\tau - \frac{\gamma_\beta}{2} \right)^2 + \frac{2 \lambda_\tau \gamma_\tau \gamma_\beta a_\sigma a_{\tau_j} a_\nu}{(a_{\tau_j}-1) b_\sigma b_\nu} (S_{jj} + m_j^2)}}{2 \lambda_\tau \gamma_\tau},\\
b_\sigma&= \frac{a_\sigma\left[\text{SSR} + \gamma_\beta\frac{a_\nu}{b_\nu}\sum_{j=1}^p\frac{a_{\tau_j}}{b_{\tau_j}}(S_{jj}+m_j^2)\right] + 2\gamma_\sigma a_\sigma b_{\pi,\sigma}}{n + \gamma_\beta p + 2\gamma_\sigma a_{\pi,\sigma}},\\
b_\nu &= \frac{\left( \gamma_\nu - \frac{p \gamma_\beta}{2} \right) (a_\nu - 1) + (a_\nu - 1) \sqrt{\left( \gamma_\nu - \frac{p \gamma_\beta}{2} \right)^2 + \frac{2 \lambda_\nu \gamma_\nu \gamma_\beta a_\sigma a_\nu \sum_{j=1}^p \frac{a_{\tau_j}}{b_{\tau_j}} (S_{jj} + m_j^2)}{(a_\nu-1) b_\sigma}}}{2 \lambda_\nu \gamma_\nu}.
\end{align*}

The updates for the remaining parameters $(\{a_{\tau_j}\}_{j=1}^p,a_\sigma,a_\nu)$ are obtained by solving the equations, respectively:

\begin{align*}
0 &= -\frac{\gamma_\beta}{2} \left[ \frac{a_\sigma a_\nu}{b_\sigma b_{\tau_j} b_\nu}(S_{jj}+m_j^2) - \psi'(a_{\tau_j}) \right] - \gamma_\tau \left[ \psi'(a_{\tau_j})(a_{\tau_j}+1) -1 - \frac{\lambda_\tau b_{\tau_j}}{(a_{\tau_j}-1)^2} \right],\\
0&= \frac{n}{2}\psi'(a_\sigma) - \frac{1}{2b_\sigma}\left[\text{SSR} + \gamma_\beta\frac{a_\nu}{b_\nu}\sum_{j=1}^p\frac{a_{\tau_j}}{b_{\tau_j}}(S_{jj}+m_j^2)\right] \\
&\quad + \frac{\gamma_\beta p}{2}\psi'(a_\sigma) - \gamma_\sigma\left[- \psi'(a_\sigma)(a_{\pi,\sigma}-a_\sigma) + \frac{b_{\pi,\sigma}}{b_\sigma} - 1\right],\\
0 &= -\frac{\gamma_\beta}{2} \left[ \sum_{j=1}^p \frac{a_\sigma a_{\tau_j}}{b_\sigma b_{\tau_j} b_\nu}(S_{jj}+m_j^2) - p \psi'(a_\nu) \right] - \gamma_\nu \left[ \psi'(a_\nu)(a_\nu+1) -1 - \frac{\lambda_\nu b_\nu}{(a_\nu-1)^2} \right].
\end{align*}

\subsubsection{SoU-SGL Logistic Regression}\label{app:sousgl_logreg}

For the Logistic Regression model in Section~\ref{app:logreg_setup}, the SoU variational objective can be written as

\begin{align*}
J = & -\sum_{i=1}^n \bigg[ (y_i - \tfrac{1}{2})(m_\alpha + \mathbf{X}_i^\top m) - \lambda(\xi_i)\big((m_\alpha + \mathbf{X}_i^\top m)^2 + s_\alpha^2 + \mathbf{X}_i^\top S \mathbf{X}_i\big) \\
&- \tfrac{\xi_i}{2} + \lambda(\xi_i)\xi_i^2 + \log\sigma(\xi_i) 
\bigg] \\
& +\frac{\gamma_\alpha}{2}\left[\frac{m_\alpha^2 + s_\alpha^2}{\sigma_\alpha^2} - \log\left(\tfrac{s_\alpha^2}{\sigma_\alpha^2}\right) - 1\right] \\
& + \frac{\gamma_\beta}{2} \bigg[ \sum_{j=1}^p (S_{jj} + m_j^2) \frac{a_{\tau_j}}{b_{\tau_j}}\frac{a_\nu}{b_\nu} - \log|S| \\
&\quad + \sum_{j=1}^p \big(\log b_{\tau_j} - \psi(a_{\tau_j})\big) + p\big(\log b_\nu - \psi(a_\nu)\big) - p \bigg] \\
& + \gamma_\tau \sum_{j=1}^p \bigg[ a_{\tau_j}\log b_{\tau_j} - \log\Gamma(a_{\tau_j}) - (a_{\tau_j}+1)\big(\log b_{\tau_j} - \psi(a_{\tau_j})\big) \\
&\quad - a_{\tau_j} - \log\lambda_\tau + \tfrac{\lambda_\tau b_{\tau_j}}{a_{\tau_j}-1} \bigg] \\
& + \gamma_\nu \bigg[ a_\nu \log b_\nu - \log\Gamma(a_\nu) - (a_\nu+1)\big(\log b_\nu - \psi(a_\nu)\big) \\
&\quad - a_\nu - \log\lambda_\nu + \tfrac{\lambda_\nu b_\nu}{a_\nu-1} \bigg],
\end{align*}
where \(\lambda(\xi_i) := \tfrac{1}{2\xi_i}\left[\sigma(\xi_i) - \tfrac{1}{2}\right]\) and \(\xi_i^2 := (m_\alpha + \mathbf{X}_i^\top m)^2 + s_\alpha^2 + \mathbf{X}_i^\top S \mathbf{X}_i\). The auxiliary parameters $\{\xi_i\}_{i=1}^n$ are updated at each iteration.

Let $D := \mathrm{diag}\!\left(\frac{a_\nu}{b_\nu}\frac{a_{\tau_1}}{b_{\tau_1}},\dots,
\frac{a_\nu}{b_\nu}\frac{a_{\tau_p}}{b_{\tau_p}}\right)$ and $\Lambda := \text{diag}(\lambda(\xi_1), \dots, \lambda(\xi_n))$. The first-order conditions yield the
following explicit updates for $(m_\alpha,s^2_\alpha,m,S,\{b_{\tau_j}\}_{j=1}^p,b_\nu$):

\begin{align*}
m_\alpha &= \frac{\sum_{i=1}^n \left[ (y_i - \tfrac{1}{2}) - 2\lambda(\xi_i)\mathbf{X}_i^\top m \right]}{2\sum_{i=1}^n \lambda(\xi_i) + \frac{\gamma_\alpha}{\sigma_\alpha^2}}, \\
s_\alpha^2 &= \frac{\gamma_\alpha}{2\sum_{i=1}^n \lambda(\xi_i) + \frac{\gamma_\alpha}{\sigma_\alpha^2}},\\
m &= \left(2\mathbf{X}^\top \Lambda \mathbf{X} + \gamma_\beta D\right)^{-1} \left( \mathbf{X}^\top\left(y - \tfrac{1}{2}\mathbf{1}_n\right) - 2m_\alpha \mathbf{X}^\top \Lambda \mathbf{1}_n \right),\\
S &= \left( \frac{2}{\gamma_\beta} \mathbf{X}^\top \Lambda \mathbf{X} + D \right)^{-1},\\
b_{\tau_j} &= \frac{\left( \gamma_\tau - \frac{\gamma_\beta}{2} \right) (a_{\tau_j} - 1) + (a_{\tau_j} - 1) \sqrt{\left( \gamma_\tau - \frac{\gamma_\beta}{2} \right)^2 + \frac{2 \lambda_\tau \gamma_\tau \gamma_\beta  a_{\tau_j} a_\nu}{(a_{\tau_j}-1)  b_\nu} (S_{jj} + m_j^2)}}{2 \lambda_\tau \gamma_\tau},\\
b_\nu &= \frac{\left( \gamma_\nu - \frac{p \gamma_\beta}{2} \right) (a_\nu - 1) + (a_\nu - 1) \sqrt{\left( \gamma_\nu - \frac{p \gamma_\beta}{2} \right)^2 + \frac{2 \lambda_\nu \gamma_\nu \gamma_\beta a_\nu  \sum_{j=1}^p \frac{a_{\tau_j}}{b_{\tau_j}} (S_{jj} + m_j^2)}{(a_\nu-1) }}}{2 \lambda_\nu \gamma_\nu}.
\end{align*}

The updates for the remaining parameters $(\{a_{\tau_j}\}_{j=1}^p,a_\nu)$ are obtained by solving the equations, respectively:

\begin{align*}
0 &= -\frac{\gamma_\beta}{2} \left[ \frac{ a_\nu}{b_{\tau_j} b_\nu}(S_{jj}+m_j^2) - \psi'(a_{\tau_j}) \right] - \gamma_\tau \left[ \psi'(a_{\tau_j})(a_{\tau_j}+1) -1 - \frac{\lambda_\tau b_{\tau_j}}{(a_{\tau_j}-1)^2} \right],\\
0 &= -\frac{\gamma_\beta}{2} \left[ \sum_{j=1}^p \frac{ a_{\tau_j}}{ b_{\tau_j} b_\nu}(S_{jj}+m_j^2) - p \psi'(a_\nu) \right] - \gamma_\nu \left[ \psi'(a_\nu)(a_\nu+1) -1 - \frac{\lambda_\nu b_\nu}{(a_\nu-1)^2} \right].
\end{align*}

\section{Proofs}\label{app:proofs}

\subsection{Closed-form characterisation}\label{app:closed_form}

This section proves Theorem~\ref{thm:Closed_Form} from the main text and establishes Corollary~\ref{cor:closed_form}.
We consider the unscaled version of the objective function without requiring the additivity of the empirical loss $L(\btheta)$:
\begin{multline} \label{eq:KB_unscaled}
J(q)
:= \mathbb{E}_{q} \left[ L(\btheta) \right] + \gamma_1 \, \text{KL}\left( 
q(\btheta_{1}) \,\|\, \pi(\btheta_{1}) 
\right)
+ \sum_{k=2}^{K} \gamma_k \, 
\mathbb{E}_{q(\btheta_{<k})} \left[ 
\text{KL}\left( 
q(\btheta_{k} \mid \btheta_{<k}) \,\|\, 
\pi(\btheta_{k} \mid \btheta_{<k}) 
\right)
\right].
\end{multline}

We start with the following observation regarding the form of the
solution $q^{*}$. 
\begin{lemma}
\label{lem:Ionescu}For any probability measure $q$ on $\Theta=\prod_{k=1}^{K}\Theta_{B_{k}}$
equipped with the standard Borel $\sigma$-algebra, there exists a
marginal measure $q_{1}\left(\mathrm{d}\btheta_{1}\right)$ and
Markov kernels $q_{k}\left(\mathrm{d}\btheta_{k}|\btheta_{<k}\right)$,
for each $k\in\left\{ 2,\dots,K\right\} $ such that 
\[
q\left(\mathrm{d}\btheta\right)=q_{1}\left(\mathrm{d}\btheta_{1}\right)\prod_{k=2}^{K}q_{k}\left(\mathrm{d}\btheta_{k}|\btheta_{<k}\right)\text{,}
\]
where the kernels are $q$ almost surely unique. 
\end{lemma}

\begin{proof}
This follows from iterative disintegration (existence of regular conditional
probabilities) and the Ionescu--Tulcea theorem. 
\end{proof}
Next, we give a conditional Donsker--Varadhan formula which is the
main engine of our closed form derivation. 
\begin{lemma}
\label{lem:Engine}
Suppose that we have a Borel space \(\mathbb{Y}\), a reference Markov kernel
\(\pi\left(\mathrm{d}y|X\right)\), a measurable function \(f\left(X,y\right)\) such that \(f\left(X,\cdot\right)\)
is bounded below for each fixed \(X\), and \(\gamma>0\).
For each fixed $X$,
define 
\[
Z\left(X\right)=\int\exp\left\{ -\frac{f\left(X,y\right)}{\bgamma}\right\} \pi\left(\mathrm{d}y|X\right)
\]
and $V\left(X\right)=-\gamma\log Z\left(X\right)$, assuming that
$Z\left(X\right)\in\left(0,\infty\right)$. Write 
\[
q^{*}\left(\mathrm{d}y|X\right)=\frac{1}{Z\left(X\right)}\exp\left\{ -\frac{f\left(X,y\right)}{\bgamma}\right\} \pi\left(\mathrm{d}y|X\right)\text{.}
\]
Then, for every $q\left(\mathrm{d}y|X\right)$ with $q\left(\cdot|X\right)\ll\pi\left(\cdot|X\right)$,
\[
\int f\left(X,y\right)q\left(\mathrm{d}y|X\right)+\gamma\mathrm{KL}\left(q\left(\cdot|X\right)\Vert\pi\left(\cdot|X\right)\right)=V\left(X\right)+\gamma\mathrm{KL}\left(q\left(\cdot|X\right)\Vert q^{*}\left(\cdot|X\right)\right)\text{.}
\]
In particular, the expression is minimised by taking $q\left(\cdot|X\right)=q^{*}\left(\cdot|X\right)$,
and the minimum is $V\left(X\right)$. 
\end{lemma}

\begin{proof}
Fix $X$ and write $q=q\left(\cdot|X\right)$, $\pi=\pi\left(\cdot|X\right)$,
$f=f\left(X,\cdot\right)$, $Z=Z\left(X\right)$, and $q^{*}=q^{*}\left(\cdot|X\right)$,
for convenience.

By definition of $q^{*}$, we have the Radon--Nikodym identity

\[
\log\frac{\mathrm{d}q^{*}}{\mathrm{d}\pi}\left(y\right)=-\frac{f\left(y\right)}{\bgamma}-\log Z\text{.}
\]
Thus, for any $q\ll\pi$, 
\[
\mathrm{KL}\left(q\Vert q^{*}\right)=\int\log\frac{\mathrm{d}q}{\mathrm{d}q^{*}}\mathrm{d}q=\int\log\frac{\mathrm{d}q}{\mathrm{d}\pi}\mathrm{d}q-\int\log\frac{\mathrm{d}q^{*}}{\mathrm{d}\pi}\mathrm{d}q=\mathrm{KL}\left(q\Vert\pi\right)+\frac{1}{\bgamma}\int f\mathrm{d}q+\log Z\text{.}
\]
Rearranging this gives the desired expression: 
\[
\int f\mathrm{d}q+\gamma\mathrm{KL}\left(q\Vert\pi\right)=-\gamma\log Z+\gamma\mathrm{KL}\left(q\Vert q^{*}\right)\text{.}
\]
Since the KL divergence is nonnegative and equal to zero iff $q=q^{*}$
almost everywhere, we have $q^{*}$ as the unique minimiser and the
minimum is $V\left(X\right)=-\gamma\log Z\left(X\right)$. 
\end{proof}
We are now ready to derive the closed form for the SoU posterior.

\begin{proof}[Proof of Theorem~\ref{thm:Closed_Form} (Closed-form characterisation)]
Let $q$ be a probability measure on $\Theta$. Then, by Lemma \ref{lem:Ionescu},
it has the form 
\[
q\left(\mathrm{d}\btheta\right)=q_{1}\left(\mathrm{d}\btheta_{1}\right)\prod_{k=2}^{K}q_{k}\left(\mathrm{d}\btheta_{k}|\btheta_{<k}\right)\text{.}
\]
Using the tower property, we have 
\[
\mathrm{E}_{q}\left[L\left(\btheta\right)\right]=\mathrm{E}_{q\left(\btheta_{<K}\right)}\left[\mathrm{E}_{q_{K}\left(\cdot|\btheta_{<K}\right)}V_{K}\left(\btheta_{\le K}\right)\right]\text{,}
\]
since $V_{K}=L$, by definition. Further, the only KL term involving
$q_{K}$ is 
\[
\gamma_{K}\mathrm{E}_{q\left(\btheta_{<K}\right)}\mathrm{KL}\left(q_{K}\left(\cdot|\btheta_{<K}\right)\Vert\pi_{K}\left(\cdot|\btheta_{<K}\right)\right)\text{,}
\]
hence $J\left(q\right)$ can be written as 
\[
J\left(q\right)=J_{K-1}\left(q_{1},\dots,q_{K-1}\right)+\mathrm{E}_{q\left(\btheta_{<K}\right)}\Phi_{K}\left(\btheta_{<K};q_{K}\left(\cdot|\btheta_{<K}\right)\right)\text{,}
\]
for functions $J_{K-1}$, and 
\[
\Phi_{K}\left(\btheta_{<K};q_{K}\left(\cdot|\btheta_{<K}\right)\right)=\mathrm{E}_{q_{K}\left(\cdot|\btheta_{<K}\right)}V_{K}\left(\btheta_{\le K}\right)+\gamma_{K}\mathrm{KL}\left(q_{K}\left(\cdot|\btheta_{<K}\right)\Vert\pi_{K}\left(\cdot|\btheta_{<K}\right)\right)\text{.}
\]

Next, fix $\btheta_{<K}$. Applying Lemma \ref{lem:Engine} with $X=\btheta_{<K}$
and $y=\btheta_{K}$, $\pi\left(\cdot|X\right)=\pi_{K}\left(\cdot|\btheta_{<K}\right)$,
$f\left(X,y\right)=V_{K}\left(\btheta_{\le K}\right)$, and $\gamma=\gamma_{K}$.
The lemma implies that 
\[
\Phi_{K}\left(\btheta_{<K};q_{K}\left(\cdot|\btheta_{<K}\right)\right)\ge V_{K-1}\left(\btheta_{<K}\right)\text{,}
\]
with equality iff $q_{K}\left(\cdot|\btheta_{<K}\right)=q_{K}^{*}\left(\cdot|\btheta_{<K}\right)$,
where $q_{K}^{*}$ is exactly the expression described in the hypothesis
of the theorem. Thus, for fixed $q_{1},\dots,q_{K-1}$, the unique
choice for $q_{K}$ that minimises $J\left(q\right)$ is $q_{K}^{*}$,
and the resulting problem reduces to the objective 
\[
\inf_{q_{K}}J\left(q\right)=J_{K-1}\left(q_{1},\dots,q_{K-1}\right)+\mathrm{E}_{q\left(\btheta_{<K}\right)}V_{K-1}\left(\btheta_{<K}\right)\text{.}
\]

We now repeat the argument above for $k=K-1,\dots,2$, to obtain the
expressions of $q_{k}^{*}$ for each $k\ge2$, and reduce the problem
to the final one-block problem: 
\[
\inf_{q_{1}}\left\{ \mathrm{E}_{q_{1}}V_{1}\left(\btheta_{1}\right)+\gamma_{1}\mathrm{KL}\left(q_{1}\Vert\pi_{1}\right)\right\} \text{.}
\]
We apply Lemma \ref{lem:Engine} a final time to get the expression
for $q_{1}^{*}$. By backward induction, we have that $q^{*}$ minimises
$J\left(q\right)$ and is indeed the unique minimiser, since each
iterative solution is unique by Lemma \ref{lem:Engine}. 
\end{proof}
\begin{corollary}\label{cor:closed_form}
Assume the hypothesis of Theorem~\ref{thm:Closed_Form}. Further assume
that $\pi\left(\mathrm{d}\btheta\right)=\bigotimes_{k=1}^{K}\pi_{k}\left(\mathrm{d}\btheta_{k}\right)$,
and $L\left(\btheta\right)=\sum_{k=1}^{K}L_{k}\left(\btheta_{k}\right)$
for measurable $L_{k}:\Theta_{B_{k}}\to\overline{\mathbb{R}}$. Then,
the unique SoU minimiser is $q^{*}\left(\mathrm{d}\btheta\right)=\bigotimes_{k=1}^{K}q_{k}^{*}\left(\mathrm{d}\btheta_{k}\right)$,
where 
\[
q_{k}^{*}\left(\mathrm{d}\btheta_{k}\right)=\frac{1}{Z_{k}}\exp\left\{ -\frac{1}{\gamma_{k}}L_{k}\left(\btheta_{k}\right)\right\} \pi_{k}\left(\mathrm{d}\btheta_{k}\right)\text{.}
\]
\end{corollary}

\begin{proof}
By Lemma \ref{lem:Ionescu}, we again have the disintegration: 
\[
q\left(\mathrm{d}\btheta\right)=q_{1}\left(\mathrm{d}\btheta_{1}\right)\prod_{k=2}^{K}q_{k}\left(\mathrm{d}\btheta_{k}|\btheta_{<k}\right)\text{,}
\]
where we will write $q_{k}=q\left(\btheta_{k}\right)$.
Firstly, we can decouple the losses in the sense that 
\[
\mathrm{E}_{q}L\left(\btheta\right)=\sum_{k=1}^{K}\mathrm{E}_{q_{k}}L_{k}\left(\btheta_{k}\right)\text{.}
\]
Then, by the tower property, for each $k\ge2$, since $\pi_{k}$ does
not depend on $\btheta_{<k}$ and KL is convex in its first argument,
we have 
\[
\mathrm{E}_{q\left(\btheta_{<k}\right)}\mathrm{KL}\left(q\left(\btheta_{k}|\btheta_{<k}\right)\Vert\pi_{k}\right)\ge\mathrm{KL}\left(\mathrm{E}_{q\left(\btheta_{<k}\right)}q\left(\btheta_{k}|\btheta_{<k}\right)\Vert\pi_{k}\right)=\mathrm{KL}\left(q_{k}\Vert\pi_{k}\right)\text{,}
\]
with equality holding when $q\left(\btheta_{k}|\btheta_{<k}\right)=q_{k}$,
$q$-a.s.. Combining the previous observations, we have 
\[
J\left(q\right)\ge\sum_{k=1}^{K}\left\{ \mathrm{E}_{q_{k}}L_{k}+\gamma_{k}\mathrm{KL}\left(q_{k}\Vert\pi_{k}\right)\right\} \text{,}
\]
and by Donsker--Varadhan, we have that 
\[
\underset{q_{k}}{\arg\min}\text{ }\mathrm{E}_{q_{k}}L_{k}+\gamma_{k}\mathrm{KL}\left(q_{k}\Vert\pi_{k}\right)=q_{k}^{*}\left(\mathrm{d}\btheta_{k}\right)\propto\exp\left\{ -\frac{L_{k}}{\gamma_{k}}\right\} \pi_{k}\left(\mathrm{d}\btheta_{k}\right)\text{.}
\]
But all of the inequalities are equalities when we take $q\left(\btheta_{k}|\btheta_{<k}\right)=q_{k}$,
for each $k\ge2$, and we have the desired result. 
\end{proof}

\subsection{Objective collapse and regularisation irrelevance}\label{app:collapse}

This section proves Theorem~\ref{thm:collapse} from the main text.
We write the objective given a confidence
vector $\gamma=\left(\gamma_{1},\dots,\gamma_{K}\right)\in\left(0,\infty\right)^{K}$ as 
\begin{multline} \label{eq:J_gamma}
J_{\bgamma}\left(q\right):=
\mathrm{E}_{q}L_{n}\left(\btheta\right)+\gamma_{1}\mathrm{KL}\left(q\left(\btheta_{1}\right)\Vert\pi\left(\btheta_{1}\right)\right)
+ \sum_{k=2}^{K}\gamma_{k}\mathrm{E}_{q\left(\btheta_{<k}\right)}\left[\mathrm{KL}\left(q\left(\btheta_{k}|\btheta_{<k}\right)\Vert\pi\left(\btheta_{k}|\btheta_{<k}\right)\right)\right].
\end{multline}

Recall $\|\bgamma\|_\infty := \max_{1\le k\le K}\gamma_k$ and, for $\delta>0$,
$S_\delta := \{\btheta \in \Theta : L_n(\btheta) \le L_n^* + \delta\}$, where
$L_n^* := \inf_{\theta \in \Theta} L_n(\btheta)$ (Section~\ref{sec:point_uncertainty}).
Let $\mathrm{KL}_k(\cdot\Vert\pi)$ denote the $k$-th blockwise KL term in \eqref{eq:J_gamma}.
Throughout, assume that $\Theta$ is Polish and we write $\rightsquigarrow$ for weak convergence on $\Theta$.

We say that a minimiser $\hat{\btheta}_{n}$ of $L_{n}$ is well-separated if for some metric $d$
\[
\Delta\left(\eta\right)=\inf_{d(\btheta,\hat{\btheta}_{n})\ge\eta}L_{n}\left(\btheta\right)-L_{n}\big(\hat{\btheta}_{n}\big)>0 \quad\text{ for every }\eta>0 .
\]

For ease of exposition, we enumerate the four claims of Theorem~\ref{thm:collapse}:
\begin{enumerate}
\item $\mathrm{E}_{q_{m}}L_{n}\to L_{n}^{*}$, 
\item $q_{m}\left(\btheta:L_{n}\left(\btheta\right)\ge L_{n}^{*}+\epsilon\right)\to0$ for every $\epsilon >0$,
\item If $L_{n}$ has a unique well-separated minimiser $\hat{\btheta}_{n}$, then $q_{m}\rightsquigarrow\delta_{\hat{\btheta}_{n}}$, 
\item $\sum_{k=1}^{K}\gamma_{k}^{\left(m\right)}\mathrm{KL}_{k}\left(q_{m}\Vert\pi\right)\to0$.
\end{enumerate}

\begin{lemma}
\label{lem:KL_upper_bound}If $q\ll\pi$, then 
\[
\sum_{k=1}^{K}\gamma_{k}\mathrm{KL}_{k}\left(q\Vert\pi\right)\le\left\Vert \bgamma\right\Vert _{\infty}\mathrm{KL}\left(q\Vert\pi\right)\text{.}
\]
\end{lemma}

\begin{proof}
When $q\ll\pi$, we have 
\[
\mathrm{KL}\left(q\Vert\pi\right)=\mathrm{KL}\left(q\left(\btheta_{1}\right)\Vert\pi\left(\btheta_{1}\right)\right)+\sum_{k=2}^{K}\mathrm{E}_{q\left(\btheta_{<k}\right)}\left[\mathrm{KL}\left(q\left(\btheta_{k}|\btheta_{<k}\right)\Vert\pi\left(\btheta_{k}|\btheta_{<k}\right)\right)\right]\text{.}
\]
All terms are nonnegative and we obtain the inequality by multiplying
termwise by $\gamma_{k}\le\left\Vert \bgamma\right\Vert _{\infty}$. 
\end{proof}
\begin{lemma}
\label{lem:Near_competitor}Fix $\delta>0$ with $\pi\left(S_{\delta}\right)>0$.
Define $q^{\delta}=\pi\left(\cdot|S_{\delta}\right)$. Then, 
\[
\mathrm{E}_{q^{\delta}}L_{n}\le L_{n}^{*}+\delta\text{, and }\mathrm{KL}\left(q^{\delta}\Vert\pi\right)=\log\left(\frac{1}{\pi\left(S_{\delta}\right)}\right)<\infty\text{.}
\]
Thus, by Lemma \ref{lem:KL_upper_bound}, we have 
\[
\sum_{k=1}^{K}\gamma_{k}\mathrm{KL}_{k}\left(q^{\delta}\Vert\pi\right)\le\left\Vert \bgamma\right\Vert _{\infty}\log\left(\frac{1}{\pi\left(S_{\delta}\right)}\right)\text{.}
\]
\end{lemma}

\begin{proof}
The loss bound holds since $q^{\delta}$ is supported on $S_{\delta}$,
and the KL bound arises by observing that on $S_{\delta}$, $\mathrm{d}q^{\delta}/\mathrm{d}\pi=1/\pi\left(S_{\delta}\right)$.
Thus, 
\[
\mathrm{KL}\left(q^{\delta}\Vert\pi\right)=\int_{S_{\delta}}\log\left(\frac{1}{\pi\left(S_{\delta}\right)}\right)q^{\delta}\left(\mathrm{d}\btheta\right)=\log\left(\frac{1}{\pi\left(S_{\delta}\right)}\right)q^{\delta}\left(S_{\delta}\right)=\log\left(\frac{1}{\pi\left(S_{\delta}\right)}\right)\text{,}
\]
and is finite whenever $\pi\left(S_{\delta}\right)>0$. 
\end{proof}

\begin{proof}[Proof of Theorem~\ref{thm:collapse} (Objective collapse as $\|\bgamma\|_\infty\to0$)]

By optimality of $q_{\bgamma}^{*}$, we have 
\[
J_{\bgamma}\left(q_{\bgamma}^{*}\right)\le J_{\bgamma}\left(q^{\delta}\right)\le L_{n}^{*}+\delta+\left\Vert \bgamma\right\Vert _{\infty}\log\left(\frac{1}{\pi\left(S_{\delta}\right)}\right)\text{.}
\]
In particular, for $\left(\bgamma^{\left(m\right)}\right)_{m}$, 
\[
\limsup_{m\to\infty}J_{\bgamma^{\left(m\right)}}\left(q_{m}\right)\le L_{n}^{*}+\delta\text{.}
\]
But since $\delta>0$ is arbitrary, we have 
\[
\limsup_{m\to\infty}J_{\bgamma^{\left(m\right)}}\left(q_{m}\right)\le L_{n}^{*}\text{.}
\]
However, $J_{\bgamma}\left(q\right)\ge\mathrm{E}_{q}L_{n}\ge L_{n}^{*}$
for every $q$ and $\bgamma$, thus the $\liminf$ inequality also
holds, and thus 
\[
J_{\bgamma^{\left(m\right)}}\left(q_{m}\right)\to L_{n}^{*}\text{ and }\mathrm{E}_{q_{m}}L_{n}\to L_{n}^{*}\text{.}
\]
This proves Part 1.

Next, fix $\epsilon>0$ and define 
\[
A_{\epsilon}=\left\{ \theta:L_{n}\left(\btheta\right)\ge L_{n}^{*}+\epsilon\right\} \text{.}
\]
Then, 
\[
\mathrm{E}_{q_{m}}L_{n}\ge L_{n}^{*}+\epsilon q_{m}\left(A_{\epsilon}\right)\text{,}
\]
since on $A_{\epsilon}$ the loss is at least $L_{n}^{*}+\epsilon$
and it is at least $L_{n}^{*}$ otherwise. But since $\mathrm{E}_{q_{m}}L_{n}\to L_{n}^{*}$,
it must hold that $q_{m}\left(A_{\epsilon}\right)\to0$, yielding
Part 2.

Fix $\eta>0$ and write $C_{\eta}=\left\{ \theta:d\left(\btheta,\hat{\btheta}_{n}\right)\ge\eta\right\} $.
By definition of $\Delta\left(\eta\right)$, for all $\btheta\in C_{\eta}$,
we have 
\[
L_{n}\left(\btheta\right)\ge L_{n}\left(\hat{\btheta}_{n}\right)+\Delta\left(\eta\right)=L_{n}^{*}+\Delta\left(\eta\right)\text{,}
\]
thus $C_{\eta}\subset A_{\Delta\left(\eta\right)}$. But from Part
2, we have that $q_{m}\left(C_{\eta}\right)\le q_{m}\left(A_{\Delta\left(\eta\right)}\right)\to0$,
$q_{m}\left(C_{\eta}^{\mathrm{c}}\right)\to1$, implying that $q_{m}\rightsquigarrow\delta_{\hat{\btheta}_{n}}$,
proving Part 3.

Finally, write 
\[
\rho_{\bgamma}\left(q\right)=\sum_{k=1}^{K}\gamma_{k}\mathrm{KL}_{k}\left(q\Vert\pi\right)\ge0\text{.}
\]
Then, $\rho_{\bgamma^{\left(m\right)}}\left(q_{m}\right)=J_{\bgamma^{\left(m\right)}}\left(q_{m}\right)-\mathrm{E}_{q_{m}}L_{n}$.
By optimality of $q_{m}$ and since $\mathrm{E}_{q_{m}}L_{n}\ge L_{n}^{*}$,
we have 
\[
0\le\rho_{\bgamma^{\left(m\right)}}\left(q_{m}\right)\le J_{\bgamma^{\left(m\right)}}\left(q^{\delta}\right)-L_{n}^{*}\le\delta+\left\Vert \bgamma^{\left(m\right)}\right\Vert _{\infty}\log\left(\frac{1}{\pi\left(S_{\delta}\right)}\right)\text{.}
\]
Taking limits $m\to\infty$ and $\delta\downarrow0$ yields $\rho_{\bgamma^{\left(m\right)}}\left(q_{m}\right)\to0$,
which is Part 4, as required.
\end{proof}

\subsection{Consistency}\label{app:consistency}

This section proves Theorem~\ref{thm:consistency} from the main text.
Recall $R_n(\btheta)=n^{-1}\sum_{i=1}^{n}\ell(\btheta,X_i)$ and
$R(\btheta)=\mathrm{E}_{\mathrm{P}_0}\ell(\btheta,X)$.
Write $\|\bgamma_n\|_\infty := \max_{1\le k\le K}\gamma_{k,n}$.

Let $\Theta$ be Polish with metric $d$, and write $B(\btheta,r)$ for the open ball.
We say that $\btheta_0$ is a well-separated minimiser of $R$ if
$R(\btheta_0)=\inf_{\btheta\in\Theta}R(\btheta)$, $R$ is continuous at $\btheta_0$, and for every $\varepsilon>0$,
\[
\Delta(\varepsilon):=\inf_{d(\btheta,\btheta_0)\ge\varepsilon} \{R(\btheta)-R(\btheta_0)\}>0.
\]

Let $q_n^*$ denote any minimiser of the SoU objective (main text, Eq.~(\ref{eq:qstar})–(\ref{eq:KB}))
with prior $\pi_n$ and confidence vector $\gamma_n=(\gamma_{1,n},\dots,\gamma_{K,n})\in(0,\infty)^K$, i.e.,
\begin{align*}
q_n^* \in \argmin_q \Bigg\{&
\mathbb{E}_q R_n(\btheta)
+\frac{\gamma_{1,n}}{n}
\mathrm{KL}\bigl(q(\btheta_1)\Vert\pi_n(\btheta_1)\bigr)
\\
&+\sum_{k=2}^K \frac{\gamma_{k,n}}{n}
\mathbb{E}_{q(\btheta_{<k})}
\left[
\mathrm{KL}\bigl(
q(\btheta_k\mid\btheta_{<k})
\Vert
\pi_n(\btheta_k\mid\btheta_{<k})
\bigr)
\right]
\Bigg\}.
\end{align*}
We prove a slightly stronger statement by allowing $\pi=\pi_n$ and replacing compactness of $\Theta$ by the following condition, which is automatically satisfied when $\Theta$ is compact:
\begin{itemize}
\item[$(ii')$] There exist measurable $h:\Theta\to\mathbb{R}_{>0}$ and $b:\mathcal{X}\to\mathbb{R}$ such that
$h$ is continuous with compact level sets, $\mathrm{E}_{\mathrm{P}_{0}}|b(X)|<\infty$, and
$\ell(\btheta,X)\ge h(\btheta)-b(X)$ for all $(\btheta,X)$.
\end{itemize}

\begin{proof}[Proof of Theorem~\ref{thm:consistency} (Consistency)]

Let us write the unscaled objective 
\[
J_{n}\left(q\right)=n\mathrm{E}_{q}R_{n}\left(\btheta\right)+P_{n}\left(q\Vert\pi_{n}\right)\text{,}
\]
where 
\[
P_{n}\left(q\Vert\pi_{n}\right)=\gamma_{1,n}\mathrm{KL}\left(q\left(\btheta_{1}\right)\Vert\pi_{n}\left(\btheta_{1}\right)\right)+\sum_{k=2}^{K}\gamma_{k,n}\mathrm{E}_{q\left(\btheta_{<k}\right)}\left[\mathrm{KL}\left(q\left(\btheta_{k}|\btheta_{<k}\right)\Vert\pi_{n}\left(\btheta_{k}|\btheta_{<k}\right)\right)\right]\text{.}
\]
Note that by definition $q_{n}^{*}$ minimises $J_{n}$.

From Lemma \ref{lem:KL_upper_bound} we have that if $q\ll\pi_{n}$,
then 
\[
\mathrm{KL}\left(q\Vert\pi_{n}\right)=\mathrm{KL}\left(q\left(\btheta_{1}\right)\Vert\pi_{n}\left(\btheta_{1}\right)\right)+\sum_{k=2}^{K}\mathrm{E}_{q\left(\btheta_{<k}\right)}\left[\mathrm{KL}\left(q\left(\btheta_{k}|\btheta_{<k}\right)\Vert\pi_{n}\left(\btheta_{k}|\btheta_{<k}\right)\right)\right]
\]
and $P_{n}\left(q\Vert\pi_{n}\right)\le\left\Vert \bgamma_{n}\right\Vert _{\infty}\mathrm{KL}\left(q\Vert\pi_{n}\right)$.
From Lemma \ref{lem:Near_competitor} and $(iv)$, we have that if $q_{n}^{\left(0\right)}=\pi_{n}\left(\cdot|B\left(\btheta_{0},r_{n}\right)\right)$,
then 
\[
\mathrm{KL}\left(q_{n}^{\left(0\right)}\Vert\pi_{n}\right)=\log\left(\frac{1}{\pi_{n}\left(B\left(\btheta_{0},r_{n}\right)\right)}\right)\text{.}
\]

These results then imply that 
\[
P_{n}\left(q_{n}^{\left(0\right)}\Vert\pi_{n}\right)\le\left\Vert \bgamma_{n}\right\Vert _{\infty}\log\left(\frac{1}{\pi_{n}\left(B\left(\btheta_{0},r_{n}\right)\right)}\right)
\]
and since $q_{n}^{*}$ minimises $J_{n}$, we have 
\[
n\mathrm{E}_{q_{n}^{*}}R_{n}+P_{n}\left(q_{n}^{*}\Vert\pi_{n}\right)\le n\mathrm{E}_{q_{n}^{\left(0\right)}}R_{n}+P_{n}\left(q_{n}^{\left(0\right)}\Vert\pi_{n}\right)
\]
implying that 
\[
\mathrm{E}_{q_{n}^{*}}R_{n}\le\mathrm{E}_{q_{n}^{\left(0\right)}}R_{n}+\frac{\left\Vert \bgamma_{n}\right\Vert _{\infty}}{n}\log\left(\frac{1}{\pi_{n}\left(B\left(\btheta_{0},r_{n}\right)\right)}\right)\text{,}
\]
where the last term is $o\left(1\right)$ by $(iv)$.

Next, write 
\[
\mathrm{E}_{q_{n}^{\left(0\right)}}R_{n}-R\left(\btheta_{0}\right)=\mathrm{E}_{q_{n}^{\left(0\right)}}\left[R_{n}-R\right]+\left[\mathrm{E}_{q_{n}^{\left(0\right)}}R-R\left(\btheta_{0}\right)\right]=I+II\text{.}
\]
By construction, $q_{n}^{\left(0\right)}$ is supported on $B\left(\btheta_{0},r_{n}\right)$.
Pick any $M\ge0$ so that $B\left(\btheta_{0},r_{n}\right)\subset C_{M}=\left\{ h\le M\right\} $,
for sufficiently large $n$ (this is possible since $r_{n}\downarrow0$
and $h$ is continuous at $\btheta_{0}$, where $h\left(\btheta_{0}\right)<\infty$).
Then, by $(i)$, 
\[
\left|I\right|\le\sup_{\btheta\in C_{M}}\left|R_{n}\left(\btheta\right)-R\left(\btheta\right)\right|\overset{\mathrm{a.s.}}{\longrightarrow}0\text{.}
\]
By $(iii)$, $\btheta_{0}$ is the unique minimiser and $R\left(\btheta_{0}\right)<\infty$.
Since $R$ is continuous at $\btheta_{0}$, 
\[
\sup_{\btheta\in B\left(\btheta_{0},r_{n}\right)}\left|R\left(\btheta\right)-R\left(\btheta_{0}\right)\right|\to0\text{,}
\]
thus $II\to0$. Thus, we have 
\[
\mathrm{E}_{q_{n}^{\left(0\right)}}R_{n}\overset{\mathrm{a.s.}}{\longrightarrow}R\left(\btheta_{0}\right)\text{,}
\]
and therefore 
\begin{equation}
\limsup_{n\to\infty}\mathrm{E}_{q_{n}^{*}}R_{n}\le\limsup_{n\to\infty}\mathrm{E}_{q_{n}^{\left(0\right)}}R_{n}+o\left(1\right)=R\left(\btheta_{0}\right)\text{, almost surely.}\label{eq:limsup_star_and_theta_0}
\end{equation}

Define $m_{n}=\inf_{\btheta\in\Theta}R_{n}\left(\btheta\right)$, where
trivially, $m_{n}\le R_{n}\left(\btheta_{0}\right)\to R\left(\btheta_{0}\right)$,
almost surely, thus $\limsup_{n}m_{n}\le R\left(\btheta_{0}\right)$.
To lower bound $m_{n}$, observe that for fixed $\eta>0$ and with
\[
\bar{b}_{n}=\frac{1}{n}\sum_{i=1}^{n}b\left(X_{i}\right)\text{,}
\]
we have $\bar{b}_{n}\to\mathrm{E}\left[b\left(X\right)\right]$ almost
surely, and for sufficiently large $n$, 
\[
\left\{ \theta:R_{n}\left(\btheta\right)\le R\left(\btheta_{0}\right)+\eta\right\} \subset\left\{ h\le M_{\eta}\right\} \text{,}
\]
where $M_{\eta}=R\left(\btheta_{0}\right)+\eta+\mathrm{E}\left[b\left(X\right)\right]+1$.
This is because $(ii')$ implies that $R_{n}\left(\btheta\right)\ge h\left(\btheta\right)-\bar{b}_{n}$.
Then, by the strong law of large numbers, for sufficiently large $n$,
$\bar{b}_{n}\le\mathrm{E}\left[b\left(X\right)\right]+1$ almost surely.
Further, if $R_{n}\left(\btheta\right)\le R\left(\btheta_{0}\right)+\eta$
then $h\left(\btheta\right)\le R\left(\btheta_{0}\right)+\eta+\bar{b}_{n}\le M_{\eta}$.
Note also that by $(ii')$, $\left\{ h\le M_{\eta}\right\} $ is compact.

Let us write $C_{\eta}=\left\{ h\le M_{\eta}\right\} $. Then for
sufficiently large $n$, we have $m_{n}=\inf_{\btheta\in C_{\eta}}R_{n}\left(\btheta\right)$.
Thus $(i)$ gives us uniform convergence on $C_{\eta}$, and since the
infimum is Lipschitz, we have 
\[
\left|m_{n}-\inf_{\btheta\in C_{\eta}}R\left(\btheta\right)\right|=\left|\inf_{\btheta\in C_{\eta}}R_{n}\left(\btheta\right)-\inf_{\btheta\in C_{\eta}}R\left(\btheta\right)\right|\le\sup_{\btheta\in C_{\eta}}\left|R_{n}\left(\btheta\right)-R\left(\btheta\right)\right|\overset{\mathrm{a.s.}}{\longrightarrow}0\text{.}
\]
Note that by definition $h\left(\btheta_{0}\right)\le R\left(\btheta_{0}\right)+\mathrm{E}\left[b\left(X\right)\right]\le M_{\eta}$,
thus $\btheta_{0}\in C_{\eta}$ and therefore $m_{n}\overset{\mathrm{a.s.}}{\longrightarrow}R\left(\btheta_{0}\right)$.
Now fix $\epsilon>0$ and define $A_{\epsilon}=\left\{ \theta:d\left(\btheta,\btheta_{0}\right)\ge\epsilon\right\} $
and let $m_{n,\epsilon}=\inf_{\btheta\in A_{\epsilon}}R_{n}\left(\btheta\right)$
and $m_{\epsilon}=\inf_{\btheta\in A_{\epsilon}}R\left(\btheta\right)$.
By $(iii)$, $m_{\epsilon}=R\left(\btheta_{0}\right)+\Delta\left(\epsilon\right)$,
where $\Delta\left(\epsilon\right)>0$. Using the same argument as
above (now on the level sets $\left\{ R_{n}\le m_{\epsilon}+1\right\} $),
we have that $m_{n,\epsilon}\overset{\mathrm{a.s.}}{\longrightarrow}m_{\epsilon}$,
and consequently 
\[
\Delta_{n}\left(\epsilon\right)=m_{n,\epsilon}-m_{n}\overset{\mathrm{a.s.}}{\longrightarrow}\Delta\left(\epsilon\right)>0\text{.}
\]

Together (\ref{eq:limsup_star_and_theta_0}), $m_{n}\overset{\mathrm{a.s.}}{\longrightarrow}R\left(\btheta_{0}\right)$
and the bound $m_{n}\le\mathrm{E}_{q_{n}^{*}}R_{n}$ imply that $\mathrm{E}_{q_{n}^{*}}R_{n}-m_{n}\overset{\mathrm{a.s.}}{\longrightarrow}0$.
Decomposing $\mathrm{E}_{q_{n}^{*}}R_{n}$ over $A_{\epsilon}$ and
$A_{\epsilon}^{\mathrm{c}}$ yields 
\[
\mathrm{E}_{q_{n}^{*}}R_{n}\ge q_{n}^{*}\left(A_{\epsilon}\right)m_{n,\epsilon}+\left\{ 1-q_{n}^{*}\left(A_{\epsilon}\right)\right\} m_{n}=m_{n}+q_{n}^{*}\left(A_{\epsilon}\right)\Delta_{n}\left(\epsilon\right)\text{.}
\]
Hence, 
\[
q_{n}^{*}\left(A_{\epsilon}\right)\le\frac{\mathrm{E}_{q_{n}^{*}}R_{n}-m_{n}}{\Delta_{n}\left(\epsilon\right)}\text{.}
\]
But the numerator tends to zero and the denominator is bounded away
from zero, almost surely, thus $q_{n}^{*}\left(A_{\epsilon}\right)\to0$,
almost surely, or equivalently $q_{n}^{*}\left(B\left(\btheta_{0},\epsilon\right)\right)\to1$,
almost surely. Since this holds for every $\epsilon>0$, we have the
consistency result that $q_{n}^{*}\rightsquigarrow\delta_{\btheta_{0}}$,
almost surely, as required. 
\end{proof}

\end{document}